\documentclass[11pt, oneside, UKenglish]{article}

\usepackage[UKenglish]{babel}

\usepackage{dsfont}
\usepackage{geometry}
\usepackage{graphics}
\usepackage{graphicx}
\usepackage{xcolor}
\usepackage{amsmath,amsfonts,amsthm,amssymb,framed,color,bbm}
\usepackage[utf8]{inputenc}

\usepackage{empheq}

\usepackage{pdfpages}

\usepackage{tikz}
\usetikzlibrary{matrix}

\usetikzlibrary{decorations.pathreplacing}

\theoremstyle{plain}
\newtheorem{thm}{Theorem}[section]
\newtheorem{prop}[thm]{Proposition}
\newtheorem{lemma}[thm]{Lemma}
\newtheorem{remark}[thm]{Remark}
\newtheorem{cor}[thm]{Corollary}

\newtheorem{defn}[thm]{Definition}

\theoremstyle{definition}

\newcommand{\Z}{\mathbb{Z}}
\newcommand{\N}{\mathbb{N}}

\newcommand{\R}{\mathbb{R}}

\newcommand {\T} {\mathbb T} 

\newcommand{\Cov}{\mathrm{Cov}}
\newcommand{\diam}{\mathrm{diam}}

\newcommand{\obs}{\mathrm{obs}}
\newcommand{\ext}{\mathrm{ext}}

\newcommand{\es}{\emptyset}

\newcommand{\1}{\mathbbm{1}}

\newcommand{\vertiii}[1]{{\left\vert\kern-0.25ex\left\vert\kern-0.25ex\left\vert #1 
    \right\vert\kern-0.25ex\right\vert\kern-0.25ex\right\vert}}

\newcommand{\de}{\mathrm{d}}								

\begin{document}

\hyphenation{co-va-ri-ance cor-re-la-tion nor-mali-sa-tion}

\title{Decay of covariance for gradient models with non-convex potential}
\author{Susanne Hilger\footnote{Email: shilger@posteo.de}}
\date{}
\maketitle

\begin{abstract}
We consider gradient models on the lattice $\Z^d$. These models serve as effective models for interfaces and are also known as \textit{continuous Ising models}. The height of the interface is modelled by a random field with an energy which is a non-convex perturbation of the quadratic interaction. We are interested in the Gibbs measure with tilted boundary condition $u$ at inverse temperature $\beta$ of this model.

In this paper we present a fine analysis of the covariance of the gradient field. We show that the covariances of the Gibbs distribution agree with the covariance of the Gaussian free field up to terms which decay at a faster algebraic rate. The key tool is the extension of the renormalisation group method to observables as developed in \cite{BBS_Crit2}.

\end{abstract}

\section{Introduction}

We analyse continuous Ising models which are effective models for random interfaces. Let $\Lambda \subset \Z^d$ be a finite subset of the lattice.  We consider fields $\varphi: \Lambda \to \R$ which can be interpreted as height variables of the interface. To each configuration $\varphi \in \R^{\Lambda}$ an energy $H_{\Lambda}(\varphi)$ is assigned This Hamiltonian is given by a potential $W:\R \to \R$ that only depends on discrete gradients of the field,
$$
H_{\Lambda}(\varphi) = \sum_{x \in \Lambda} \sum_{i = 1}^d W(\nabla_i \varphi (x)),
$$
where $\nabla_i \varphi (x) = \varphi(x+e_i) - \varphi(x)$ is the finite difference quotient on the lattice.
We impose tilted boundary conditions, namely
$$
\varphi(x) = \psi^u(x) \quad\mbox{for }x \in \partial\Lambda,
\quad \psi^u(x) = u \cdot x \mbox{ for } u \in \R^d.
$$
The finite-volume Gibbs measure with boundary condition $\psi^u$ at inverse temperature $\beta>0$ is then
$$
\gamma_{\beta,\Lambda}^{\psi^u} (\de \varphi) = \frac{1}{Z_{\beta,\Lambda}^{\psi^u}} \,\, e^{-\beta H_{\Lambda}(\varphi)} \prod_{x \in \Lambda} \de \varphi(x) \prod_{x \in \partial \Lambda} \delta_{\psi^u(x)}(\de \varphi (x)),
$$
where $Z_{\beta,\Lambda}^{\psi_u}$ is the partition function which normalizes the measure.

~\\
In the case of strictly convex, symmetric $W$ a lot is known about the behaviour of $\gamma_{\beta,\Lambda}^{\psi^u} (\de \varphi)$: The infinite-volume gradient Gibbs measure exists and is uniquely determined by the tilt, see \cite{FS97}. The long distance behaviour is described by the Gaussian free field (see \cite{NS97} and \cite{GOS01}) and the decay of the covariance is polynomial as in the massless Gaussian case (\cite{DD05}). Moreover the surface tension is strictly convex (DGI00).

~\\
The situation is not that clear for models with non-convex potentials.

A special class of gradient fields with non-convex potentials (log-mixture of centered Gaussians) is considered in \cite{BK07}. At tilt $u=0$, a phase transition is shown to happen at some critical value of inverse temperature $\beta_c$. This result demonstrates that one can expect neither the uniqueness of gradient Gibbs measures corresponding to a fixed tilt $u$ nor strict convexity of the surface tension. However, the scaling limit in this case is still the Gaussian free field, as shown in \cite{BS11}.

For a class of gradient models where the potential is a small non-convex perturbation of a strictly convex one, \cite{CDM09}shows strict convexity of the surface tension at high temperature. For the same class in the same temperature regime, in \cite{CD12} it is shown that for any $u$ there exists a unique ergodic, shift-invariant gradient Gibbs measure . Moreover, the measure scales to the Gaussian free field and the decay of the covariance is algebraic as above.

~\\
The complementary temperature regime is considered in \cite{AKM16}. The authors consider potentials which are small perturbations of the quadratic one, the perturbation chosen such that it does not disturb the convexity at the minimum of the potential. For small tilt $u$ and large inverse temperature $\beta$ they prove strict convexity of the surface tension obtained as a limit of a subsequence of $(N_l)_{l \in \N}$, where $L^N$ is the side length of the box $\Lambda$, and relying on a quite restrictive lower bound on $W$, namely
$$
W(s) \geq (1-\epsilon) s^2
$$
for a small $\epsilon$.

In the same setting the paper \cite{Hil16} shows that there is $q \in \R^{d \times d}_{\text{sym}}$ small, such that the scaling limit is the Gaussian free field on $\T^d$ with covariance $\mathcal{C}^q_{\T^d}$, where
$$
 \left( \mathcal{C}_{\T^d}^q \right)^{-1}
 = - \sum_{i,j=1}^d \left( \delta_{ij} + q_{ij} \right) \partial_i \partial_j,
$$
 and that a "smoothed" covariance decays algebraically. The convergences are on a subsequence.

~\\
In \cite{ABKM} the class of potentials is widened to such which satisfy less restrictive bounds on the potential, namely
$$
W(s) \geq \epsilon s^2,
$$
and to vector-valued fields and finite-range instead of only nearest-neighbour interaction. The last two improvements are of interest for the application in nonlinear elasticity. The authors show that the surface tension is strictly convex and that the scaling limit is the Gaussian free field on the torus. All convergences are still on a subsequence. This assumption is removed in \cite{Hil19_1}.

~\\
The setting in this paper is similar to the one from \cite{ABKM} and \cite{Hil19_1}: We restrict to small tilts and large inverse temperature and use the same smallness condition on the potential. For the sake of simplicity we formulate our results and proofs for scalar-valued fields and nearest-neighbour interaction. We show refined covariance estimates, namely
$$
\left| \Cov(\nabla_i \varphi(a), \nabla_j\varphi(b))
\right|
\leq C \frac{1}{|a-b|^d}.
$$
More precisely, it is shown that to first order in $|a-b|$ the Gaussian covariance $C^q_{\Z^d}$ appears, where $C^q_{\Z^d}$ is the kernel of $\mathcal{C}^q_{\Z^d}$ with $\left( \mathcal{C}^q_{\Z^d} \right)^{-1} = \sum_{i,j=1}^d (\delta_{ij} + q_{ij}) \nabla_j^* \nabla_i$:
$$
\Cov(\nabla_i \varphi(a), \nabla_j\varphi(b))
= \nabla_j^* \nabla_i C^q(a,b) + R_{ab},
\quad
\left|R_{ab}\right| \leq C \frac{1}{|a-b|^{d + \nu}},
\quad
\nu>0.
$$

~\\
The proof builds on a rigorous renormalisation group approach for the partition function as developed by Bauerschmidt, Brydges and Slade in a series of papers (\cite{BS1},\cite{BS2}, \cite{BS3}, \cite{BS4}, \cite{BS5}). This approach is developed for the model at hand in \cite{AKM16} and improved in \cite{ABKM} and \cite{Hil19_1}. We augment the technique in the following direction: The renormalisation group analysis is enlarged from the bulk flow (which determines the partition function) to observables. This allows us to prove fine estimates for the covariance.

\paragraph{Structure of the paper}

In Section \ref{sec:SettingResults}, gradient models are introduced and the main result on a fine estimate on the covariance (Theorem \ref{Thm:DecayCovariance}) is stated. Furthermore, a technical theorem on which the proof of the result is based is formulated (Theorem \ref{Thm:RepresentationExtendedPartitionFunction}). The technical theorem contains a representation of the generating partition function and provides a straightforward proof of the main result.

In Section \ref{sec:RG-analysis_BulkFlow} steps from the RG analysis for the bulk flow in \cite{Hil19_1} are outlined. They are needed for the extended proof in the next section.

Section \ref{sec:RG-analysis_ObservableFlow} is dedicated to the RG analysis for the observable flow and the proof of Theorem \ref{Thm:RepresentationExtendedPartitionFunction}.

In Section \ref{sec:Proofs_ExtendedFlow}, details for certain extensions and intermediate steps are provided. The presentation follows closely the one in \cite{ABKM} in order to facilitate the understanding of the extensions. Proofs are only provided if they differ from the ones in \cite{ABKM}.

\paragraph{Notations}
Throughout the whole paper we will use the following notations.
\begin{itemize}
	\item $C_c^{\infty}$ will denote the set of smooth, compactly supported functions.
	\item Partial derivatives will be denoted by $\partial_s$ instead of $\frac{\partial}{\partial s}$.
	\item The symbol $\partial_i$ will be used for usual derivatives, in contrast to $\nabla_i$ for discrete finite differences.
	\item $C^r$ denotes the set of $r$-times differential functions.
	\item $\R^{d \times d}_{{\text{sym}}}$ denotes the set of $d \times d$ symmetric matrices.
	\item The Kronecker-delta $\delta_{ij}$ is $1$ if $i=j$ and $0$ else.
	\item The indicator function $\1_z$ is given by $\1_z = 1$ if condition $z$ is satisfied and $\1_z = 0$ otherwise.
	\item We use the \textit{big O notation} $f(x) = \mathcal{O}(g(x))$ as $x \rightarrow \infty$ to describe the limiting behaviour of the function $f$ in terms of the function $g$. It means that for all sufficiently large values of $x$, the absolute value of $f(x)$ is at most a positive constant multiple of $g(x)$.
	\item For $x \in \R$ let $(x)_+$ be $x$ if $x \geq 0$ and $0$ else.
	\item For $x,y \in \R$ let $(x \wedge y)$ denote the minimum of $x$ and $y$.
	\item The symbol $C$ will mostly denote a positive constant whose value is allowed to change in a chain of inequalities from line to line.
\end{itemize}

\section{Setting and result} \label{sec:SettingResults}

We start by describing gradient models and their finite-volume Gibbs distributions and stating the main result, namely the decay of correlations in Theorem \ref{Thm:DecayCovariance}.

Then we state a technical key theorem (Theorem \ref{Thm:RepresentationExtendedPartitionFunction}), which is the main component of the proof of the main result. It contains a powerful representation of the normalisation constant of the Gibbs measure with observables. From this representation the proof of the main result can be deduced straightforwardly.

\subsection{Gradient models}\label{sec:GradientModels}

Fix an odd integer $L \geq 3$ and a dimension $d \geq 2$. Let $\T_N = \left( \Z / L^N \Z \right)^d$ be the $d$-dimensional discrete torus of side length $L^N$ where $N$ is a positive integer. We equip $\T_N$ with the quotient distances $| \cdot |$ and $| \cdot |_{\infty}$ induced by the Euclidean and maximum norm respectively. The torus can be represented by the cube
$$
\Lambda_N = \left\lbrace x \in \Z^d: |x|_{\infty} \leq \frac{1}{2} \left(L^N - 1 \right) \right\rbrace
$$
of side length $L^N$ once it is equipped with the metric
$$
|x-y|_{\mathrm{per}} = \inf \left\lbrace |x-y+k|_{\infty}: k \in \left( L^N \Z \right)^d \right\rbrace.
$$

 Define the space of fields on $\Lambda_N$ as 
$$
\mathcal{V}_N = \lbrace \varphi: \Lambda_N \rightarrow \R \rbrace = \R^{\Lambda_N}.
$$
Since we will consider shift invariant energies, we are only interested in gradient fields on $\mathcal{V}_N$. Gradient fields can be described by elements in $\mathcal{V}_N /_{\lbrace\text{constants}\rbrace}$, or, equivalently, by usual fields with vanishing average
$$
\chi_N = \bigg\lbrace \varphi \in \mathcal{V}_N: \sum_{x \in \Lambda_N} \varphi(x) = 0 \bigg\rbrace.
$$
We equip $\chi_N$ with a scalar product via
$$
(\varphi,\psi) = \sum_{x \in \Lambda_N} \varphi(x) \psi(x).
$$
Let $\lambda_N$ be the $\left(L^{Nd}-1\right)$-dimensional Hausdorff measure on $\chi_N$. Let $e_i$, $i= 1, \ldots, d$, be the standard unit vectors in $\Z^d$. Then the discrete forward and backward derivatives are defined by
\begin{align*}
\nabla_i \varphi(x) = \varphi (x + e_i) - \varphi(x), \quad i \in \lbrace 1, \ldots, d \rbrace,
\\
\nabla_i^* \varphi(x) = \varphi (x - e_i) - \varphi(x), \quad i \in \lbrace 1, \ldots, d \rbrace.
\end{align*}

Let $W: \R \rightarrow \R$ be a potential which is a perturbation of a quadratic potential,
$$
W(s) = \frac{1}{2} s^2 + V(s), \quad V: \R \rightarrow \R.
$$
We study a class of random gradient fields defined in terms of a Hamiltonian
\begin{align*}
H_N(\varphi)
= \sum_{x \in \Lambda_N} \sum_{i=1}^d W(\nabla_i \varphi(x))
= \sum_{x \in \Lambda_N} \sum_{i=1}^d \left( \frac{1}{2} |\nabla_i \varphi (x)|^2 + V(\nabla_i \varphi(x))\right).
\end{align*}

We equip the space $\chi_N$ with the $\sigma$-algebra $\mathfrak{B}_{\chi_N}$ induced by the Borel-$\sigma$-algebra with respect to the product topology, and use $\mathcal{M}_1(\chi_N) = \mathcal{M}_1(\chi_N, \mathfrak{B}_{\chi_N} )$ to denote the set of probability measures on $\chi_N$.

The finite-volume gradient Gibbs measure $\gamma_{N,\beta} \in \mathcal{M}_1(\chi_N)$ at inverse temperature $\beta$ is defined as
\begin{align*}
\gamma_{N,\beta}(\de \varphi) = \frac{1}{Z_{N,\beta}} e^{-\beta H_N(\varphi)} \lambda_N(\de \varphi)
\end{align*}
with partition function
\begin{align*}
Z_{N,\beta} = \int_{\chi_N} e^{-\beta H_N(\varphi)} \lambda_N(\de \varphi).
\end{align*}
The model describes the behaviour of a random microscopic interface. A microscopic tilt applied to the discrete interface can be implemented by the Funaki-Spohn trick introduced in \cite{FS97}. Given $u \in \R^d$, we define the Hamiltonian $H_N^u$ on the torus~$\T_N$ with tilt $u$ by
\begin{align*}
H_N^u(\varphi) = \sum_{x \in \Lambda_N} \sum_{i=1}^d W(\nabla_i \varphi(x) + u_i).
\end{align*}
Consequently, the finite-volume gradient Gibbs measure $\gamma_{N,\beta}^u$ with tilt $u$ is defined~as
\begin{align*}
\gamma_{N,\beta}^u(\de \varphi) = \frac{1}{Z_{N,\beta}(u)} e^{-\beta H_N^u(\varphi)} \lambda_N(\de\varphi),
\end{align*}
where $Z_{N,\beta}(u)$ is the normalisation constant. A useful generalisation of the partition function with a source term $f \in \mathcal{V}_N$ is given by the generating functional
\begin{align}
Z_{N,\beta}(u,f) = \int_{\chi_N} e^{-\beta H_N^u(\varphi)+(f,\varphi)} \lambda_N(\de \varphi).
\label{eq:gen_func}
\end{align}

\subsection{Main results}

We assert an asymptotic expression for the gradient-gradient covariance of the Gibbs measure.

~\\
We impose the following assumptions on the potential $W$:
\begin{align}
\begin{cases}
&\mbox{Let } r_0 \geq 3, r_1 \geq 2, \, V \in C^{r_0+r_1}, V'(0) = V''(0) = 0. \nonumber
\\
&\mbox{Let } 0 < \omega < \frac{1}{16}  \mbox{ and suppose that } \sum_{i=1}^d W(z_i) \geq \omega |z|^2 \mbox{ and}
\tag{$\star$}
\label{eq:AssumptionsW}
\\
&\lim_{t \rightarrow \infty}t^{-2} \ln \Psi (t) = 0
\\ & \quad\quad
 \mbox{ where }
\Psi(t) = \sup_{|z|\leq t} \sum_{3 \leq |\alpha| \leq r_0+r_1} \frac{1}{\alpha !} \lvert \partial^{\alpha} \sum_{i=1}^d W(z_i) \rvert.
\nonumber
\end{cases}
\end{align}

~\\

We give a formula for the gradient-gradient covariance. Given $a,b \in \Lambda_N$ and directions $m_a,m_b \in \lbrace 1, \ldots, d \rbrace$, define
\begin{align*}
&\Cov_{\gamma_{N,\beta}^u} \left(\nabla_{m_a} \varphi(a), \nabla_{m_b} \varphi(b) \right)
\\
&=
\int_{\chi_N} \nabla_{m_a} \varphi(a) \nabla_{m_b} \varphi(b) \gamma_{N,\beta}^u(\de \varphi)
- \int_{\chi_N} \nabla_{m_a} \varphi(a) \gamma_{N,\beta}^u(\de \varphi)
\int_{\chi_N} \nabla_{m_b} \varphi(b) \gamma_{N,\beta}^u(\de \varphi).
\end{align*}

For $q \in \R^{d \times d}_{\text{sym}}$ small, let $\mathcal{C}^q_{\Z^d}$ be the inverse of the differential operator on gradient fields on~$\Z^d$,
$$
\mathcal{C}^q_{\Z^d} = \left( \mathcal{A}^q_{\Z^d} \right)^{-1},
\quad
\mathcal{A}^q_{\Z^d} = \sum_{i,j=1}^d  \left( \delta_{ij} + q_{ij} \right) \nabla_j^* \nabla_i.
$$
 Let $C^q_{\Z^d}$ be the kernel corresponding to the operator $\mathcal{C}^q_{\Z^d}$.

The following theorem states that in the thermodynamic limit $\Lambda_N \rightarrow \Z^d$ the gradient-gradient covariance is dominated by the covariance $\mathcal{C}^q_{\Z^d}$ of the discrete Gaussian free field on $\Z^d$.

\begin{thm}[Decay of the covariance]\label{Thm:DecayCovariance}
Let $W$ satisfy \eqref{eq:AssumptionsW}.
There is $L_1$ such that for all odd integers $L \geq L_1$ there is $\delta>0$ and $\beta_0$ with the following property. For all $u \in B_{\delta}(0)$ and $\beta \geq \beta_0$ there is $q = q(u,\beta,V) \in \R^{d \times d}_{\mathrm{sym}}$ such that
\begin{align*}
\lim_{N \rightarrow \infty}
\Cov_{\gamma_{N,\beta}^u} \left(\nabla_{m_a} \varphi(a), \nabla_{m_b} \varphi(b) \right)
=
\frac{1}{\beta}\left(
\nabla_{m_b}^* \nabla_{m_a} C^{q}_{\Z^d}(a,b) + R_{ab}
\right).
\end{align*}
Here, $R_{ab}$ can be estimated as follows. There is $\nu > 0$ and a constant $C_1 = C_1(L)$ such that for $a \neq b$
\begin{align*}
\left| R_{ab} \right| \leq C_1 \frac{1}{|a-b|^{d+\nu}}.
\end{align*}
\end{thm}

Let us mention a straightforward consequence of Theorem \ref{Thm:DecayCovariance}.

\begin{cor}[Algebraic decay of the covariance]\label{Cor:Algebraic_DecayCovariance}
Under the assumptions of Theorem~\ref{Thm:DecayCovariance} there is a constant $C$ such that the following estimate holds:
\begin{align*}
\left| 
\lim_{N \rightarrow \infty} \Cov_{\gamma_{N,\beta}^u}\left(
\nabla_{m_a} \varphi(a) \nabla_{m_b} \varphi(b)
\right)
\right|
\leq C \frac{1}{|a-b|^d}.
\end{align*}
\end{cor}


\begin{remark} \label{rem:extension_main_thm}
\begin{enumerate}
	\item As in \cite{Hil19_1} one can state the assumptions \eqref{eq:AssumptionsW} on the potential $W$ in a more general form allowing a bigger class of perturbations $V$. We will comment on this again in the next section, see Lemma \ref{Lemma:From_K_to_V} and Remark \ref{rem:weak_ass_V}. For the sake of simplicity we decided to state the main results with assumptions \eqref{eq:AssumptionsW}.
	\item Theorem \ref{Thm:DecayCovariance} can also be formulated for $m$-component fields on~$\T_N$,
	$$
	\varphi: \Lambda \rightarrow \R^m.
	$$
	Discrete derivatives are understood component-wise,
	$$
	(\nabla_i \varphi)_s(x)= \varphi_s(x+e_i)-\varphi_s(x),
	\quad s \in \lbrace 1, \ldots, m \rbrace, i \in \lbrace 1, \ldots, d \rbrace.
	$$
	The potential $W$ and the perturbation $V$ are maps from $\R^m$ to $\R$ and the tilted boundary condition $u \in \R^d$ is replaced by a deformation $F \in \R^{m \times d}$. See \cite{ABKM} and \cite{Hil19_1} for more details on the set-up. This extension shows up in the notation but does not change the arguments in the proofs.
	\item The statement in Theorem \ref{Thm:DecayCovariance} can also be extended to more general finite-range interaction (not only nearest-neighbour).
	Let $A \subset \Z^d$ be a finite set. Consider the potential
	$$
	W: \left( \R^m \right)^A \rightarrow \R.
	$$
	Then one can define the Hamiltonian with finite-range interaction and \textit{external deformation} $F \in \R^{d \times m}$ as
	$$
	H_N^F(\varphi) = \sum_{x \in\T_N} W\left((\varphi + F)_{\tau_x(A)}\right),
	$$
	where for any $\varphi \in \chi_N$ and $B \subset \Z^d$ we use $\varphi_B$ to denote the restriction of $\varphi$ to $B$, and $\tau_x(A)$ denotes the set $A$ translated by $x$.
	
	For $m=d$, this is the setting for microscopic models of nonlinear elasticity with $F$ representing an affine deformation applied to a solid.
	See \cite{ABKM} and \cite{Hil19_1} for more details on the set-up and \cite{ABKM} for the application to elasticity.

\end{enumerate}

\end{remark}

\subsection{Key theorem and proof of the main result}

The goal of this section is the formulation of a technical key theorem, which states a powerful representation of the generating functional with observables of the model. It is based on a representation obtained in \cite{Hil19_1}, Theorem 2.4. The proof is obtained by a subtle renormalisation group (RG) analysis which is an extension of the RG method in \cite{Hil19_1}. We will sketch the arguments presented in \cite{Hil19_1} in Section \ref{sec:RG-analysis_BulkFlow} and give the proof of the representation needed here in Section \ref{sec:RG-analysis_ObservableFlow}.

\subsubsection{Reformulation of $Z_{N,\beta}(u,f)$}

 As is often the case in statistical mechanics we compute correlation functions as derivatives with respect to an external field, which we refer to as an \textit{observable} field. Namely, we express the gradient-gradient covariance in terms of the perturbed generating partition function:
\begin{align}
\Cov_{\gamma_{N,\beta}^u}\left( \nabla_{m_a} \varphi(a), \nabla_{m_{b}} \varphi(b) \right)
= \partial_s \partial_t \Big\vert_{s=t=0} \ln Z_{N,\beta}\left( u,f_{ab}(s,t) \right)
\label{Cov-Z}
\end{align}
where
\begin{align*}
f_{ab}(s,t) = s \nabla^*_{m_a} \1_a + t \nabla^*_{m_b} \1_b
\end{align*}
is the \textit{observable}. We start by a reformulation of $Z_{N,\beta}(u,f)$ (the very same one as in \cite{Hil19_1}).

~\\
Let $\overline{V}(z,u)$ be the remainder of the linear Taylor expansion of $V(z+u)$ around~$u$,
\begin{align*}
\overline{V}(z,u) = V(z+u) - V(u) - V'(u)z.
\end{align*}
We can write the generating functional $Z_{N,\beta}(u,f)$ from \eqref{eq:gen_func} in the form
\begin{align*}
Z_{N,\beta}(u,f) &=
e^{- \beta L^{Nd} \left( \frac{|u|^2}{2} + \sum_{i=1}^d V(u_i) \right)} 
\\ & \quad \times
\int_{\chi_N} e^{(f,\varphi)} e^{- \beta \sum_{x \in \Lambda_N} \sum_{i=1}^d \left( \overline{V} (\nabla_i \varphi(x), u_i) + \frac{1}{2}|\nabla_i\varphi(x)|^2\right)} \lambda_N(\de \varphi).
\end{align*}
Let
\begin{align}
\mu_{\beta}(\de\varphi) = \frac{1}{Z_{N,\beta}^{(0)}} e^{-\frac{\beta}{2} \sum_{x \in \Lambda_N}\sum_{i=1}^d |\nabla_i \varphi(x)|^2} \lambda_N(\de\varphi)
\end{align}
be the Gaussian measure at inverse temperature $\beta$ with corresponding normalisation
\begin{align}
Z_{N,\beta}^{(0)} = \int_{\chi_N} e^{-\frac{\beta}{2} \sum_{x \in \Lambda_N}\sum_{i=1}^d |\nabla_i \varphi(x)|^2} \lambda_N(\de\varphi).
\label{eq:10}
\end{align}
Consequently,
\begin{align*}
Z_{N,\beta}(u,f)
=
e^{- \beta L^{Nd} \sum_{i=1}^d W(u_i) } Z_{N,\beta}^{(0)}
\int_{\chi_N} e^{(f,\varphi)} e^{- \beta \sum_{x \in \Lambda_N} \sum_{i=1}^d \overline{V} (\nabla_i \varphi(x), u_i)} \mu_{\beta}(\de\varphi).
\end{align*}
Now we rescale the field by $\sqrt{\beta}$ and introduce the \textit{Mayer function} $\mathcal{K}_{u,\beta,V}: \R^d \rightarrow \R$,
\begin{align}
\mathcal{K}_{u,\beta,V}(z) = e^{-\beta \sum_{i=1}^d \overline{V}(\frac{z_i}{\sqrt{\beta}},u_i)} - 1.
\label{eq:zero_perturbation_K}
\end{align}
We can express the partition function $Z_{N,\beta}(u,f)$ in terms of the polymer expansion:
\begin{align*}
Z_{N,\beta}(u,f)
&=
e^{- \beta L^{Nd} \sum_{i=1}^d W(u_i) } Z_{N,\beta}^{(0)}
\int_{\chi_N} e^{\left(f,\frac{\varphi}{\sqrt{\beta}}\right)}
e^{- \beta \sum_{x \in \Lambda_N} \sum_{i=1}^d \overline{V} \left(\frac{\nabla_i \varphi(x)}{\sqrt{\beta}}, u_i\right)} \mu_1(\de\varphi)
\\
&=
e^{- \beta L^{Nd} \sum_{i=1}^d W(u_i)} Z_{N,\beta}^{(0)}
\int_{\chi_N} e^{\left(f,\frac{\varphi}{\sqrt{\beta}}\right)}
\prod_{x \in \Lambda_N} \left( 1 + \mathcal{K}_{u,\beta,V}(\nabla \varphi(x)) \right) \mu_1(\de\varphi)
\\
&=
e^{- \beta L^{Nd} \sum_{i=1}^d W(u_i)} Z_{N,\beta}^{(0)}
\int_{\chi_N} e^{\left(f,\frac{\varphi}{\sqrt{\beta}}\right)}
\sum_{X \subset \Lambda_N} \prod_{x \in X} \mathcal{K}_{u,\beta,V}(\nabla \varphi(x)) \mu_1(\de\varphi).
\end{align*}
The integral in the last expression gives the perturbative contribution
\begin{align*}
\mathcal{Z}_{N,\beta}\left( u,\frac{f}{\sqrt{\beta}} \right) =
\int_{\chi_N} e^{\left(\frac{f}{\sqrt{\beta}},\varphi\right)}
\sum_{X \subset \Lambda_N} \prod_{x \in X} \mathcal{K}_{u,\beta,V}(\nabla \varphi(x)) \mu_1(\de\varphi).
\end{align*}
In summary, we obtain the representation
\begin{align}
Z_{N,\beta}(u,f)
=
e^{- \beta L^{Nd} \sum_{i=1}^d W(u_i)} Z_{N,\beta}^{(0)}
\,
\mathcal{Z}_{N,\beta}\left( u,\frac{f}{\sqrt{\beta}} \right).
\label{repr_gesamt}
\end{align}

We introduce a space for the perturbation $\mathcal{K}_{u,\beta,V}$.
Let $\zeta \in (0,1)$. For $r_0 \geq 3$ we define the Banach space $\mathbf{E}_{\zeta}$ consisting of functions $\mathcal{K}:\R^d \rightarrow\R$ such that the following norm is finite
\begin{align*}
\Vert \mathcal{K} \Vert_{\zeta} = \sup_{z \in \R^d} \sum_{|\alpha| \leq r_0} \frac{1}{\alpha !} \vert \partial^{\alpha}\mathcal{K}(z) \vert e^{-\frac{1}{2}(1- \zeta)|z|^2}.
\end{align*}

Let us generalise the expression for the perturbative part to arbitrary $\mathcal{K} \in \mathbf{E}_{\zeta}$ from the rather explicit $\mathcal{K}_{u,\beta,V}$ in \eqref{eq:zero_perturbation_K}. Namely, let
\begin{align}
\mathcal{Z}_{N}\left( u,f \right) =
\int_{\chi_N} e^{\left(f,\varphi\right)}
\sum_{X \subset \Lambda_N} \prod_{x \in X} \mathcal{K}(\nabla \varphi(x)) \mu_1(\de\varphi).
\label{eq:expression_PF}
\end{align}

Proposition 2.4 in \cite{ABKM} provides conditions on $V$ such that $\mathcal{K} \in B_{\rho}(0) \subset \mathbf{E}_{\zeta}$ for any $\rho>0$ is satisfied. We cite the proposition in the following lemma.

\begin{lemma}
\label{Lemma:From_K_to_V}
Let $W$ satisfy \eqref{eq:AssumptionsW}.
Then there exist $\tilde{\zeta}$, $\delta_0 > 0$, $C_1$ and $\theta > 0$ such that for all $\delta \in \left(0,\delta_0\right]$ and for all $\beta \geq 1$ the map 
$$
B_{\delta}(0) \ni u \mapsto \mathcal{K}_{u,\beta,V} \in \mathbf{E}_{\tilde{\zeta}}
$$
is $C^{r_1}$ and satisfies
\begin{align}
&\Vert \mathcal{K}_{u,\beta,V} \Vert_{\tilde{\zeta}}
\leq
C_1 \left( \delta + \beta^{-\frac{1}{2}} \right)
\quad \text{and} \quad
\sum_{|\gamma|\leq r_1} \frac{1}{\gamma !} \Vert \partial_u^{\gamma} \mathcal{K}_{u,\beta,V} \Vert_{\tilde{\zeta}} \leq \theta.
\label{eq:Bound_K}
\end{align}
In particular, given $\rho>0$, there exist $\delta > 0$ and $\beta_0 \geq 1$ such that for all $\beta \geq \beta_0$ and all $u \in B_{\delta}(0)$ we have
$$
\Vert \mathcal{K}_{u,\beta,V} \Vert_{\tilde{\zeta}} \leq \rho
$$ 
and the bound on the derivatives in \eqref{eq:Bound_K} holds.
\end{lemma}

\begin{remark}\label{rem:weak_ass_V}
As noted in the previous section we can state more general assumptions on the potential $W$ than \eqref{eq:AssumptionsW}. Namely, it is enough to assume the smallness condition on the Mayer function $\mathcal{K}$, $\Vert \mathcal{K}_{u,\beta,V} \Vert_{\tilde{\zeta}} \leq \rho$. Then Theorem \ref{Thm:DecayCovariance} can be applied for every $V$ such that its Mayer function satisfies the bound.
\end{remark}

\subsubsection{Representation of $Z_{N,\beta}(u,f)$ and conclusion}

Let us introduce $\mathcal{C}^q_{\Lambda_N} = \left( \mathcal{A}_{\Lambda_N}^{q} \right)^{-1}$ for $q \in \R^{d \times d}_{\text{sym}}$, where
$$
\mathcal{A}^{q}_{\Lambda_N}:
\chi_N \rightarrow \chi_N,
\quad
\mathcal{A}^{q}_{\Lambda_N}
 = \sum_{i,j=1}^d \left( \delta_{ij} + q_{ij} \right) \nabla_j^* \nabla_i.
$$
We use $\Vert q \Vert$ to denote the operator norm of $q$ viewed as an operator on $\R^d$ equipped with the $l_2$ metric.
If $q$ is small, $\Vert q \Vert \leq \frac{1}{2}$, we can define a Gaussian measure $\mu_{\mathcal{C}^{q}_{\Lambda_N}}$ on $\chi_N$ with covariance $\mathcal{C}^{q}_{\Lambda_N}$,
$$
\mu_{\mathcal{C}^{q}_{\Lambda_N}} (\de \varphi) = \frac{1}{Z_N^{(q)}} e^{-\frac{1}{2}\left( \varphi, \mathcal{A}^{q}_{\Lambda_N} \varphi \right)} \de \lambda_N(\varphi).
$$
Observe that we changed notation from $Z_{N,\beta=1}^{(0)}$ in \eqref{eq:10} to $Z_N^{(0)}$.

~\\
Remember from \eqref{Cov-Z} that with \eqref{repr_gesamt} the covariance can be computed as follows:
\begin{align}
\Cov_{\gamma_{N,\beta}^u}\left( \nabla_{m_a} \varphi(a), \nabla_{m_{b}} \varphi(b) \right)
&= \partial_s \partial_t \Big\vert_{s=t=0} \ln Z_{N,\beta}\left( u,f_{ab}(s,t) \right) \nonumber
\\ &
= \partial_s \partial_t \Big\vert_{s=t=0} \ln \mathcal{Z}_{N,\beta}\left( u,\frac{f_{ab}(s,t)}{\sqrt{\beta}} \right),
\label{eq:Cov-partitionfct}
\end{align}
where
\begin{align}
f_{ab}(s,t) = s \nabla^*_{m_a} \1_a + t \nabla^*_{m_b} \1_b
\end{align}
is the observable. The observable fields $s$ and $t$ are constant external fields which couple to the field $\varphi$ only at the points $a$ and $b$ due to the indicator functions. 
An external field is also employed to analyse the scaling limit in \cite{Hil19_1}, but there the macroscopic regularity of this test function is important. The application of the representation in Theorem 2.4 in \cite{Hil19_1} does not give a good estimate on $Z_N\left(\mathcal{K},\mathcal{C}_{\Lambda_N}^{(q(\mathcal{K}))}f_{ab}\right)$ since $f_{ab}$ is too rough. If we smooth out $f_{ab}$, we can get a decay for the "smoothed covariance" by exploiting the decay $\eta^N$. This is done in \cite{Hil16}.

Instead we use a finer analysis based on the RG method for the bulk flow but extended to observables and obtain a refined representation of the generating partition function in Theorem \ref{Thm:RepresentationExtendedPartitionFunction}.

~\\
In view of \eqref{eq:Cov-partitionfct}, we are only interested in the behaviour of $\mathcal{Z}_{N,\beta}\left( u, \frac{f_{ab}(s,t)}{\sqrt{\beta}}\right)$ up to first order in $s,t$ and $st$. To make this precise, one considers the quotient algebra in which two maps of $s,t$ become equivalent if their formal power series in $s, t$ agree to order $1,s,t,st$, see Section \ref{sec:RG-analysis_ObservableFlow} for the details.

\begin{thm}[Representation of the extended partition function]\label{Thm:RepresentationExtendedPartitionFunction}
Fix $a,b \in \Lambda_N$, $\zeta \in (0,1)$ and $\eta \in~(0,\frac{1}{4})$. There is $L_1$ such that for all odd integers $L \geq L_1$ there is $\epsilon_1 > 0$ with the following properties.
For any $N \in \N$ there exist smooth maps (with bounds on the derivatives which are independent of $N$)
$$
\lambda: B_{\epsilon_1}(0) \subset \mathbf{E}_{\zeta} \rightarrow \R, \quad
q:  B_{\epsilon_1}(0) \subset \mathbf{E}_{\zeta} \rightarrow \R^{d \times d}_{\text{sym}},
$$
and, for any $N \in \N$, a smooth map
 $Z_N^{\ext}: B_{\epsilon_1}(0) \times \chi_N \rightarrow \R$ such that (up to first order in $s$ and $t$)
\begin{align}
\mathcal{Z}_{N}(u,f_{ab})
=\frac{Z_N^{\left(q(\mathcal{K})\right)}}{Z_N^{\left(0\right)}} e^{-L^{Nd}|\lambda(\mathcal{K})|}
e^{st q_N^{ab} + s \lambda_N^a + t \lambda_N^b}
Z_N^{\ext}(\mathcal{K},0).
\label{eq:formula_extended_partition_function}
\end{align}
There is a constant $C_1 = C_1(L)$, such that
$$
q_N^{ab} = \nabla^*_{m_b} \nabla_{m_a} C_{\Lambda_N}^q(a,b) + R_{ab},
\quad
\vert R_{ab}\vert \leq C_1 \frac{1}{|a-b|^{d+\nu}},
$$
where $0 < \nu \leq - \frac{\ln(4 \eta)}{\ln L}$,
and
$\lambda^a_N$ and $\lambda^b_N$ are uniformly bounded in $N$.

Moreover, the remainder $Z_N^{\ext}(\mathcal{K},0)$ can be expressed (up to first order in $s$ and $t$) as follows:
\begin{align*}
&Z_N^{\ext}(\mathcal{K},0) = Z_N(\mathcal{K},0) + s K_N^a + t K_N^b + st K_N^{ab},
\\
& \left\vert Z_N(\mathcal{K},0) - 1 \right\vert = \mathcal{O}\left(\eta^N \right),
\quad
K_N^a, K_N^b = \mathcal{O}\left(2^{-N}\right), 
\quad
K_N^{ab} = \mathcal{O} \left(\eta^N 4^{-N}\right). 
\end{align*}
\end{thm}

The maps $\lambda$, $q$ and $Z_N$ are the ones that are studied in Theorem 2.4 in \cite{Hil19_1}.

~\\
This representation can be used for a straightforward proof of Theorem~\ref{Thm:DecayCovariance}.

\begin{proof}[Proof of Theorem \ref{Thm:DecayCovariance}]
Let $\tilde{\zeta}$ be the parameter from Lemma \ref{Lemma:From_K_to_V}, fix $\eta \in (0,1/4)$ and let $L_1$ and $\epsilon_1$ be the corresponding parameters from Theorem 2.4 in \cite{Hil19_1} (stated below in Theorem \ref{Thm:RepresentationPartitionFunction}). Then, for $\beta$ large enough and $\delta$ small enough, $\mathcal{K}_{u,\beta,V} \in B_{\epsilon_1}(0) \subset \mathbf{E}_{\tilde{\zeta}}$ is satisfied. Therefore we can apply the representation \eqref{eq:formula_extended_partition_function} from Theorem \ref{Thm:RepresentationExtendedPartitionFunction} with
$$
f_{ab}(\tilde{s}, \tilde{t}) = f_{ab}\left(\frac{s}{\sqrt{\beta}}, \frac{t}{\sqrt{\beta}}\right)
$$
in the computation of the correlations as follows:
\begin{align*}
\Cov_{\gamma_{N,\beta}^u}\left( \nabla_{m_a} \varphi(a), \nabla_{m_{b}} \varphi(b) \right)
&
= \partial_s \partial_t \Big\vert_{s=t=0}
\ln \mathcal{Z}_{N,\beta}\left( u,\frac{f_{ab}(s,t)}{\sqrt{\beta}} \right)
\\ &=
\partial_s \partial_t \Big\vert_{s=t=0}
\ln \mathcal{Z}_{N,\beta}\left( u,f_{ab}(\tilde{s},\tilde{t}) \right)
\\&=
\partial_s \partial_t \Big\vert_{s=t=0}
\ln \left[ e^{st \frac{q_N^{ab}}{\beta} + s \frac{\lambda_N^a}{\sqrt{\beta}} + t \frac{\lambda_N^b}{\sqrt{\beta}}} Z_N^{\ext}(\mathcal{K}_{u,\beta,V},0) \right]
\\&=
\frac{1}{\beta}q_N^{ab}
+ \frac{K_N^{ab}}{\beta Z_N^{\es}(\mathcal{K}_{u,\beta,V},0)}
- \frac{K_N^a K_N^b}{\beta Z_N^{\es}(\mathcal{K}_{u,\beta,V},0)}
\\&=
\frac{1}{\beta}
\left(
 \nabla^*_{m_b} \nabla_{m_a} C_{\Lambda_N}^q(a,b) + R_{ab} + \mathcal{O}\left(2^{-N}\right)
\right).
\end{align*}
By a standard argument $C_{\Lambda_N}^q \rightarrow C^q_{\Z^d}$ as $N \rightarrow \infty$, and thus the theorem is proven.
\end{proof}

\section{RG analysis for the bulk flow} \label{sec:RG-analysis_BulkFlow}

To prove Theorem \ref{Thm:RepresentationExtendedPartitionFunction} we extend the proof of Theorem 2.4 in \cite{Hil19_1} by \textit{observables}. In this section we outline the steps from \cite{Hil19_1} which are needed for the extended proof in the next section. For motivations and details we refer to the original paper.

~\\
The goal of Theorem 2.4 in \cite{Hil19_1} is an expression for
\begin{align*}
\mathcal{Z}_N(\mathcal{K},f)
= \int_{\chi_N} e^{(f,\varphi)} \sum_{X \subset \Lambda_N} \prod_{x \in X} \mathcal{K}(\nabla\varphi(x)) \mu_1(\de\varphi)
\end{align*}
where $u \in \R^d$, $f \in \chi_N$, $\mathcal{K} \in \mathbf{E}_{\zeta}$ and $\zeta \in (0,1)$ fixed.

\begin{thm}[Theorem 2.4 in \cite{Hil19_1}]\label{Thm:RepresentationPartitionFunction}
Fix $\zeta, \eta \in (0,1)$. There is $L_0$ such that for all odd integers $L \geq L_0$ there is $\epsilon_0 > 0$ with the following properties. There exist smooth maps
$$
\lambda: B_{\epsilon_0}(0) \subset \mathbf{E}_{\zeta} \rightarrow \R,\quad q: B_{\epsilon_0}(0) \subset \mathbf{E}_{\zeta} \rightarrow \R^{d \times d}_{\mathrm{sym}},
$$
and, for any $N \in \N$, a smooth map $Z_N^{\es}: B_{\epsilon_0}(0) \times \chi_N \rightarrow \R$ such that for any $f \in \chi_N$ and $\mathcal{K} \in B_{\epsilon_0}(0)$ the following representation holds:
\begin{align}
\mathcal{Z}_{N}(\mathcal{K},f) =  e^{\frac{1}{2}\left(f,\mathcal{C}_{\Lambda_N}^{q(\mathcal{K})}f\right)} \frac{Z_N^{(q(\mathcal{K}))}}{Z_N^{(0)}} e^{- L^{Nd} \lambda(\mathcal{K}) } Z_N^{\es}\left(\mathcal{K},\mathcal{C}_{\Lambda_N}^{q(\mathcal{K})}f\right).
\label{formula...}
\end{align}
If $f(x) = g_N(x)-c_N$, $g_N(x)=L^{-N\frac{d+2}{2}}g(L^{-N}x)$ for $g \in C_c^{\infty}(\T^d)$ with $\int g =0$, $c_N$ such that $\sum_{x \in \T_N}f(x) = 0$, then there is a constant $C$ which is independent of $N$ such that the remainder $Z_N^{\es}(\mathcal{K})$ satisfies the estimate
\begin{align*}
 \left\vert Z_N^{\es} \left(\mathcal{K}, \mathcal{C}_{\Lambda_N}^{q(\mathcal{K})}f\right) - 1 \right\vert \leq C \eta^N.
\end{align*}
\end{thm}

~\\
The study of the bulk flow is carried out by renormalisation group analysis, an iterative averaging process over different scales. By this method, the quantity $Z_N^{\es}$ in \eqref{formula...} arises in the following way.
As a first step, we write
\begin{align}
\mathcal{Z}_N(\mathcal{K},f)
= \frac{Z_N^{q(\mathcal{K})}}{Z_N^{(0)}} e^{-L^{Nd} \lambda(\mathcal{K})}
\int_{\chi_N} e^{(\varphi,f)} F_0^{\es}(\Lambda_N, \varphi) \mu_{\mathcal{C}^{q(\mathcal{K})}} (\de\varphi),
\end{align}
where
$$
\mathcal{C}^q = \mathcal{C}^q_{\Lambda_N} = \left( \mathcal{A}^q_{\Lambda_N} \right)^{-1},
\quad
\mathcal{A}^q_{\Lambda_N}
= \sum_{i,j=1}^d \left( \delta_{ij} + q_{ij} \right) \nabla_j^* \nabla_i,
$$
is the covariance of the Gaussian free field on $\Lambda_N$. For ease of notation, we dropped the subscript $\Lambda_N$ above. The map $F_0^{\es}$ contains the added Gaussian part $e^{\frac{1}{2}(\nabla \varphi, q(\mathcal{K}) \nabla \varphi)}$, the constant term $e^{L^{Nd}\lambda(\mathcal{K})}$ and the $\mathcal{K}$-term.

~\\
A \textit{finite-range decomposition} of $\mu_{\mathcal{C}^{q}} = \mu_{\mathcal{C}_1} \ast \ldots \ast \mu_{\mathcal{C}_{N}}$ enables us to integrate out iteratively scale by scale,
\begin{align*}
\int_{\chi_N} F_0^{\es}(\varphi + \phi) \mu_{\mathcal{C}^{q}}(\de \varphi)
&=
\int_{\chi_N} F_0^{\es}(\xi_1 + \ldots + \xi_{N} + \phi) \mu_{\mathcal{C}_1}(\de\xi_1) \ldots \mu_{\mathcal{C}_{N}}(\de\xi_{N})
\\ &
= \int_{\chi_N} F_1^{\es}(\xi_2 + \ldots + \xi_{N} + \phi) \mu_{\mathcal{C}_2}(\de\xi_2)\ldots \mu_{\mathcal{C}_{N}}(\de\xi_{N})
\\
&= \ldots
\\ &
= \int_{\chi_N} F_{N-1}^{\es}(\xi_{N} + \phi) \mu_{\mathcal{C}_{N}}(\de \xi_{N})
= F_N^{\es}(\phi)
.
\end{align*}

~\\
The map $F_0^{\es}$ can be written by polymer expansion as
$$
F_0^{\es}(\Lambda_N) = \sum_{X \subset \Lambda_N} e^{H_0^{\es}(X)} \circ K_0^{\es}(\Lambda_N \setminus X)
= \left( e^{H_0^{\es}} \circ K_0^{\es} \right) (\Lambda_N).
$$

This decomposition can be maintained on each scale $k\in \lbrace 1, \ldots, N \rbrace$, that is there are maps $\left(H_k^{\es}, K_k^{\es}\right)$ such that $F_k^{\es} = e^{H_k^{\es}} \circ K_k^{\es}$. This so-called \textit{circ product} acts on scale~$k$ with polymers consisting of \textit{$k$-blocks}, which are cubes of side length $L^k$ (a precise definition can be found in \eqref{eq:defn_circ_product} in Subsection \ref{subsubsec:polymers_functionals_norms}). At the last scale $N$ there is only one block left, namely the whole set $\Lambda_N$, and the circ product is just a sum of two terms, $\left(e^{H_N^{\es}} + K_N^{\es}\right)(\Lambda_N)$.

 The maps $H_k^{\es}$ are the \textit{relevant} (more precisely: relevant and marginal) directions which collect all increasing (and constant) parts in the procedure $F \mapsto \mu_{k+1}\ast F$ and they live in finite dimensional spaces. The maps $K_k^{\es}$ collect all \textit{irrelevant} directions.

~\\
This method is described and performed in detail in \cite{BS1}, \cite{BS2}, \cite{BS3}, \cite{BS4} and \cite{BS5} and adapted to gradient models in \cite{AKM16}, \cite{ABKM} and \cite{Hil19_1}.

~\\
In the next subsections we introduce the finite range decomposition, the norms and spaces for the functionals, the renormalisation map $(H_k^{\es}, K_k^{\es}) \mapsto (H_{k+1}^{\es}, K_{k+1}^{\es})$, key properties of the map and the existence of the finite volume and global flow. Most of the presented material is adopted unchanged from \cite{Hil19_1}. We just skipped details which are not needed for the extension in Section \ref{sec:RG-analysis_ObservableFlow}.

\subsection{Definitions} \label{subsec:finite-volume_bulk_single_step}

We start by describing the \textit{finite-range decomposition} of the measure $\mu_{\mathcal{C}^q}$. See \cite{Hil19_1} for details.

\subsubsection{Finite-range decomposition}\label{sec:FRD}

Let $C^{q}: \Lambda_N \rightarrow \R$ be the kernel to the operator $\mathcal{C}^q$, i.e.,
$$
\mathcal{C}^{q}\varphi (x) = \sum_{y \in \Lambda_N}C^{q}(x-y) \varphi(y).
$$
The next proposition is Theorem 2.3 in \cite{BUC18}.

\begin{prop}[Finite-range decomposition]\label{Prop:FRD_Buchholz}
Fix $q \in \R^{d \times d}_{\text{sym}}$ such that $\mathcal{C}^{q}$ is positive definite. Let $L > 3$ be an odd integer and $N \geq 1$. Then there exist positive, translation invariant operators $\mathcal{C}_k^{q}$ such that
\begin{align*}
\mathcal{C}^{q} &= \sum_{k=1}^{N+1} \mathcal{C}^{q}_k,
\\
C^{q}_k(x) &= -M_k \quad \text{for} \quad |x|_{\infty} \geq \frac{L^k}{2}, \quad k \in \lbrace 1, \ldots, N \rbrace,
\end{align*}
where $M_k \geq 0$ is a constant that is independent of $q$. The following bounds hold for any positive integer $l$ and any multiindex $\alpha$:
\begin{align*}
\sup_{x \in \Lambda_N}
\sup_{\Vert \dot{q} \Vert \leq \frac{1}{2}} \Big\vert \nabla^{\alpha} D^l_q C^{q}_k(x)(\dot{q}, \ldots, \dot{q}) \Big\vert
\leq
\begin{cases}
C_{\alpha,l} L^{-(k-1)(d-2+|\alpha|)} & \quad \text{for } d + |\alpha| > 2 \\
C_{\alpha,l} \ln(L) L^{-(k-1)(d-2+|\alpha|)} & \quad \text{for } d + |\alpha| = 2 .
\end{cases}
\end{align*}
Here, $C_{\alpha,l}$ denotes a constant that does not depend on $L$, $N$, and $k$.
\end{prop}
In \cite{BUC18} further bounds in Fourier space are stated. 
For the sake of simplicity they are omitted here.

~\\
The last two covariances are combined to a single one:
\begin{align}
\mathcal{C}_{N,N}^{q} = \mathcal{C}_N^{q} + \mathcal{C}_{N+1}^{q}. \label{eq:Last_Scale_Cov}
\end{align}
We use the following decomposition:
\begin{align}
\mathcal{C}^{q} &= \sum_{k=1}^{N-1} \mathcal{C}^{q}_k + \mathcal{C}^{q}_{N,N}.
\label{eq:dec}
\end{align}

Let us denote by $\mu_k$ the Gaussian measure with covariance $\mathcal{C}_k^q$.

\subsubsection{Polymers, functionals and norms}\label{subsubsec:polymers_functionals_norms}

In this subsection, we discuss several key notions and introduce the setting of the scales and spaces for functionals. The representation is exactly as in \cite{Hil19_1}.

~\\
At each scale $k$ we pave the torus with blocks of side length $L^k$. These so-called \textit{$k$-blocks} are translations by $(L^k \Z)^d$ of the block $B_0 = \left\lbrace z \in \Z^d: |z_i| \leq \frac{L^k-1}{2} \right\rbrace $. Together, they form the set of $k$-blocks denoted by
$$
\mathcal{B}_k = \lbrace B: B \text{ is a $k$-block} \rbrace.
$$
Unions of blocks are called \textit{polymers}. For $X \subset \Lambda$ let $\mathcal{P}_k(X)$ be the set of all $k$-polymers in $X$ at scale~$k$.

Furthermore we need the following notations:
\begin{itemize}
\item
A polymer $X$ is \textit{connected} if for any $x,y \in X$ there is a path $x_1 = x,x_2, \ldots$, $x_n = y$ in $X$ such that $|x_{i+1} - x_i|_{\infty}=1$ for $i=1,\ldots, n-1$. The set of all connected $k$-polymers in $X$ is denoted by $\mathcal{P}_k^c(X)$. The set of connected components of a polymer $X$ is denoted by $\mathcal{C}_k(X)$.
\item
Let $\mathcal{B}_k(X)$ be the set of $k$-blocks contained in $X$ and $|X|_k = |\mathcal{B}_k(X)|$ be the number of $k$-blocks in $X$.
\item
The \textit{closure} $\bar{X} \in \mathcal{P}_{k+1}$ of $X \in \mathcal{P}_k$ is the smallest $(k+1)$-polymer containing~$X$.
\item
The set of \textit{small polymers} $\mathcal{S}_k$ is given by all polymers $X \in \mathcal{P}_k^c$ such that $|X|_k \leq 2^d$. The other polymers in $\mathcal{P}_k \setminus \mathcal{S}_k$ are \textit{large}.
\item
For any block $B \in \mathcal{B}_k$ let $\hat{B} \in \mathcal{P}_k$ be the cube of side length $(2^{d+1}+1)L^k$ centered at~$B$.
\item
The \textit{small set neighbourhood} $X^{*} \in \mathcal{P}_{k-1}$ of $X \in \mathcal{P}_k$ is defined by
$$
X^* = \bigcup_{B \in \mathcal{B}_{k-1}(X)} \hat{B}.
$$
\item
The \textit{large neighbourhood} $X^+$ of $X \in \mathcal{P}_k$ is defined by
$$
X^+ = \bigcup_{\substack{B \in \mathcal{B}_k:\\ B \text{ touches }X}} B \cup X.
$$
\end{itemize}

Additionally, we introduce a class of functionals.
\begin{itemize}
	\item Let $M(\mathcal{V}_N)$ be the set of measurable real functions on $\mathcal{V}_N$ with respect to the Borel-$\sigma$-algebra.
	\item Let $\mathcal{N}^{\es}$ be the space of real-valued functions of $\varphi$ which are in $C^{r_0}$.
	\item A map $F: \mathcal{P}_k \rightarrow \mathcal{N}^{\es}$ is called \textit{translation invariant} if for every $y \in (L^k \Z)^d$ we have $F(\tau_y(X),\tau_y(\varphi)) = F(X,\varphi)$ where $\tau_y(B) = B+y$ and $\tau_y \varphi(x) = \varphi(x-y)$.
	\item A map $F: \mathcal{P}_k \rightarrow \mathcal{N}^{\es}$ is called \textit{local} if $\varphi \big\vert_{X^*} = \psi \big\vert_{X^*}$ implies $F(X,\varphi) = F(X,\psi)$.
	\item A map $F: \mathcal{P}_k \rightarrow \mathcal{N}^{\es}$ is called \textit{shift invariant} if $F(X, \varphi + \psi) = F(X,\varphi)$ for $\psi$ such that $\psi(x) = c$, $x \in X^*$ on each connected component of $X^*$.
\end{itemize}

We set
\begin{align*}
M(\mathcal{P}_k, \mathcal{V}_N)
= \lbrace
F: \mathcal{P}_k \rightarrow \mathcal{N}^{\es} \big\vert
F(X) \in M(\mathcal{V}_N), F \text{ translation inv., shift inv., local}
\rbrace.
\end{align*}

Notice that we included $C^{r_0}$-smoothness in the definition of the space $M(\mathcal{P}_k,\mathcal{V}_N)$ which is not done in \cite{ABKM}.

Generalisations of $M(\mathcal{P}_k,\mathcal{V}_N)$ are given by $M(\mathcal{P}^c_k,\mathcal{V}_N)$, $M(\mathcal{S}_k,\mathcal{V}_N)$ and $M(\mathcal{B}_k,\mathcal{V}_N)$ where the first component is changed appropriately. We will write $M(\mathcal{P}_k)$, $M(\mathcal{P}^c_k)$, $M(\mathcal{S}_k)$ and $M(\mathcal{B}_k)$ for short.

~\\
The \textit{circ product} of two functionals $F,G \in M(\mathcal{P}_k)$ is defined by
\begin{align}
(F \circ G)(X) = \sum_{Y \in \mathcal{P}_k(X)} F(Y) G(X \setminus Y).
\label{eq:defn_circ_product}
\end{align}

The space of \textit{relevant Hamiltonians} $M_0(\mathcal{B}_k)$, a subspace of $M(\mathcal{B}_k)$, is given by all functionals of the form
$$
H(B,\varphi) = \sum_{x \in B} \mathcal{H}\left(\lbrace x \rbrace, \varphi\right)
$$
where $\mathcal{H}(\lbrace x \rbrace, \varphi)$ is a linear combination of the following \textit{relevant monomials}:
\begin{itemize}
\item
The constant monomial $M(\lbrace x \rbrace)_{\es}(\varphi) = 1$;
\item
the linear monomials $M(\lbrace x \rbrace)_{\beta}(\varphi) = \nabla^{\beta} \varphi(x)$ for $1 \leq |\beta| \leq \lfloor \frac{d}{2} \rfloor + 1$;
\item
the quadratic monomials $M(\lbrace x \rbrace)_{\beta,\gamma}(\varphi) = \nabla^{\beta} \varphi(x) \nabla^{\gamma} \varphi(x)$ for $1 = |\beta| = |\gamma|$.
\end{itemize}

~\\
Next we introduce norms on the space of functionals. Fix $r_0 \in \N$, $r_0 \geq 3$.

~\\
Define
\begin{align*}
&\bigoplus_{r=0}^{\infty} \mathcal{V}_N^{\otimes r}
\\ & \quad
= \left\lbrace g = \left(g^{(0)}, g^{(1)}, \ldots \right) \Big\vert \, g^{(r)} \in \mathcal{V}_N^{(r)}, \text{ only finitely many non-zero elements} \right\rbrace.
\end{align*}

The space of test function is given by
$$
\Phi = \Phi_{r_0} = \left\lbrace g \in \bigoplus_{r=0}^{\infty} \mathcal{V}_N^{\otimes r}: g^{(r)} = 0 \,\, \forall r \geq r_0 \right\rbrace.
$$

A norm on $\Phi$ is given as follows: On $\mathcal{V}_N^{\otimes 0} = \R$ we take the usual absolute value on~$\R$.
For $\varphi \in \mathcal{V}_N$ we define
\begin{align*}
\vert \varphi \vert_{j,X} = \sup_{x \in X^*} \sup_{1 \leq |\alpha| \leq p_{\Phi}} \mathfrak{w}_j(\alpha)^{-1} \big\vert \nabla^{\alpha} (\varphi)(x) \big\vert
\end{align*}
where $\mathfrak{w}_j(\alpha) = h_j L^{-j|\alpha|} L^{-j \frac{d-2}{2}}$, $h_j = 2^j h$ and $p_{\Phi} = \left\lfloor \frac{d}{2} \right\rfloor + 2$.
For $g^{(r)} \in \mathcal{V}_N^{\otimes r}$ we define
\begin{align*}
\left\vert g^{(r)} \right\vert_{j,X}
= \sup_{x_1, \ldots, x_r \in X^*} \sup_{1 \leq |\alpha_1|, \ldots, |\alpha_r| \leq p_{\Phi}} \left(\prod_{l=1}^r \mathfrak{w}_j(\alpha_l)^{-1}\right) \nabla^{\alpha_1} \otimes \ldots \otimes \nabla^{\alpha_r} g^{(r)}(x_1, \ldots, x_r).
\end{align*}
Then set $|g|_{j,X} = \sup_{r \leq r_0} \left|g^{(r)}\right|_{j,X}$.

~\\
A homogeneos polynomial $P^{(r)}$ of degree $r$ on $\mathcal{V}_N$ can be uniquely identified with a symmetric $r$-linear form and hence with an element $\overline{P^{(r)}}$ in the dual of $\mathcal{V}_N^{\otimes r}$.
So we can define the pairing
$$
\langle P,g \rangle = \sum_{r=0}^{\infty} \left\langle \overline{P^{(r)}}, g^{(r)} \right\rangle
$$
and a norm
$$
\vert P \vert_{j,X} = \sup \left\lbrace \langle P,g \rangle: g \in \Phi, |g|_{j,X} \leq 1 \right\rbrace.
$$
For $F \in C^{r_0}(\mathcal{V}_N)=\mathcal{N}^{\es}$ the pairing is given by $\langle F,g \rangle_{\varphi} = \langle \mathrm{Tay}_{\varphi} F,g \rangle$ which defines a norm
$$
\vert F \vert_{j,X,T_{\varphi}} = \vert \mathrm{Tay}_{\varphi} F \vert_{j,X}
= \sup \left\lbrace \langle F,g \rangle_{\varphi}: g \in \Phi, |g|_{j,X} \leq 1 \right\rbrace.
$$
Here, $\mathrm{Tay}_{\varphi} F $ denotes the Taylor polynomial of order $r_0$ of $F$ at $\varphi$.

~\\
Let $F \in M(\mathcal{P}_k^c)$. In \cite{ABKM} weights $W_k^{X}, w_k^X, w_{k:k+1}^X \in M(\mathcal{P}_k)$ are defined. Useful properties are summarized in Lemma \ref{lemma:Properties_of_weights}. Weighted norms are given~by
\begin{align*}
\vertiii{F(X)}_{k,X} &= \sup_{\varphi} \vert F(X) \vert_{k,X,T_{\varphi}} W_k^X(\varphi)^{-1},
\\
\Vert F(X) \Vert_{k,X} &= \sup_{\varphi} \vert F(X) \vert_{k,X,T_{\varphi}} w_k^X(\varphi)^{-1},
\\
\Vert F(X) \Vert_{k:k+1,X} &= \sup_{\varphi} \vert F(X) \vert_{k,X,T_{\varphi}} w_{k:k+1}^X(\varphi)^{-1}.
\end{align*}

The global weak norm for $F \in M(\mathcal{P}_k^c)$ for $A \geq 1$ is given by
$$
\Vert F \Vert_k^{(A)} = \sup_{X \in \mathcal{P}_k^c} \Vert F(X) \Vert_{k,X} A^{|X|_k}.
$$

~\\
 A norm on relevant Hamiltonians is given as follows. For $H \in M_0(\mathcal{B}_k)$ we can write
$$
H(B,\varphi) = \sum_{x \in B}
\left(
 a_{\es} +
 \sum_{\beta \in \mathfrak{v}_1} a_{\beta} \nabla^{\beta} \varphi(x) +
\sum_{x \in B} \sum_{\beta,\gamma \in \mathfrak{v}_2} a_{\beta,\gamma} \nabla^{\beta} \varphi(x) \nabla^{\gamma} \varphi(x)
\right).
$$
Here
\begin{align*}
\mathfrak{v}_1 &= \left\lbrace \beta \in \N_0^{\mathcal{U}}, 1 \leq |\beta| \leq \left\lfloor \frac{d}{2} \right\rfloor + 1 \right\rbrace,
\\
\mathfrak{v}_2 &= \left\lbrace (\beta, \gamma) \in \N_0^{\mathcal{U}} \times \N_0^{\mathcal{U}}, |\beta| = |\gamma|=1, \beta < \gamma \right\rbrace,
\end{align*}
where $\mathcal{U} = \lbrace e_1, \ldots, e_d \rbrace$ and the expression $\beta < \gamma$ refers to any ordering of $\lbrace e_1, \ldots, e_d \rbrace$.
With these preparations we define a norm on $M_0(\mathcal{B}_k)$ as follows:
$$
\Vert H \Vert_{k,0} =
L^{dk} \left| a_{\es}\right|
+ \sum_{\beta \in \mathfrak{v}_1} h_k L^{kd} L^{-k\frac{d-2}{2}} L^{-k|\beta|} \left| a_{\beta} \right|
+ \sum_{(\beta,\gamma)\in \mathfrak{v}_2} h_k^2 \left| a_{(\beta,\gamma)} \right|.
$$

\subsubsection{The renormalisation map}\label{subsubsec:The_renormalisation_map}

We use the finite-range decomposition of $\mathcal{C}^{q}$ into covariances $\mathcal{C}^q_1,\ldots, \mathcal{C}^q_{N-1},\mathcal{C}^q_{N,N}$ defined in Subsection \ref{sec:FRD} (see \eqref{eq:dec}) and the corresponding decomposition of the measure $\mu_{\mathcal{C}^q} = \mu_1 \ast \ldots \ast \mu_N \ast \mu_{N,N}$.

The renormalisation map is defined as
\begin{align*}
\mathcal{R}_k F(\varphi) = \int_{\chi_N}F(\varphi + \xi) \mu_k(\de \xi).
\end{align*}
Then
$$
\int_{\chi_N} F(\varphi) \mu_{\mathcal{C}^{(q)}}(\de \varphi)
= \mathcal{R}_{N,N} \mathcal{R}_{N-1} \ldots \mathcal{R}_1(F)(0).
$$

The flow under $\mathcal{R}_k$ is described by two sequences of functionals $H_k \in M_0(\mathcal{B}_k)$ and $K_k \in M(\mathcal{P}_k^c)$. In the following we define those sequences 
as far as it is needed for the understanding of the extension to observables in the next section.

~\\
The flow is given by
\begin{align*}
\mathbf{T}_k: M_0(\mathcal{B}_k) \times M(\mathcal{P}_k^c) \times \R_{\text{sym}}^{d \times d}
&\quad\rightarrow\quad
 M_0(\mathcal{B}_{k+1}) \times M(\mathcal{P}_{k+1}^c),
\\
(H,K,q) &\quad\mapsto\quad (H_+,K_+).
\end{align*}
Note that we sometimes omit the scale $k$ from the notation; if doing so, the $+$ indicates the change of scale from $k$ to $k+1$. The maps $H_+ \in M_0(\mathcal{B}_{k+1})$ and $K_+ \in M(\mathcal{P}_{k+1})$ are chosen such that
$$
\mathcal{R}_+ (e^H \circ K)(\Lambda_N) = (e^{H_+} \circ K_+)(\Lambda_N).
$$

The relevant part of the flow on the next scale, the map $H_+$, is defined as follows: For $B_+ \in  \mathcal{B}_{k+1}$
\begin{align*}
H_+(B_+)
&= \mathbf{A}^q_k H (B_+) + \mathbf{B}^q_k K(B_+)
\\
&= \sum_{B \in \mathcal{B}_{k+1}(B_+)} \Pi_2 \mathcal{R}_{k+1} H(B)
+ \sum_{B \in \mathcal{B}_{k+1}(B_+)} \Pi_2 \mathcal{R}_{k+1} K(B).
\end{align*}

Here, $\Pi_2: M(\mathcal{B}_k) \rightarrow M_0(\mathcal{B}_k)$ is a projection on the space of relevant Hamiltonians. Heuristically, for $F \in M(\mathcal{B}_{k})$, $\Pi_2 F$ is attained as homogenisation of the second order Taylor expansion of $F(B)$ given by $\dot{\varphi} \mapsto  F(B,0) + DF(B,0) \dot{\varphi} + \frac{1}{2} D^2F(B,0)(\dot{\varphi},\dot{\varphi})$.
More precisely, $\Pi_2 F$ is the relevant Hamiltonian  $F(B,0) + l(\dot{\varphi}) + Q(\dot{\varphi},\dot{\varphi})$ where $l$ is the unique linear relevant Hamiltonian that satisfies $l(\dot{\varphi}) = DF(B,0) \dot{\varphi}$ for all $\dot{\varphi}$ who are polynomials of order $\left\lfloor \frac{d}{2} + 1 \right\rfloor$ on $B^+$, and $Q$ is the unique quadratic relevant Hamiltonian that agrees with $\frac{1}{2} D^2 F(B,0)(\dot{\varphi},\dot{\varphi})$ on all $\dot{\varphi}$ which are affine on $B^+$. These heuristics are made precise in \cite{ABKM}, Section 8.4.

For the definition of the irrelevant part $K_+$ of the flow at the next scale, set
$$
\widetilde{H}(B) = \Pi_2 \mathcal{R}_{k+1} H(B) + \Pi_2 \mathcal{R}_{k+1} K(B),
$$
and for $X \in \mathcal{P}_k$ and $U \in \mathcal{P}_{k+1}$,
\begin{align*}
&\chi(X,U) = \1_{\pi(x)=U}, \quad\mathrm{where}
\\
&\pi(X) = \bigcup_{Y \in \mathcal{C}(X)} \tilde{\pi}(Y) \quad \mathrm{and}
\\
&\tilde{\pi}(Y)
= \begin{cases}
\bar{X} & \text{ if } X \in \mathcal{P}^c \setminus \mathcal{S},
\\
B_+ & \text{ where } B_+ \in \mathcal{B}_+ \text{ with } B_+ \cap X \neq \es \text{ for } X \in \mathcal{S} \setminus \es,
\\
\es & \text{ if } X = \es.
\end{cases}
\end{align*}
Then
\begin{align}
K_+(U,\varphi) &= \mathbf{S}^q_k (H_+,K_+)(U,\varphi)
\nonumber
\\
&= \sum_{X \in \mathcal{P}} \chi(X,U) \left( e^{\widetilde{H}(\varphi)} \right)^{U \setminus X} \left( e^{\widetilde{H}(\varphi)} \right)^{- X \setminus U}
\nonumber
\\ & \qquad \times
\int 
\left[
\left( 1 - e^{\widetilde{H}(\varphi)} \right)
\circ \left( e^{H(\varphi + \xi)} - 1 \right)
\circ K(\varphi + \xi)
\right]
(X) \mu_+(\de\xi).
\label{eq:defn_K+}
\end{align}
If the dependence of $\mathbf{S}^q_k$ on $q$ is not of direct importance we omit it from the notation.

~\\
For the construction of the infinite-volume flow later we  consider the family $(K^{\Lambda})_{\Lambda}$ in dependence on the torus $\Lambda$. More precisely, we consider tori $\Lambda_N$ with increasing side length $L^N$, $N \in \N$. Let $\mathcal{P}_k(\Z^d)$ be the set of finite unions of $k$-blocks in~$\Z^d$. We need the following compatibility condition.

\begin{defn}\label{Defn:Zd-property}
We say that a family of maps $(K^{\Lambda})_{\Lambda}$ satisfies the $(\Z^d)$-property if for any $X \in \mathcal{P}_k(\Z^d)$ and for $\Lambda \subset \Lambda'$ satisfying $\diam (X) \leq \frac{1}{2} \diam(\Lambda)$ it holds that
$$
K^{\Lambda}(X) = K^{\Lambda'}(X).
$$
\end{defn}

~\\
We review the following properties of the map $(H,K) \mapsto K_+$ from Lemma 6.4 in \cite{ABKM}, and Propositions 3.8 and 3.10 in \cite{Hil19_1}.

\begin{lemma}\label{Prop:field_locality}
For $H \in M_0(\mathcal{B}_k)$ the functional $K_+$ defined above has the following properties.
\begin{enumerate}
	\item If $K \in M(\mathcal{P}_k)$, then $K_+ \in M(\mathcal{P}_+)$.
	\item If $K \in M(\mathcal{P}_k)$ factors on scale $k$, then $K_+$ factors on scale $k+1$.
	\item The map $(H,K) \mapsto K_+$ satisfies the restriction property, that is
for $U \in \mathcal{P}_{k+1}$ the value of $K_+(U)$ depends on $U$ only via the restriction $K \big\vert_{U^*}$ of $K$ to polymers in $\mathcal{P}(U^*)$.
\item Let $(K^{\Lambda})_{\Lambda}$ satisfy the $(\Z^d)$-property and let $H \in M_0(\mathcal{B})$. Then $(\mathbf{S}^{\Lambda}(H,K,q))_{\Lambda}$ also satisfies the $(\Z^d)$-property. \label{Prop:Zd-property}
\end{enumerate}
\end{lemma}

Now we sketch the extension of the map $(H,K)$ to infinite volume.

~\\
Let $\mathcal{B}_k(\Z^d)$ be the set of all $k$-blocks in $\Z^d$ and $\mathcal{P}_k(\Z^d)$ be the set of all finite unions of $k$-blocks. Since we are dealing with boxes $\Lambda$ of varying side length $L^N$ let us introduce the notation $N(\Lambda)$ for the exponent describing the side length of the box~$\Lambda$.

~\\
A relevant functional $H \in M_0(\mathcal{B}_k)$ can easily be thought of as an element dependent on a block living in $\Z^d$ instead of $\Lambda$ due to translation invariance. More precisely, given $H \in M_0(\mathcal{B}_k(\Lambda))$, we define $H^{\Z^d}$ on a block $B \in M_0(\Z^d)$ as $H(B)$ for a translation of $B$ to the fundamental domain of $\Lambda$ and suppress the index $\Z^d$ as well as the translation of the block in the notation.

~\\
The irrelevant part is extended as follows.
Let $(K^{\Lambda})_{\Lambda}$ be a family of maps which satisfy the $(\Z)^d$-property.
 For $X \in \mathcal{P}_k(\Z^d)$ choose $\Lambda$ large enough such that $k< N(\Lambda)$ and $\diam (X) \leq \frac{1}{2} \diam(\Lambda)$. Then we define
\begin{align*}
K^{\Z^d}(X) = K^{\Lambda}(X).
\end{align*}

Here we use that $X \in \mathcal{P}_k(\Z^d)$ has a straight-forward analogon in $\mathcal{P}_k(\Lambda)$ if $\Lambda$ is large enough which we do not record in the notation.

~\\
Given $(H,K^{\Z^d})$ and the finite-volume maps $\left(\mathbf{S}^{\Lambda}\right)_{\Lambda}$, we define $K_+^{\Z^d}$ as follows.
For $U \in \mathcal{P}_{k+1}(\Z^d)$ choose $\Lambda$ large enough such that $k+1 < N(\Lambda)$ and $\diam (U) \leq \frac{1}{2} \diam (\Lambda)$. Then
\begin{align*}
K_+^{\Z^d}\left(H,K^{\Z^d}\right)(U)
= \mathbf{S}^{\Lambda}\left(H,K^{\Lambda} \vert_{U^*}\right).
\end{align*}

Defining the relevant flow in infinite volume is straightforward: Fix $B \in \mathcal{B}_{k+1}\left(\Z^d\right)$ and $\left(H,K^{\Z^d}\right)$. Define
$$
H^{\Z^d}_+(B) = \mathbf{A}^q H (B) + \mathbf{B}^q K^{\Z^d} (B).
$$
As before we can skip the index $\Z^d$ on $H$.

~\\
Now we extend the norms.
There is no need to change the norm for the relevant variable since it does not depend at all on the size of the torus.
For the irrelevant variable let $X \in \mathcal{P}_k^c(\Z^d)$ and choose $\Lambda$ large enough such that $\diam(X) \leq \frac{1}{2} \diam(\Lambda)$. Then $K^{\Z^d}(X) = K^{\Lambda}(X)$ and we can use the same definition as in \cite{ABKM} for
$$
\Big\Vert K^{\Z^d} (X) \Big\Vert_k = \Big\Vert K^{\Lambda} (X) \Big\Vert_k
= \sup_{\varphi \in \mathcal{V}(X^*)}w_k^{-X}(\varphi) \vert K(X,\varphi) \vert_{k,X,T_{\varphi}}.
$$

\subsubsection{Existence of the global and finite volume flow}

First we cite the statement concerning the existence of the global flow.

\begin{prop}[Proposition 3.18 and 3.19 in \cite{Hil19_1}]
Fix $\zeta,\eta \in (0,1)$. There is $L_0$ such that for all integers $L \geq L_0$ there is $A_0,h_0$ and $\kappa$ with the following property.
Given $\epsilon > 0$ there exist $\epsilon_1 > 0$ and $\epsilon_2 > 0$ such that for each $(\mathcal{K},\mathcal{H},q) \in B_{\epsilon_1}(0)\times B_{\epsilon_2} (0)\times B_{\kappa}(0) \subset \mathbf{E} \times M_0(\mathcal{B}_0) \times \R^{(d \times m) \times (d \times m)}_{\text{sym}}$ there exists a unique global flow $\left(H_k,K_k^{\Z^d}\right)_{k \in \N}$ such that
\begin{align*}
\Vert H_k \Vert_{k,0} ,\, \left\Vert K^{\Z^d}_k\right\Vert_k^{(A)} \leq \epsilon \eta^k  \quad \text{for all } k \in \N_0,
\end{align*}
with initial condition given by
$$
K^{\Z^d}_0(X,\varphi) = e^{-\mathcal{H}(X,\varphi)} \prod_{x \in X} \mathcal{K}(\nabla \varphi(x))
$$
and
$$
\left(H_{k+1},K_{k+1}^{\Z^d}\right)
= \mathbf{T}_k^{\Z^d} \left(H_k,K_k^{\Z^d},q\right).
$$
Moreover, the flow is smooth in $(\mathcal{K},\mathcal{H},q)$ with bounds on the derivatives which are independent of $N$ and there is $0 < \delta \leq \epsilon_1$ and a smooth map
$$
\hat{\mathcal{H}}: B_{\delta}(0) \subset \mathbf{E} \rightarrow B_{\epsilon_2}(0) \subset M_0(\mathcal{B}_0)
$$
such that
$$
H_0(\hat{\mathcal{H}}(\mathcal{K}), \mathcal{K}) = \hat{\mathcal{H}}(\mathcal{K})
$$
and $q(\hat{\mathcal{H}(\mathcal{K})}) \subset B_{\kappa}(0)$ for all $\mathcal{K} \in B_{\delta}(0)$. Moreover, the derivatives of $\hat{\mathcal{H}}$ can be bounded uniformly in $N$.
\end{prop}

 Now we cite the existence of and estimates on the finite volume flow.
For fixed $\eta$ and $\rho_0$, let us introduce the space
\begin{align}
&\mathbb{D}_k(\rho_0,\eta,\Lambda)
\nonumber
\\ & \quad
= 
\left\lbrace (H,K) \in M_0(\mathcal{B}_k) \times M(\mathcal{P}_k(\Lambda)):
H \in B_{\rho_0 \eta^k}(0),
K \in B_{\rho_0 \eta^{2k}}(0)
\right\rbrace.
\label{eq:defn_D_bulk}
\end{align}
\begin{prop}[Proposition 3.21 in \cite{Hil19_1}]\label{Prop:Existence_BulkFlow}
Fix $\zeta,\eta \in (0,1)$. There is $L_0$ such that for all odd integers $L \geq L_0$ there is $A_0,h_0,\kappa$ with the following property. There is $\bar{\delta}$ and $\bar{\epsilon}$ such that for a fixed $\Lambda$ the finite-volume flow
$$
(H_k,K_k^{\Lambda}) \mapsto (H_{k+1},K_{k+1}^{\Lambda})
$$
exists for all $k \leq N(\Lambda)$, is smooth in $\mathcal{K}\in B_{\bar{\delta}}(0)$ with bounds which are uniform in $N(\Lambda)$ and satisfies $(H_k,K_k^{\Lambda}) \in \mathbb{D}_k(\bar{\epsilon},\eta,\Lambda)$.

Moreover,
$$
\Pi_2(H_0(\mathcal{K})) = q(\mathcal{K})
$$
and
$$
K_0(\varphi,X) = K_0(\mathcal{K},H_0)(\varphi,X) = e^{H_0(\varphi,X)}\prod_{x \in X} \mathcal{K}(\nabla\varphi(x)).
$$
\end{prop}

\section{RG analysis for the observable flow}
\label{sec:RG-analysis_ObservableFlow}

This section is dedicated to the proof of Theorem \ref{Thm:RepresentationExtendedPartitionFunction}. The theorem contains a representation of the partition function with inserted observables $s \nabla_{m_a} \varphi(a)$ and $t \nabla_{m_b} \varphi(b)$. In order to work with such a singular external field we extend the analysis of Section \ref{sec:RG-analysis_BulkFlow}. This will truly be an extension in the sense that the bulk flow needs no modification. We will show how observables can be incorporated into the analysis to obtain the pointwise asymptotic formula in Theorem \ref{Thm:RepresentationExtendedPartitionFunction}.

We will follow the flow of these observables in detail and study the corresponding properties.
First we extend spaces and norms in Subsection \ref{sec:extension_spaces_norms}. In Subsection \ref{sec:Extension_RG_map} the RG map is defined. We have to provide a good definition for the flow such that we can extract the Gaussian covariance $\mathcal{C}^q$.
 This is achieved by using second order perturbation in the map $\mathbf{A}$ instead of a first order expansion as before. 

The proof of Theorem \ref{Thm:RepresentationExtendedPartitionFunction} consists of two steps. A first estimate on the covariance is proven in Subsection \ref{subsec:first_Cov_bound}, a refined one in Subsection \ref{subsec:refined_estimate}. The proof of Theorem \ref{Thm:RepresentationExtendedPartitionFunction} is then immediate from these estimates (see Subsection \ref{Subsec:Proof_Decay}).

~\\
Remember that we aim to obtain a representation of
$$
\mathcal{Z}_{N}(u,f_{ab}), \quad \text{where}
\quad
f_{ab} = s \nabla^*_{m_a} \1_a + t \nabla^*_{m_b}\1_b.
$$
Let $(H_k, K_k)$ be the bulk flow of the last section. We can rewrite $\mathcal{Z}_{N}(u,f_{ab})$ as follows:
\begin{align*}
\mathcal{Z}_{N}(u,f_{ab})
&= \int e^{(\varphi,f_{ab})} \sum_{X \subset \Lambda_N} \prod_{x \in X} \mathcal{K}(\nabla \varphi(x)) \mu_1(\de\varphi)
\\
&= 
\frac{Z_N^{\left(q(\mathcal{K})\right)}}{Z_N^{(0)}} e^{-L^{Nd} \lambda(\mathcal{K})}
\int e^{(\varphi,f_{ab})}
\left(
e^{H_0}\circ K_0
\right)(\Lambda_N,\varphi) \mu_{\mathcal{C}^q}(\de\varphi).
\end{align*}
We include $(\varphi,f_{ab})$ into the circ product and extend the maps $H_0$ and $K_0$ to
\begin{align*}
H_0^{\ext}(\varphi) &= H_0(\varphi) + s \nabla_{m_a}\varphi(a) \1_a + t \nabla_{m_b} \varphi(b) \1_b,
\\
K_0^{\ext}(\varphi) &= K_0(\varphi) e^{s \nabla_{m_a}\varphi(a) \1_a + t \nabla_{m_b} \varphi(b) \1_b}.
\end{align*}
Then
$$
\mathcal{Z}_{N}(u,f_{ab})
= \frac{Z_N^{\left(q(\mathcal{K})\right)}}{Z_N^{(0)}} e^{-L^{Nd} \lambda(\mathcal{K})}
\int e^{H_0^{\ext}} \circ K_0^{\ext}(\Lambda_N,\varphi) \mu_{\mathcal{C}^q}(\de\varphi).
$$

We want to follow the relevant observable flow explicitly in order to extract the Gaussian covariance $C^q(a,b)$. For this purpose we extend the space of functionals of the bulk flow to these observables. We introduce extended norms, where the observable part is weighted by a carefully chosen weight $l_{\obs,k}$, see Definition \ref{Def:ObservableNormWeight} and the motivation in Remark \ref{Rem:MotivationObservableNormWeight}.
In order to gain the factor $\nabla^* \nabla C^q(a,b)$ in every step we define the flow
$$
(H^{\ext}, K^{\ext}) \mapsto H_+^{\ext}= \mathbf{A}H^{\ext} + \mathbf{B}K^{\ext}
$$
such that second order perturbation is reflected in the observable part of the map~$\mathbf{A}$. Then the observable part of $H^{\ext}$ appears in $K_+^{\ext}$ only to third order (see Proposition \ref{Prop:2nd_order_perturbation_effect_on_S}) which leads to a refined single step estimate (Proposition \ref{Prop:SingleStepRG_Observables}). For the contractivity property of the extended map $(H^{\ext},K^{\ext})\mapsto K_+^{\ext}$ in Proposition \ref{Prop:Contractivity_Observables} the operator $\mathbf{B}$ also has to be adjusted.

~\\
Roughly speaking, the flow then satisfies estimates  which result in a leading term
$$
(1 + S^a)(1+S^b) \nabla^*_{m_b} \nabla_{m_a} C^q(a,b)
$$
in the covariance, see Proposition \ref{Prop:Cov_first}.

In order to show that $S^a, S^b$ do not contribute to the leading order but only at order $\frac{1}{|a-b|^{d + \nu}}$ we will have to perform an additional step: we consider the flow with just one observable in infinite volume and compare a smoothed version to the result on the scaling limit (Proposition \ref{Prop:Cov_better}). Finally Proposition \ref{Prop:Cov_better} together with Proposition~\ref{Prop:Cov_first} will result in the proof of Theorem \ref{Thm:RepresentationExtendedPartitionFunction}.

~\\
From this point on we use the following \textbf{change of notation}: quantities which belong to the bulk flow will get an superscript $\es$. Consequently, the bulk flow becomes $\left(H_k^{\es}, K_k^{\es}\right)$. The superscript "$\ext$" which was used in the motivation above will disappear in most cases, so $(H_k,K_k)$ will denote the extended flow.

\subsection{Extension of functionals, spaces and norms}
\label{sec:extension_spaces_norms}


\subsubsection{Extended spaces}

As before, let $\mathcal{N}^{\es} = C^{r_0}(\chi_N,\R)$ be the space of real-valued functions of fields having at least $r_0$ continuous derivatives. We are interested in functions not only of $\varphi \in \chi_N$ but also of $s$ and $t$, but only in the dependence up to terms of the form $1,s,t,st$. We formalise this via the introduction of a quotient space, in which two functions of $\varphi,s,t$ become equivalent if their formal power series in the observable fields agree to order $1,s,t,st$, as follows.

Let $\tilde{\mathcal{N}}$ be the space of real-valued functions of $\varphi,s,t$ which are $C^{r_0}$ in $\varphi$ and $C^{\infty}$ in $s,t$. Consider the elements of $\tilde{\mathcal{N}}$ whose formal power series expansion to second-order in the external fields $s,t$ is zero. These elements form an ideal $\mathcal{I}$ in $\tilde{\mathcal{N}}$, and the quotient algebra $\mathcal{N} = \tilde{\mathcal{N}}/\mathcal{I}$ has a direct sum decomposition
$$
\mathcal{N} = \mathcal{N}^{\es} \oplus \mathcal{N}^{a} \oplus \mathcal{N}^{b} \oplus \mathcal{N}^{ab}.
$$
The elements of $\mathcal{N}^{a}, \mathcal{N}^{b}, \mathcal{N}^{ab}$ are given by elements of $\mathcal{N}^{\es}$ multiplied by $s$, by $t$ and by $st$ respectively. As functions of the observable field, elements of $\mathcal{N}$ are then identified with polynomials of degree at most $2$. For example, we identify $e^{s \nabla \varphi(a) + t \nabla \varphi(b)}$ and $1 + s \nabla \varphi(a) + t \nabla \varphi(b) + st \nabla \varphi(a) \nabla \varphi(b)$, as both are elements of the same equivalence class in the quotient space. An element $F \in \mathcal{N}$ can be written as
$$
F = F^{\es} + s F^a + t F^b + st F^{ab},
$$
where $F^{\alpha} \in \mathcal{N}^{\es}$ for each $\alpha \in \lbrace \es,a,b,ab \rbrace$. We define projections $\pi^{\alpha}: \mathcal{N} \rightarrow \mathcal{N}^{\alpha}$ by $\pi^{\es}F = F^{\es}$, $\pi^{a}F = sF^{a}$, $\pi^{b}F = t F^{b}$ and $\pi^{ab}F = st F^{ab}$.

Furthermore, let $\pi^* F = \pi^a F + \pi^b F + \pi^{ab} F$ be the projection to the observable part.

~\\
The class of functionals we are going to work with is
\begin{align*}
M^{\ext}(\mathcal{P}_k,\mathcal{V}_N)
= & \left\lbrace
 F: \mathcal{P}_k  \rightarrow \mathcal{N}
 \,\big\vert\, F^{\alpha}(X) \in M(\mathcal{V}_N) \text{ for all } X \in \mathcal{P}_k \text{ and } \alpha \in \lbrace \es,a,b,ab\rbrace,
   \right.
 \\ & \quad\quad\quad\quad\quad\quad\quad \left.
  \pi^{\es}F \in M(\mathcal{P}_k), \pi^* F \text{ shift invariant and local}
    \right\rbrace.
\end{align*}

Note that $\pi^*F$ is not required to be translation invariant.

As in the case of bulk functionals we have immediate generalisations to $M^{\ext}(\mathcal{P}_k^c)$, $M^{\ext}(\mathcal{S}_k)$ and $M^{\ext}(\mathcal{B}_k)$.

~\\
We define the \textit{coalescence scale} 
\begin{align}
j_{ab} = \Big\lfloor \log_L ( 2|a-b| ) \Big\rfloor.
\end{align}
Since by definition
$$
\frac{L^k}{2} \leq |a-b| \quad \text{for all } k \leq j_{ab},
$$
it holds that
\begin{align}
\nabla^*_j \nabla_i C_k(a,b) = 0
\quad \text{for all } k \leq j_{ab},
\quad i,j \in \lbrace 1, \ldots, d \rbrace,
\end{align}
due to the finite-range property of the covariance decomposition.

~\\
The extended space of relevant Hamiltonians $M_0^{\ext}(\mathcal{B}_k) \subset M^{\ext}(\mathcal{B}_k)$ consists of all functionals of the form
$$
H(B,\varphi) = H^{\es}(B,\varphi) + s H^a(B,\varphi) + t H^b(B,\varphi) + st  H^{ab}(B,\varphi)
$$
where
\begin{align*}
H^{\alpha} (B,\varphi) &= \1_{\alpha \in B} \left( \lambda^{\alpha} + \sum_{i=1}^d n^{\alpha}_i \nabla_i \varphi(\alpha) \right),
\quad
\lambda^{\alpha} \in \R, n^{\alpha} \in \R^d,
\quad \alpha \in \lbrace a,b \rbrace,
\\
H^{ab} (B,\varphi) &= \1_{a,b \in B} \, q^{ab},
\quad q^{ab} \in \R.
\end{align*}

We also define a subspace where no constants appear in the observable part: Let
\begin{align*}
\mathcal{V}^{(0)}_k = \lbrace H \in M_0^{\ext}(\mathcal{B}_k): \lambda^a = \lambda^b = q^{ab} = 0 \rbrace,
\end{align*}
so $H \in \mathcal{V}_k^{(0)}$ is of the form
$$
H(\varphi) = H^{\es}(\varphi) + s n^a \nabla \varphi(a)\1_a + t n^b \nabla \varphi(b) \1_b,
\quad n^a, n^b \in \R^d.
$$
Here the scalar product on $\R^d$ is hidden in the notation,
$$
n^{\alpha} \nabla\varphi(\alpha)
= \sum_{i=1}^d n_i^{\alpha} \nabla_i \varphi(\alpha).
$$

\subsubsection{Extended norms}


\begin{defn}\label{Def:ObservableNormWeight}
Let $h_k = 2^k h$ and $l_k = L^{-\frac{d}{2}k}h_k$. For a fixed $\eta \in (0,1)$ set $g_k = \eta^k$. Fix $\rho_0 > 0$. We define the observable weight $l_{\obs,k}$ by
\begin{align*}
l_{\obs,k} &= \rho_0 g_k 2^{-k} 4^{(k-j_{ab})_+} L^{\frac{d}{2}(k \wedge j_{ab})}.
\end{align*}
\end{defn}
The parameter $\rho_0$ will be determined a-posteriori in Proposition \ref{Prop:Existence_ObservableFlow}.

~\\
In the following we provide a brief motivation for the choice of $l_{\obs,k}$. A more detailed discussion can be found in Remark \ref{Rem:MotivationObservableNormWeight}.
\begin{itemize}
\item 
The sequence $h_k$ is a scaling factor in the norm for the fields, see Subsection~\ref{subsubsec:polymers_functionals_norms}. It has the effect that in norm  $s \nabla \varphi(a) \approx l_{\obs,k} l_k$, where the growing factor $2^k$  appears on the right hand side in $l_k$. This term is eliminated by $2^{-k}$ in~$l_{\obs,k}$.
\item
$4^{(k-j_{ab})_+}$ makes a sum converging at the end of the analysis;
\item
$L^{\frac{d}{2}(k \wedge j_{ab})}$ gives the desired decay since $\left( L^{\frac{d}{2} j_{ab}} \right)^2 = \left( L^{j_{ab}} \right)^d \approx \frac{1}{|a-b|^d}$;
\item
$g_k$ makes sure that the observables live in decreasing balls.
\end{itemize}

Note that
\begin{align} \label{scale_transformation}
\frac{l_{\obs,k+1}}{l_{\obs,k}} =
\begin{cases}
\frac{\eta}{2} L^{d/2}  & \text{ if } k \leq j_{ab}-1, \\
2 \eta & \text{ else }.
\end{cases}
\end{align}

We set, for $F \in M^{\ext}(\mathcal{P}_k)$,
\begin{align*}
\big\vert F(X,\varphi) \big\vert^{\ext}_{k,X,T_{\varphi}} = \sum_{\alpha \in \lbrace \es, a,b,ab \rbrace} \big\vert F^{\alpha}(X,\varphi) \big\vert_{k,X,T_{\varphi}} l^{|\alpha|}_{\obs,k}
\end{align*}
where, with a slight abuse of notation, $|\es |=0 ,|a| = |b| = 1$ and $|ab| = 2$.
The norms $\Vert \cdot \Vert_{k,X}^{\ext}$, $\Vert \cdot \Vert_{k:k+1,X}^{\ext}$,  $\vertiii{\cdot}_{k,X}^{\ext}$ and $\Vert \cdot \Vert_{k}^{(A),\ext}$ on functionals $F \in M^{\ext}(\mathcal{P}_k^c)$ are defined as before in Section \ref{subsubsec:polymers_functionals_norms}.

The norm on $M_0(\mathcal{B}_k)$ is extended to  $M_0^{\ext}(\mathcal{B}_k)$ as follows. Recall that we defined elements of $M_0^{\ext}(\mathcal{B}_k)$ to be functionals of the form
\begin{align*}
&H(\varphi) 
= H^{\es}(\varphi)
+ s \1_a
 \Big( \lambda^a + \sum_i n_i^a \nabla_i \varphi(a) \Big)
+ t \1_b
 \Big( \lambda^b + \sum_i n_i^b \nabla_i \varphi(b) \Big)
+ st \1_{a,b} \, q^{ab}.
\end{align*}
Then
\begin{align*}
\left\Vert H \right\Vert_{k,0}^{\ext} 
= \left\Vert H^{\es} \right\Vert_{k,0}
+ l_{\obs,k} 
\left( 
\left|\lambda^a \right| 
+ l_k \sum_{i=1}^d \left|n^a_i \right| + \left|\lambda^b \right| + l_k \sum_{i=1}^d \left|n^b_i \right| 
\right) 
+ l_{\obs,k}^2 \left|q^{ab} \right|.
\end{align*}
We will use the following notation:
\begin{align*}
\left\Vert H^{\alpha} \right\Vert_{k,0}^{\alpha} 
&= l_{\obs,k} \left( \left|\lambda^{\alpha}\right| + l_k \sum_{i=1}^d \left|n^{\alpha}_i \right| \right)
\quad \text{for } \alpha \in\lbrace a,b \rbrace,
\\
\left\Vert H^{ab} \right\Vert_{k,0}^{ab} 
&= l_{\obs,k}^2 \left| q^{ab} \right|.
\end{align*}

\subsection{Extension of the renormalisation map}
\label{sec:Extension_RG_map}

\subsubsection{Definition of the extended map}

The goal of this section is the definition and preliminary study of the extended renormalisation map
\begin{align*}
\mathbf{T}_k^{\ext}: \R^3 \times \mathcal{V}^{(0)}_k \times M^{\ext}(\mathcal{P}_k^c)
&\quad\rightarrow\quad
\R^3 \times \mathcal{V}^{(0)}_{k+1} \times M^{\ext}(\mathcal{P}_{k+1}^c),
\\
(\lambda^a,\lambda^b,q^{ab}, H, K )
&\quad\mapsto\quad (\lambda^a_+,\lambda^b_+,q^{ab}_+, H_+, K_+).
\end{align*}

Initially, we extend the operator $\mathbf{B}_k$:
$$
\mathbf{B}_k: M^{\ext}\left(\mathcal{P}_k^c\right) \rightarrow M_0^{\ext}\left(\mathcal{B}_{k+1}\right),
\quad
\mathbf{B}_k K(B_+) = \sum_{B\in \mathcal{B}_k(B_+)} \Pi_k \mathcal{R}_{k+1} K(B)
$$
where $\Pi_k$ is the scale-dependent localisation operator
\begin{align*}
\Pi_k:  M^{\ext}(\mathcal{B}_k) \rightarrow M_0^{\ext}(\mathcal{B}_k),
\quad
\Pi_k F = \Pi_2 F^{\es} + \1_a \Pi^a_k F^a + \1_b \Pi^b_k F^b + \1_{ab}\Pi_0 F^{ab},
\end{align*}
$\Pi^{\alpha}_k$ defined explicitly below in Section \ref{subsubsec:Projection}. Roughly speaking, for $\alpha \in \lbrace a,b \rbrace$,
\begin{align*}
\Pi_k^{\alpha} = \begin{cases}
\Pi_1 & \quad \text{if }  k < j_{ab},
\\
\Pi_0 & \quad \text{if }  k \geq j_{ab}.
\end{cases}
\end{align*}
Similar to the definition of $\Pi_2$ in the bulk flow case (see Section \ref{subsubsec:The_renormalisation_map}),
$$
\Pi_0F(\varphi)=F(0),
\quad \text{and} \quad
\Pi_1^{\alpha} F(\varphi) = F(0) + l^{\alpha}(\varphi)
$$ where $l^{\alpha}(\varphi)$ is the unique map of the form $l^{\alpha}(\varphi) = \sum_j n^{\alpha}_j \nabla_j\varphi(\alpha)$ which coincides with $DF(0)(\varphi)$ for all functions $\varphi$ which are on $(B_{\alpha}^*)^*$ of the form
$$
\varphi(x) = \sum_{i} m_i (x_i-\alpha_i), 
\quad m \in \R^d.
$$

This implies that in $(\mathbf{B}_k K)^{ab}$ only the zeroth order polynomial remains after projection whereas in the $a$- and $b$-part of $\mathbf{B}_kK$ we follow the linear flow up to the scale~$j_{ab}$ but not further.

Note that $\mathbf{B}_k$ is a linear operator, so $\left( \mathbf{B}_k K \right)^{\alpha} = \mathbf{B}_k \left( K^{\alpha} \right)$.

Let us introduce the following notation: For $\alpha \in \lbrace a,b \rbrace$, we denote the constant and linear coefficients of $\mathbf{B}_k K^{\alpha}$ by
$$
 \mathbf{B}_k K^{\alpha}
=
\left( \mathbf{B}_k K^{\alpha}\right)^0
+
\sum_{i=1}^d
\left( \mathbf{B}_k K^{\alpha} \right)^1_i
\nabla_i \varphi(\alpha).
$$

~\\
Now we can give a definition of the map
$$
\mathbf{T}_k^{\ext}:
(\lambda^a,\lambda^b,q^{ab}, H,K)
\mapsto
(\lambda^a_+,\lambda^b_+,q^{ab}_+, H_+, K_+).
$$
Namely,
\begin{align*}
&\lambda^{\alpha}_+
= \lambda^{\alpha} + \left( \mathbf{B}_k K^{\alpha}\right)^0, \quad \alpha \in \lbrace a,b \rbrace,
\\
&q_+^{ab}
= q^{ab} + \mathbf{B}_k K^{ab} + \int H^a H^b \de \mu_{k+1},
\\
&\left(H_+\right)^{\es}
= \left( H^{\es} \right)_+,
\quad
H_+^{\alpha}
= H^{\alpha} + \left( \mathbf{B}_k K^{\alpha}\right)^1 \nabla\varphi(\alpha), \quad \alpha \in \lbrace a,b \rbrace,
\end{align*}
and the irrelevant $K_+$ is defined by
\begin{align*}
K_+ = e^{- s\left(\mathbf{B}_k K^a\right)^0 - t\left(\mathbf{B}_k K^b\right)^0 - st\left(\int H^a H^b \de \mu_{k+1} + \mathbf{B}_k K^{ab}\right)} \mathbf{S}_k(H,K),
\end{align*}
where $\mathbf{S}_k$ is the map from the bulk flow, defined in \eqref{eq:defn_K+}. Let us denote
\begin{align*}
\mathbf{S}^{\ext}_k (H,K) = e^{-s\left(\mathbf{B}_k K^a\right)^0 - t\left(\mathbf{B}_k K^b)\right)^0 - st\left(\int H^a H^b \de \mu_{k+1} + \mathbf{B}_k K^{ab}\right)} \mathbf{S}_k(H,K).
\end{align*}
Moreover, let us combine the definitions above into the map $\mathbf{A}_k$,
\begin{align*}
&\mathbf{A}_k: \mathcal{V}_k^{(0)} \rightarrow M^{\ext}_0(\mathcal{B}_{k+1}),
\quad
\mathbf{A}_k H = \mathbf{A}_k H^{\es} + \mathbf{A}_k H^{\obs},
\\
&\mathbf{A}_k H^{\es}(B_+) = \sum_{B \in \mathcal{B}_k(B_+)} \Pi_2\mathcal{R}_{k+1} H^{\es}(B),
\\
&\mathbf{A}_k H^{\obs} = s H^a + t H^b + st \int H^a H^b \de\mu_{k+1}.
\end{align*}

\begin{remark}
We are no longer interested in the dependence of the maps on the parameter $q$ since we will fix the bulk flow obtained in the last section - with the caveat that the choice of $\kappa$ in $q \in B_{\kappa}(0)$ depends on the choice of $L$ which will be chosen larger than in \cite{ABKM}.
\end{remark}

In the next lemma we show that the map $\mathbf{T}_k$ is well-defined, and we state first properties.
A motivation for the definition of $\mathbf{T}_k$ follows afterwards in Remark \ref{Remark:MotivationT}.

~\\
Let $K \in M^{\ext}(\mathcal{P}_k)$ satisfy \textit{field locality} if for $\alpha \in \lbrace a,b,ab \rbrace$ and for any $X \in \mathcal{P}_k$, $K^{\alpha}(X)=0$ unless $\alpha \in X$. Here we use the notation $ab \in X$ which means $a \in X$ and $b \in X$.

\begin{lemma}\label{Lemma:FirstConclusions_Observables}
Fix $(\lambda^a, \lambda^b, q^{ab},H,K) \in \R^3 \times \mathcal{V}^{(0)}_k \times M^{\ext}(\mathcal{P}_k^c)$. Then the map $\mathbf{T}_k^{\ext}$ defined above satisfies the following properties.
\begin{enumerate}
	\item
	$K_+ \in M^{\ext}(\mathcal{P}_{k+1}^c)$, and the map $\mathbf{S}_k^{\ext}$ satisfies the restriction property and preserves the $(\Z^d)$-property as well as field locality.
	\item
	If $K$ satisfies field locality, then $H_+ \in \mathcal{V}^{(0)}_{k+1}$, i.e., there are $n_+^a, n_+^b \in \R^d$ such that
	$$
	H_+(\varphi) = H_+^{\es}(\varphi) + s n^a_+ \nabla \varphi(a)\1_a + t n_+^b \nabla\varphi(b) \1_b.
	$$
	\item
	Let us denote $\zeta = s \lambda^a + t \lambda^b + st q^{ab}$ and $\zeta_+ s \lambda_+^a + t \lambda_+^b + st q_+^{ab}$. Then
	\begin{align}
	e^{\zeta} \mathcal{R}_{k+1} \left( e^H \circ K \right)
	= e^{\zeta_+} \left( e^{H_+} \circ K_+ \right).
	\label{RG_equation_obs}
	\end{align}
	\item
	If $K$ satisfies field locality, then $H_+^a$ is independent of $H^b$, $K^b$ and $K^{ab}$, and the same holds for $a,b$ interchanged.
	\item
	The observable flow leaves the bulk flow unchanged, i.e.,
	$$
	\left(H_+\right)^{\es}
	= \left( H^{\es} \right)_+,
	\quad
	\left( K_+ \right)^{\es}
	= \mathbf{S}(H^{\es},K^{\es}).
	$$
	\end{enumerate}
\end{lemma}

\begin{proof}
\begin{enumerate}
	\item The definition immediately implies that $K_+ \in M^{\ext}(\mathcal{P}_{k+1}^c)$ and that $\mathbf{S}^{\ext}$ satisfies the restriction property and preserves the $(\Z^d)$-property, since the map $\mathbf{S}$ fulfils the desired properties. The preservation of field locality can be verified by inspection of the definition.
	\item Since $K$ satisfies field locality, it holds that $\mathbf{B}_k K^{\alpha} = \mathbf{B}_k K^{\alpha} \1_{\alpha}$. Thus we can set
	$$
	n_+^{\alpha} = n^{\alpha} + \left( \mathbf{B}_k K^{\alpha} \right)^1
	$$
	and so $H_+ \in \mathcal{V}_{k+1}^{(0)}$.
	\item
	The definition of the map $\mathbf{S}^{\ext}$ is specifically designed so that this integration property holds. Namely, use that in the bulk flow case the maps $\mathbf{A}_k, \mathbf{B}_k$ and $\mathbf{S}_k$ are made such that
$$
e^{(\mathbf{A}_kH + \mathbf{B}_k K)} \circ \mathbf{S}_k(H,K) = \mathcal{R}_{k+1}(e^{H} \circ K).
$$
Then
\begin{align*}
&e^{\zeta} \mathcal{R}_{k+1} (e^{H} \circ K)
\\ & \quad
= e^{\zeta} 
\left[
e^{(\mathbf{A}_k + \mathbf{B}_k)} \circ \mathbf{S}_k(H,K)
\right]
\\ & \quad
= e^{\zeta + s \left(\mathbf{B}_k K^a\right)^0 + t \left(\mathbf{B}_k K^b\right)^0 + st\left( \int H^a H^b \de \mu_{k+1} + \mathbf{B}_kK^{ab}\right)}
\\& \qquad\qquad \times
\left[ e^{H_+} \circ \left(
e^{ -s \left(\mathbf{B}_k K^a\right)^0 - t \left(\mathbf{B}_k K^b\right)^0 - st\left(\int H^a H^b \de \mu_{k+1} +\mathbf{B}_k K^{ab}\right)} \mathbf{S}_k(H,K) \right)
\right]
\\ & \quad
= e^{\zeta_+} \left[ e^{H_+} \circ \mathbf{S}_k^{\ext}(H,K)\right].
\end{align*}
	\item
	Since $H_+^a = H^a + \left( \mathbf{B} K^a \right)^1 \nabla \varphi(a)$ the statement follows straightforwardly by field locality.
\item
Due to the definition of $\mathbf{A}_k$ and $\mathbf{B}_k$, for $H = H^{\es} + \pi^* H$ and $K = K^{\es} + \pi^* K$, it holds that $H^{\es}_+ = \mathbf{A}_kH^{\es} + \mathbf{B}_kK^{\es}$.
\end{enumerate}
\end{proof}

\begin{remark}\label{Remark:MotivationT}
We try to motivate the definition of the map $\mathbf{T}_k^{\ext}$.

In principle we want to define $H_+ = \mathbf{A}_k H + \mathbf{B}_k K$ as before in the bulk flow case through extended maps $\mathbf{A}_k$ and $\mathbf{B}_k$. We perform some changes in the definition of $\mathbf{A}_k$ and $\mathbf{B}_k$.

On the one hand, we want to extract not only to linear but also to quadratic order in $H$, so that we can observe the Gaussian covariance. Heuristically, up to second order in $H$,
$$
\mathcal{R}_+ \left(e^{H}\right)
\approx
1 + \mathcal{R}_+H + \frac{1}{2}\mathcal{R}_+ \left( H^2 \right)
$$
since
\begin{align*}
\mathcal{R}_+ \left(e^{H}\right)
\approx \mathcal{R}_+ \left( 1 + H + \frac{1}{2}H^2 \right)
= 1 + \mathcal{R}_+H + \frac{1}{2}\mathcal{R}_+ \left( H^2 \right)
\end{align*}
and
\begin{align*}
e^{\mathcal{R}_+ H + \frac{1}{2} \mathcal{R}_+(H^2) - \frac{1}{2}(\mathcal{R}_+H)^2}
&\approx 1 + \mathcal{R}_+H + \frac{1}{2}\mathcal{R}_+ \left(H^2 \right)- \frac{1}{2} \left( \mathcal{R_+}H \right)^2 + \frac{1}{2}\left( \mathcal{R}_+ H \right)^2
\\
&=  1 + \mathcal{R}_+H + \frac{1}{2}\mathcal{R}_+ \left( H^2 \right).
\end{align*}
Given $H \in \mathcal{V}_k^{(0)}$ with
\begin{align*}
&H^{\obs} = s H^a + t H^b,
\\
&
H^a(\varphi) = n^a \nabla \varphi(a) \1_a, \quad
H^b (\varphi) = n^b \nabla \varphi(b) \1_b, \quad
n^a,n^b \in \R^d,
\end{align*}
then, up to first order in $s$, $t$ and $st$,
\begin{align*}
&\mathcal{R}_{k+1} H^{\obs} + \frac{1}{2} \mathcal{R}_{k+1} \left( (H^{\obs})^2 \right) - \frac{1}{2} \left( \mathcal{R}_{k+1} (H^{\obs}) \right)^2
\\ & \quad
=
s H^a + t H^b + st \int H^a H^b \de \mu_{k+1}.
\end{align*}
Since
$$
\int H^a H^b \de \mu_{k+1} = n^a n^b \nabla^* \nabla C_{k+1}(a,b),
$$
we explicitly observe a part of the Gaussian covariance.
This motivates the definition of the map $\mathbf{A}_k$ given above. Note that the map is no longer linear, unlike in the bulk flow case.

On the other hand, the map $\mathbf{B}_k$ extracts as much from $\mathcal{R}_+K$ as is needed in order to have a contraction in the irrelevant part. In the case of observables it is enough to extract the linear order up to coalescence scale $j_{ab}$ and only the constant order above.

In a last step in the definition of the map $(H,K)\mapsto H_+$ we extract constant observable parts which arise by the application of the maps $\mathbf{A}_k$ and $\mathbf{B}_k$. We put them out of the circ product into $\zeta_+$.

The irrelevant part $K_+$ is defined such that \eqref{RG_equation_obs} holds.

\end{remark}

~\\
Let us denote by $B_a \in \mathcal{B}_k$ and $B_b \in \mathcal{B}_k$ the block at scale $k$ which contains $a$ and $b$, respectively. By definition of the coalescence scale $j_{ab}$,
$$
\frac{L^{j_{ab}-1}}{2}
< \frac{L^{j_{ab}}}{2} \leq |a-b|
< \frac{L^{j_{ab}+1}}{2}.
$$
For simplicity let us assume that there is $B \in \mathcal{B}_{j_{ab}}$ such that $a,b \in B$, but $B_a, B_b \in \mathcal{B}_{j_{ab}-1}$ are disjoint as in the following picture. All other cases can be done similarly.

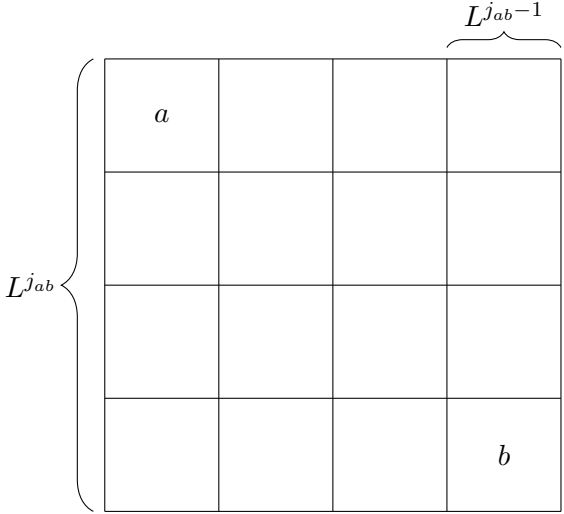
\vspace{0.3cm}
\begin{tikzpicture}[scale=1.5]
	\draw[decoration= {brace, amplitude = 12 pt}, decorate] (0.9, 0) -- (0.9,4);
	\node at (0.35,2) {$L^{j_{ab}}$};
	\draw (1,0) -- (1,4);
		\node at (4.5,0.5) {$b$};
	\draw (2,0) -- (2,4);
	\draw (3,0) -- (3,4);
	\draw (4,0) -- (4,4);
		\node at (1.5,3.5) {$a$};
	\draw (5,0) -- (5,4);
	\draw (1,0) -- (5,0);
	\draw (1,1) -- (5,1);
	\draw (1,2) -- (5,2);
	\draw (1,3) -- (5,3);
	\draw (1,4) -- (5,4);
		\draw[decoration= {brace, amplitude = 6 pt}, decorate] (4,4.1) -- (5,4.1);
		\node at (4.5,4.4) {$L^{j_{ab} - 1}$};
\end{tikzpicture}

~\\

\begin{lemma}\label{lemma:Further_properties_Observables}
For initial coupling constants $\lambda^a_0 = \lambda^b_0 = q^{ab}=0$, $n_0^a, n_0^b \in \R^d$ 
we obtain the following formulas for the coupling constants:
\begin{enumerate}
	\item
	$\lambda_k^{\alpha} = \sum_{l=0}^{k-1} \left( \mathbf{B}_l K_l^{\alpha} \right)^0$,
	\item
	$q_k^{ab} = 0$ for $k \leq j_{ab}$ and
	$$
	q_k^{ab}= \sum_{l=j_{ab}}^{k-1} \left( \mathbf{B}_l K_l^{ab} + \int H_l^a H_l^b \de\mu_{l+1} \right),
	\quad \text{for } k > j_{ab},
	$$
	\item
	$n_k^{\alpha} = n_0^{\alpha} + \sum_{l=0}^{(k-1)\wedge (j_{ab}-1)} \left( \mathbf{B}_l K_l^a \right)^1$.
\end{enumerate}
\end{lemma}

\begin{proof}
These formulas follow iteratively by definition of the flow and Lemma \ref{Lemma:FirstConclusions_Observables}.
\end{proof}

In the next statement we will deliver a precise formulation of what was described heuristically in Remark \ref{Remark:MotivationT} when we motivated the definition of the map $\mathbf{A}_k$, namely that the relevant flow absorbs the irrelevant part up to second order.

\begin{prop}\label{Prop:2nd_order_perturbation_effect_on_S}
The $st$-part of the second derivative in direction $H$ of $\mathbf{S}^{\ext}$ is zero:
$$
\left[ D^2_H \mathbf{S}^{\ext}(0,0) (\dot{H},\dot{H}) \right]^{ab} = 0.
$$
\end{prop}

The proof can be found in Lemma \ref{Lemma:2nd_order_perturbation_exact}. 

~\\
At this point, we have obtained that
\begin{align*}
\int e^{H_0} \circ K_0 \,\de\mu_{\mathcal{C}^q}
= e^{\zeta_N} \left( e^{H_N(\varphi=0)} + K_N(\varphi=0) \right),
\quad \zeta_N &= st q_N^{ab} + s\lambda_N^a + t \lambda_N^b.
\end{align*}
Since
$$
\sum_{k=j_{ab}}^{N}C_k(a,b) = \sum_{k=0}^{N}C_k(a,b) = C^q(a,b),
$$
it holds that
\begin{align*}
q_N^{ab} &= \left( n_0^a + S_{j_{ab}}^a \right) \left( n_0^b + S_{j_{ab}}^b \right) \nabla^*_{m_b} \nabla_{m_a} C^q(a,b) + \tilde{R}_{ab},
\\
 S_{j_{ab}}^{\alpha} &= \sum_{l=0}^{j_{ab}-1} \left( \mathbf{B}_l K_l^{\alpha}\right)^1,
\quad \tilde{R}_{ab} = \sum_{l=j_{ab}}^{N-1} \mathbf{B}_l K_l^{ab},
\\
\lambda_N^{\alpha} &= \sum_{l=0}^{N-1} \left( \mathbf{B}_l K_l^{\alpha} \right)^0.
\end{align*}

In the following section we develop estimates on the involved quantities which lead to a first bound on the covariance in Proposition \ref{Prop:Cov_first}. In order to get rid of the $S^{\alpha}_{j_{ab}}$ in the leading term, an additional argument is needed. We implement this by considering the flow of a single observable. The refined bound on the covariance can be found in Proposition \ref{Prop:Cov_better}.

\subsubsection{Estimates on the extended map}

The separation of the bulk flow into relevant and irrelevant directions with corresponding estimates can be extended to the observable flow. 

Let $U_{\rho} \subset \mathcal{V}_k^{(0)} \times M^{\ext}(\mathcal{P}_k^c)$ be the subset
\begin{align*}
U_{\rho} = \lbrace (H,K) \in \mathcal{V}_k^{(0)} \times M^{\ext}(\mathcal{P}_k^c): \Vert H \Vert^{\ext}_{k,0} < \rho, \Vert K \Vert_k^{(A),\ext} < \rho \rbrace.
\end{align*}

\begin{prop}[Smoothness of the extended flow]\label{Prop:Smoothness_Observables}
There exists a constant $L_0$ such that for all odd integers $L \geq L_0$ there is $A_0$ and $h_0$ with the following property. For all $A \geq A_0$ and $h \geq h_0$ there exists $\rho^* = \rho^*(A)$ such that the map $\mathbf{S}^{\ext}_k$ satisfies
\begin{align*}
\mathbf{S}^{\ext}_k \in C^{\infty} \left(U_{\rho^*},M^{\ext}(\mathcal{P}_{k+1}^c)\right).
\end{align*}
For any $j_1, j_2 \in \N$ there is a constant $C^*_{j_1,j_2} = C^*_{j_1,j_2}(L,h,A)$ such that for any $(H,K) \in U_{\rho^*}$
\begin{align*}
\left\Vert
D_H^{j_1} D_K^{j_2}  \mathbf{S}_k^{\ext}(H,K)(\dot{H}^{j_1}, \dot{K}^{j_2})
\right\Vert_{k+1}^{(A),\ext}
\leq C^*_{j_1,j_2} \left( \Vert \dot{H} \Vert_{k,0}^{\ext}\right)^{j_1} \left( \Vert\dot{K} \Vert_k^{(A),\ext}\right)^{j_2}.
\end{align*}
\end{prop}

The proof of this proposition can be found in Section \ref{subsec:Proof_Smoothness}.

~\\
The extended flow also satisfies contraction estimates for the derivative of $\mathbf{S}_k^{\ext}$ at zero.

\begin{prop}[Contractivity of the extended flow] \label{Prop:Contractivity_Observables}
The first derivative of $\mathbf{S}^{\ext}_k$ at $(H,K)=(0,0)$ satisfies
\begin{align*}
D \mathbf{S}^{\ext}_k(0,0) (\dot{H},\dot{K}) = \mathbf{C}_k \dot{K},
\end{align*} 
where
\begin{align*}
\mathbf{C}_k \dot{K} (U,\varphi) = \sum_{\substack{B \in \mathcal{B}_k:\\ \bar{B} = U}} (1 - \Pi) \mathcal{R}_{k+1} \dot{K}(B,\varphi) + \sum_{\substack{X \in \mathcal{P}^c_k \setminus \mathcal{B}_k\\ \pi(X) = U}} \mathcal{R}_{k+1} \dot{K}(X,\varphi).
\end{align*}
For any $\theta \in (0,1)$ there is $L_0$ such that for all odd integers $L \geq L_0$ there is $A_0$ and $h_0$ with the following property. For all $A \geq A_0$, $h \geq h_0$ the following estimate holds independent of $k$ and $N$,
\begin{align*}
\Vert \mathbf{C}_k \Vert \leq \theta.
\end{align*}
The norm on the left hand side denotes the operator norm for the map
$$
\left( M^{\ext}(\mathcal{P}_k^c), \Vert \cdot \Vert_k^{(A),\ext} \right) \rightarrow \left( M^{\ext}(\mathcal{P}_{k+1}^c), \Vert \cdot \Vert_{k+1}^{(A),\ext} \right).
$$
\end{prop}

\begin{proof}
Here we only show the validity of the expression for $\mathbf{C}_k$. The contractivity is shown in Section \ref{subsec:Bound_C}, see Lemma \ref{lemma:contractivity_of_C}.

We claim that
$$
D \mathbf{S}^{\ext}_k(0,0)(\dot{H},\dot{K}) = D \mathbf{S}_k(0,0) (\dot{H},\dot{K}).
$$
Then the expression for $\mathbf{C}_k$ follows just as in the case of the bulk flow, see Proposition~3.13 in \cite{Hil19_1}.
The above equation holds with product rule since $\mathbf{S}_k(0,0) = 0$:
\begin{align*}
&D \mathbf{S}^{\ext} (0,0)(\dot{H},\dot{K}) 
\\
&=D_H \mathbf{S}(0,0)\dot{H}
+
D_H \left( e^{-s(\mathbf{B}K^a)^0 - t(\mathbf{B}K^b)^0 - st(\int H^a H^b \de \mu_+ +\mathbf{B}K^{ab})} \right) \dot{H} \Big\vert_{H=K=0} \mathbf{S}(0,0)
\\&\quad
+
 D_K \left( e^{-s(\mathbf{B}K^a)^0 - t(\mathbf{B}K^b)^0 - st(\int H^a H^b \de \mu_+ +\mathbf{B}K^{ab})} \right) \dot{K}\Big\vert_{H=K=0} \mathbf{S}(0,0)
\\&\quad
 +
  e^0 D_K \mathbf{S}(0,0) \dot{K}.
\end{align*}
\end{proof}

We also state bounds on the map $\mathbf{B}_k$. They are proven in Lemma \ref{lemma:Estimate_for_B}.

\begin{prop}[Bounds on $\mathbf{B}_k$]\label{Prop:Bounds_on_B_ObservableFlow}
The following bounds on the observable part of the map $\mathbf{B}_k$ hold:
\begin{align*}
\big\vert(\mathbf{B}_k K_k^{\alpha})^1 \big\vert
& \leq l_k^{-1} l_{\obs,k}^{-1} \frac{A_{\mathcal{B}}}{2} \Vert K_k \Vert_k^{(A),\ext},
\quad \alpha \in \lbrace a,b \rbrace
\\
\big\vert(\mathbf{B}_k K_k^{\alpha})^0 \big\vert
& \leq l_{\obs,k}^{-1} \frac{A_{\mathcal{B}}}{2} \Vert K_k \Vert_k^{(A),\ext},
\quad \alpha \in \lbrace a,b \rbrace
\\
\big\vert \mathbf{B}_k K_k^{ab} \big\vert
& \leq l_{\obs,k}^{-2} \frac{A_{\mathcal{B}}}{2} \Vert K_k^{ab} \Vert_k^{(A),\ext}.
\end{align*}
\end{prop}

~\\
We can combine Proposition~\ref{Prop:Smoothness_Observables} and \ref{Prop:Contractivity_Observables} and additionally Proposition \ref{Prop:2nd_order_perturbation_effect_on_S} to get a refined single step estimate.

To state it, we extend the space $\mathbb{D}_k(\rho_0,\eta,\Lambda)$ (defined in \eqref{eq:defn_D_bulk}) to observables. In the following definition, $C_{\mathcal{D}}$ is fixed, determined a posteriori in the proof of Proposition~\ref{Prop:Existence_ObservableFlow}. Let
\begin{align*}
&\mathbb{D}^{\ext}_k(\rho_0,g_k,\Lambda)
\\
& \quad=
 \left\lbrace (H,K) \in \mathcal{V}_k^{(0)} \times M^{\ext}(\mathcal{P}_k^c)(\Lambda): 
H \in B_{C_{\mathcal{D}}\rho_0 g_k},
K \in B_{\rho_0 g_k^2},
K^{ab} \in B_{\rho_0 g_k^3}
\right\rbrace.
\end{align*}

\begin{prop}[Single step estimate for the extended flow]\label{Prop:SingleStepRG_Observables}

Fix $\eta \in (0,1)$ and $C_{\mathcal{D}} > 1$. There is $L_0$ such that for all odd integers $L \geq L_0$ there are $A_0$ and $h_0$ with the following property. For $A \geq A_0$ and $h \geq h_0$ there is $\rho_0 > 0$ such that if $(H,K) \in \mathbb{D}_k^{\ext}(\rho_0,g_k,\Lambda)$ then
$$
\Vert \mathbf{S}^{\ext}(H,K,q) \Vert_{k+1}^{(A),\ext} \leq \rho_0 g_{k+1}^2
\quad
\text{and}
\quad
K_{k+1}^{ab} \in B_{\rho_0 g_{k+1}^3}.
$$

\end{prop}

\begin{proof}
Fix $\theta < \eta^3$. Let $L_0$ be large enough such that Proposition \ref{Prop:Smoothness_Observables} and \ref{Prop:Contractivity_Observables} can be applied. Define $C_2^*= \max(C_{2,0}^*,C_{1,1}^*,C_{0,2}^*)$ where $C^*_{j_1,j_2}$ are the constants from Proposition \ref{Prop:Smoothness_Observables}. Choose $\rho_0$ small enough that
$$
C_{\mathcal{D}} \rho_0 \leq \rho^*(A)
\quad
\text{and}
\quad
\theta + \frac{1}{2} C_2^* \rho_0 \left( C_{\mathcal{D}} + 1 \right)^2 \leq \eta^2.
$$
Then $(H,K) \in \mathbb{D}_k^{\ext}(\rho_0,g_k,\Lambda)$ implies $(H,K) \in U_{\rho^*(A)}$ so we can apply Proposition~\ref{Prop:Smoothness_Observables} to estimate as follows.

We expand $\mathbf{S}^{\ext}$ around $(0,0)$ up to linear order,
$$
\mathbf{S}^{\ext} (H,K) = \mathbf{C} K + \int_0^1 D^2 \mathbf{S}^{\ext} (t H,tK) (H,K)^2 (1-t) \de t.
$$
Then
\begin{align*}
&\Vert \mathbf{S}^{\ext} (H,K) \Vert^{(A),\ext}_{k+1}
\\
&\quad\quad\leq
 \theta \Vert K \Vert_k^{(A),\ext} + \frac{1}{2} C_2^* \left( \left(\Vert H \Vert_{k,0}^{\ext}\right)^2 + 2 \Vert H \Vert_{k,0}^{\ext} \Vert K \Vert_k^{(A),\ext} + \left(\Vert K \Vert_k^{(A),\ext} \right)^2 \right)
\\
&\quad\quad\leq 
\rho_0 g_{k+1}^2 \frac{1}{\eta^2}
\left( \theta + \frac{1}{2} C_2^* \rho_0 (C_{\mathcal{D}} + 1)^2 \right)
\leq \rho_0 g_{k+1}^2.
\end{align*}
The last inequality follows by the assumption on $\rho_0$.

~\\
For the improved estimate on the $ab$-part we expand $\mathbf{S}^{\ext}$ up to second order and exploit the fact that we used second order perturbation in the observable flow. With Lemma~\ref{Prop:2nd_order_perturbation_effect_on_S} we obtain
\begin{align*}
K_{+}^{ab} =  
\mathbf{C} K^{ab} 
&
+ 2 \left[D_H D_K \mathbf{S}^{\ext}(0,0)(H,K)\right]^{ab} + \left[D_K^2 \mathbf{S}^{\ext}(0,0)K^2\right]^{ab}
\\
&+  \left[ \frac{1}{2} \int_0^1 D^3 \mathbf{S}^{\ext}(tH,tK)(H,K)^3 (1-t)^2 dt \right]^{ab}.
\end{align*}
Now let $C_3^* = \max (C_{3,0}^*,C_{2,1}^*,C_{1,2}^*,C_{0,3}^*)$ and choose $\rho_0$ such that additionally
$$
\theta + C^*_2 \rho_0 (2 C_{\mathcal{D}}+1) + \frac{1}{6} C^*_3 \rho_0^2(C_{\mathcal{D}} + 1)^3 \leq \eta^3
$$
is satisfied. Then
\begin{align*}
\Vert K_{+}^{ab} \Vert_{k+1}^{(A),\ext}
&\leq \theta \Vert K^{ab} \Vert_k^{(A),\ext}
+ 2 C^*_2 \Vert H \Vert_{k,0}^{\ext} \Vert K\Vert_k^{(A),\ext}
+ C^*_2 \left( \Vert K\Vert_k^{(A),\ext} \right)^2
\\ & \quad\quad\quad
+ \frac{1}{2} \frac{1}{3} C^*_3
\left( \left( \Vert H \Vert_{k,0}^{\ext}\right)^3
+ 3 \left( \Vert H \Vert_{k,0}^{\ext}\right)^2 \Vert K\Vert_k^{(A),\ext}
\right.
\\& \quad\quad\quad
\left. 
+ 3 \Vert H \Vert_{k,0}^{\ext} \left( \Vert K\Vert_k^{(A),\ext} \right)^2 + \left( \Vert K\Vert_k^{(A),\ext} \right)^3 \right)
\\
& \leq
\rho_0 g_{k+1}^3
\frac{1}{\eta^3}
\left( \theta + C^*_2 \rho_0 (2 C_{\mathcal{D}}+1) + \frac{1}{6} C^*_3 \rho_0^2(C_{\mathcal{D}} + 1)^3 \right)
\leq \rho_0 g_{k+1}^3
\end{align*}
and the proof is finished.
\end{proof}

\begin{remark} \label{Rem:MotivationObservableNormWeight}
Here we give some motivation for the choice of the weight for the extended norms and the choice of the extended localisation operator.

The relevant part of the flow at scale $k=0$ is
$$
H_0(\varphi) = H_0^{\es}(\varphi) + s n_0^a \nabla \varphi(a) \1_a + t n_0^b \nabla \varphi(b) \1_b.
$$
So at least on that scale one has a linear part in the observable flow.
The norm of the linear part creates the factor $l_{\obs,k} l_l$ which has to satisfy $l_{\obs,k} l_k \leq \rho^*(A)$ for the smoothness statement on $\mathbf{S}^{\ext}$ and $l_{\obs,k} l_k \leq \rho_0 \eta^k$ for the single step estimate. Thus $l_{\obs,k}$ has to include $\rho_0 \eta^k$ for $\rho_0$ small enough.

To get a contraction we have to put at least the constant part of the integrated irrelevant flow into the relevant flow. We aim to get an estimate
$$
\sum_{k=j_{ab}}^{N} \mathbf{B}K_k^{ab}
\leq C \frac{1}{|a-b|^{d+\nu}}
$$
Since
$$
\sum_{k=j_{ab}}^{N} \vert \mathbf{B}K_k^{ab} \vert
\leq \sum_{k=j_{ab}}^{N} l_{\obs,k}^{-2} \Vert K_k^{ab} \Vert_k^{(A),\ext}
\leq \sum_{k=j_{ab}}^{N} l_{\obs,k}^{-2} \rho_0 \eta^{3k}
$$
we need $L^{\frac{d}{2} j_{ab}}$ in $l_{\obs,k}$ for $k \geq j_{ab}$.

We cannot just put the constant $L^{\frac{d}{2} j_{ab}}$ in each $l_{\obs,k}$ for any $k$ since then $l_{\obs,k}l_k \leq \rho^*(A)$ cannot be satisfied for the scales where the linear part exists (at least at scale $0$). So we insert $L^{\frac{d}{2} (k \wedge j_{ab})}$ into the weight, until scale $j_{ab}$. Then we have to extract the linear part out of the irrelevant flow until coalescence to get a contraction since $\frac{l_{\obs,k+1}}{l_{\obs,k}}$ contains $L^{d/2}$ up to scale $j_{ab}$ which has to be extinguished for contraction by pulling out the linear part.

Another consequence of the inserted factor $L^{\frac{d}{2}k}$ into the weight is, that now we have to kill the growing sequence $h_k$ in $l_k$ so that the factor $2^{-k}$ appears in the weight.
\end{remark}

\subsection{A first estimate on the covariance}\label{subsec:first_Cov_bound}

Propositions \ref{Prop:Contractivity_Observables}, \ref{Prop:Bounds_on_B_ObservableFlow} and \ref{Prop:SingleStepRG_Observables} provide us with the following intermediate result: If $(H_k,K_k) \in \mathbb{D}_k(\rho_0,g_k,\Lambda)$, then we have good control of the differences $q_+^{ab} - q^{ab}$, $\lambda_+^{\alpha} - \lambda^{\alpha}$, $n_+^{\alpha} - n^{\alpha}$ and also of the observable part of $K_+$ (whose bulk part had been controlled along with the bulk coupling constants already in Proposition \ref{Prop:Existence_BulkFlow}). The following proposition links scales together via an inductive argument to conclude that $(H_k,K_k)$ remains in $\mathbb{D}_k$ for all $k \leq N$. It establishes a choice for the parameters $\rho_0$ and $C_{\mathcal{D}}$ as we had indicated above Proposition \ref{Prop:SingleStepRG_Observables}.

\begin{prop}[Existence of the observable flow]\label{Prop:Existence_ObservableFlow}
Fix $\eta \in (0,1)$. There is $L_0$ such that for all odd integers $L \geq L_0$ there are $A_0,h_0$ with the following property. For all $A \geq A_0$ and $h \geq h_0$ there is $\tilde{\epsilon}$ and $\rho_0$ (and $C_{\mathcal{D}}$) such that the flow $(\zeta_k,H_k,K_k)_{k \leq N}$ satisfies
\begin{align}
(H_k,K_k) \in \mathbb{D}_k(\rho_0,g_k,\Lambda) \label{eq:Staying_nice}
\end{align}
for any $k \leq N$.
\end{prop}

\begin{proof}
Let $L_0$ be large enough such that Propositions \ref{Prop:Smoothness_Observables}, \ref{Prop:Contractivity_Observables} and \ref{Prop:SingleStepRG_Observables} hold.

The proof of \eqref{eq:Staying_nice} is by induction on $k$ with the induction hypothesis
\begin{align*}
(IH)_k: \text{ for all } l \leq k, \, (H_l,K_l)\in \mathbb{D}_l(\rho_0,g_l,\Lambda).
\end{align*}
Note that by Proposition \ref{Prop:Existence_BulkFlow} the bulk flow satisfies
$$
\left(H_k^{\es},K_k^{\es}\right) \in B_{\bar{\epsilon}\eta^k}(0) \times B_{\bar{\epsilon}\eta^{2k}}(0)
$$
if $\mathcal{K}\in B_{\bar{\delta}}$. Furthermore, $\bar{\epsilon}$ can be made arbitrarily small by decreasing $\bar{\delta}$.
\begin{itemize}
\item Base clause $k=0$: We show that $H_0 \in B_{C_{\mathcal{D}}\rho_0}$ and $K_0 \in B_{\rho_0}$.
First, we have that, for $\alpha \in \lbrace a,b \rbrace$,
\begin{align*}
\Vert H_0^{\alpha} \Vert_{0,0} = l_{\obs,0} l_0 |n_0^{\alpha}| = \rho_0 h
\end{align*}
and thus
\begin{align*}
\Vert H_0 \Vert_{0,0}^{\ext}
= \Vert H_0^{\es} \Vert_{0,0} + \Vert H_0^a\Vert_{0,0} + \Vert H_0^b \Vert_{0,0}
\leq
\Vert H_0^{\es} \Vert_{0,0}
+ 2 \rho_0 h.
\end{align*}
Choose $\tilde{\epsilon}$ sufficiently small such that $\mathcal{K} \in B_{\tilde{\epsilon}}(0)$ implies $H_0^{\es} \in B_{\rho_0}(0)$. Let $C_{\mathcal{D}} \geq 1 + 2h$. Then
$$
\Vert H_0 \Vert_{0,0}^{\ext}
\leq C_{\mathcal{D}} \rho_0.
$$

To estimate $K_0$, note that
\begin{align*}
K_0(\varphi)
&= e^{s \nabla_{m_a} \varphi(a) \1_a + t \nabla_{m_b} \varphi(b) \1_b} K_0^{\es} (\varphi)
= e^{s \nabla_{m_a} \varphi(a) \1_a + t \nabla_{m_b} \varphi(b) \1_b}
e^{H_0^{\es}(\varphi)} \mathcal{K}(\varphi)
\\
&= e^{H_0^{\es} + s \nabla_{m_a} \varphi(a) \1_a + t \nabla_{m_b} \varphi(b) \1_b} \mathcal{K}
\\ &
= K_0^{\es}\left(\mathcal{K}, \mathcal{H} + s \nabla_{m_a} \varphi(a) \1_a + t \nabla_{m_b} \varphi(b)\1_b\right).
\end{align*}
Choose $\tilde{\epsilon}$ small enough such that $\mathcal{K} \in B_{\tilde{\epsilon}}(0)$ implies that $\mathcal{H} + s \nabla_{m_a} \varphi(a) \1_a + t \nabla_{m_b} \varphi(b)\1_b$ in turn is small enough such that
$$
K_0^{\es}(\mathcal{K}, \mathcal{H} + s \nabla_{m_a} \varphi(a) \1_a + t \nabla_{m_b} \varphi(b)\1_b) \in B_{\rho_0}(0)
$$
(use Lemma 12.2 in \cite{ABKM} for verification).

\item Induction hypothesis:
\begin{align*}
\forall \, 0 \leq l \leq k \quad (IH)_l \quad \text{holds.}
\end{align*}
\item Induction step:

For $\alpha \in \lbrace a,b \rbrace$, the following formula for the relevant observable flow holds:
$$
H_{k+1}^{\alpha} = \sum_{i=1}^d \left( \delta_{m^{\alpha}}(i) + \sum_{l=0}^{k \wedge (j_{ab}-1)} (\mathbf{B}_l K_l^{\alpha})^1_i \right) \nabla_i \varphi(\alpha).
$$
We use Proposition \ref{Prop:Bounds_on_B_ObservableFlow} and the induction hypothesis to estimate
\begin{align*}
\Vert H_{k+1}^{\alpha} \Vert_{k+1,0}^{\alpha}
& \leq l_{\obs,k+1} l_{k+1} \sum_{i=1}^d \left( \delta_{m^{\alpha}}(i) + \sum_{l=0}^{k \wedge (j_{ab}-1)} \left| (B K_l^{\alpha})^1_i \right| \right)
\\
& \leq \rho_0 g_{k+1} h \left(
1 + \frac{A_{\mathcal{B}}}{2} d \sum_{l=0}^{k \wedge (j_{ab}-1)} l_{\obs,l}^{-1} l_l^{-1} \Vert K_l \Vert_l^{(A),\ext}
\right)
\\
& \leq \rho_0 g_{k+1} h \left(
1 + \frac{A_{\mathcal{B}}}{2} h^{-1} d \sum_{l=0}^{k \wedge (j_{ab}-1)} (\rho_0 g_l)^{-1} \rho_0 g_{l}^2
\right)
\\
& \leq \rho_0 g_{k+1} h \left(
1 + \frac{A_{\mathcal{B}}}{2} h^{-1} d \sum_{l=0}^{\infty} \eta^l
\right).
\end{align*}
Let $C_{\mathcal{D}} \geq 1 + 2 h + A_{\mathcal{B}} d \frac{1}{1-\eta}$ and choose $\tilde{\epsilon}$ small enough such that $\mathcal{K} \in B_{\tilde{\epsilon}}(0)$ implies $H_k^{\es} \in B_{\rho_0 \eta^k}$. Then
\begin{align*}
\Vert H_{k+1} \Vert_{k+1,0}^{\ext}
&\leq \rho_0 \eta^{k+1} + 2 \rho_0 g_{k+1} h \left( 1 + \frac{A_{\mathcal{B}}}{2} d h^{-1} \frac{1}{1-\eta} \right)
\\ &
 \leq \rho_0 g_{k+1} \left( 1 + 2 h + A_{\mathcal{B}} d \frac{1}{1-\eta} \right)
 \leq C_{\mathcal{D}}\rho_0 g_{k+1}.
\end{align*}

For the estimate on $K_{k+1}$ we use Proposition \ref{Prop:SingleStepRG_Observables}. We can apply it by induction hypothesis and we obtain exactly what we want.
\end{itemize}

\end{proof}

From this result we can conclude a first estimate on the covariance.

\begin{prop}\label{Prop:Cov_first}
Fix $\eta \in \left(0,\frac{1}{4}\right)$. Then there is $L_1$ such that for all odd integers $L \geq L_1$ and the corresponding $A_0,h_0$ there is $\tilde{\epsilon}>0$ with the following property. For all $\mathcal{K} \in B_{\tilde{\epsilon}}\subset \mathbf{E}_{\zeta}$
\begin{align}
&\int e^{H_0(\varphi)}\circ K_0(\varphi) \mu_{\mathcal{C}^q}(\de\varphi)
= e^{\zeta_N} \left( e^{H_N(0)} + K_N(0) \right),
\label{equ:flow_for_observables}
\\
& \text{with} \quad
\zeta_N = st q_N^{ab} + s \lambda_N^a + t \lambda_N^b
\nonumber
\end{align}
where $(\zeta_k,H_k,K_k)$ is the flow from Proposition \ref{Prop:Existence_ObservableFlow}.
The term $q_N^{ab}$ can be written as follows:
\begin{align}
&q_N^{ab} = \Big( \delta_{m_a} + S_{j_{ab}}^a \Big) \Big( \delta_{m_b} + S_{j_{ab}}^b \Big) \nabla^* \nabla C^q(a,b) + \tilde{R}_{ab},
\label{eq:Cov_first}
\\ &\text{with}\quad
S_{j_{ab}}^{\alpha} = \sum_{k=0}^{j_{ab}-1} \left( \mathbf{B}K_k^{\alpha} \right)^1
\label{eq:formula_for_S_jab},
\end{align}
and there is $\tilde{C}_1$ such that for $0<\nu \leq -\frac{\ln (4 \eta)}{\ln L}$
\begin{align*}
\left|\tilde{R}_{ab}\right|
\leq \tilde{C}_1 \frac{1}{|a-b|^{d+\nu}}.
\end{align*}
Moreover, $\lambda_N^{\alpha}$ is uniformly bounded in $N$.

\end{prop}

\begin{proof}
The formulas \eqref{equ:flow_for_observables}, \eqref{eq:Cov_first} and \eqref{eq:formula_for_S_jab} follow from Proposition \ref{Prop:Existence_ObservableFlow} and Lemma~\ref{lemma:Further_properties_Observables} with
$$
\tilde{R_{ab}} = \sum_{k=j_{ab}}^{N-1} \int K_k^{ab}(\xi) \mu_{k+1}(\de\xi).
$$
Fix $\eta < \frac{1}{4}$. Choose $L_1$ large enough such that $\theta < \eta^3$, and that Proposition \ref{Prop:Existence_ObservableFlow} can be applied. Then there is $\tilde{\epsilon}>0$ such that for all $\mathcal{K} \in B_{\tilde{\epsilon}}(0)$ we can estimate:
\begin{align*}
\left| \tilde{R}_{ab} \right|
&\leq 
\sum_{l= j_{ab}}^{N} \left| \int K_l^{ab} \de \mu_{l+1} \right|
\leq
\frac{A_{\mathcal{B}}}{2} \sum_{l= j_{ab}}^{N} l_{\obs,l}^{-2} \left\Vert K_l^{ab} \right\Vert_l^{(A),\ext} 
\\
& \leq 
\frac{A_{\mathcal{B}}}{2} \rho_0^{-1}  L^{-d j_{ab}} \sum_{l= j_{ab}}^{N} 4^{-2(l - j_{ab})} 4^l g_l
\leq
\frac{A_{\mathcal{B}}}{2} \rho_0^{-1}  L^{-d j_{ab}} (4 { \eta})^{j_{ab}} \sum_{k=j_{ab}}^{N} 16^{-(l-j_{ab})}
\\
& \leq
\frac{A_{\mathcal{B}}}{2} \rho_0^{-1}  L^{-d j_{ab}} (4{\eta})^{j_{ab}} \sum_{k=0}^{\infty} 16^{-k} =
\frac{A_{\mathcal{B}}}{2} \rho_0^{-1} L^{-d j_{ab}} (4{\eta})^{j_{ab}} \frac{1}{1-1/16}.
\end{align*}
If $\eta < \frac{1}{4}$ there is additional decay on terms of $|a-b|$ due to $(4 \eta)^{j_{ab}}$:
$$
(4 \eta)^{j_{ab}} \leq (4 \eta)^{\log_L(2|a-b|)} = (2|a-b|)^{\frac{\ln(4\eta)}{\ln L}}
$$
and so
$$
\left( L^{-d} 4 \eta \right)^{j_{ab}}
 \leq (2|a-b|)^{-(d-\frac{\ln(4 \eta)}{\ln L})}
 \leq (2|a-b|)^{-(d+\nu)}
$$
for $0 < \nu \leq - \frac{\ln(4 \eta)}{\ln L}$. This gives
$$
\left|\tilde{R}_{ab}\right| \leq C \frac{1}{|a-b|^{d + \nu}}.
$$
The uniform bound on $\lambda_N^{\alpha}$ follows similarly.
\end{proof}

\subsection{A refined estimate on the covariance}
\label{subsec:refined_estimate}

Proposition \ref{Prop:Cov_first} can be used to show that
\begin{align*}
&\Cov_{\gamma_{N,\beta}^u} \left( \nabla_{m_a} \varphi(a), \nabla_{m_b} \varphi(b) \right)
=
q_N^{ab} + \mathcal{O}\left(2^N\right)
\\&
 \quad\quad\quad\quad\quad\quad
=
\left( \delta_{m_a} + S_{j_{ab}}^a \right)
\left( \delta_{m_b} + S_{j_{ab}}^b \right)
\nabla^* \nabla C^q(a,b) + \tilde{R}_{ab} + \mathcal{O}\left(2^N\right).
\end{align*}

The goal of this subsection is to establish an improved formula for $q_N^{ab}$, namely
$$
q_N^{ab} = \nabla^*_{m_b} \nabla_{m_a} C^q(a,b) + R_{ab},
\quad \text{with} \quad
\left| R_{ab} \right| \leq C \frac{1}{|a-b|^{d + \nu}}.
$$
This estimate follows from formula \eqref{eq:Cov_first} if we can show that
$$
\left| S^{\alpha}_{j_{ab}} \nabla^* \nabla C^q(a,b) \right|
\leq C \frac{1}{|a-b|^{d + \nu}}.
$$

We analyse the dependence of $S_{j_{ab}}^{\alpha}$ on $j_{ab}$ as $j_{ab} \rightarrow \infty$ in order to obtain the desired bound. Precisely, we prove the following.

\begin{prop}\label{Prop:Cov_better}
Under the assumptions of Proposition \ref{Prop:Cov_first} there is a constant~$C$ which depends on $A_{\mathcal{B}}$, $h$, and $\eta$ such that
\begin{align*}
S^a_{j_{ab}}, S^b_{j_{ab}} \leq C\eta^{j_{ab}}.
\end{align*}
\end{prop}

We start by motivating the ideas of the proof in the following section. Afterwards, the rigorous proof follows.

\subsubsection{Motivation for the proof of Proposition \ref{Prop:Cov_better}}

Using the results in Subsection \ref{subsec:first_Cov_bound} we can construct sequences $(n_k^a, n_k^b)_{k \leq j_{ab}}$ and $(q_k^{ab})_{k \leq N}$ with a coalescence scale $j_{ab}$ and
$$
n_k^{\alpha} = n_0^{\alpha} + \sum_{l=0}^{k-1} \left( \mathbf{B}_l K_l^{\alpha} \right)^1
= n_0^{\alpha} + S_k^{\alpha}.
$$
The goal is to analyse the dependence of $n_{j_{ab}}^{\alpha}$ on $j_{ab}$ as $j_{ab} \rightarrow \infty$. The key steps in the proofs are:
\begin{itemize}
	\item \textit{Single observable flow:} From 4.\ in Lemma \ref{Lemma:FirstConclusions_Observables} we can deduce that $n_k^a$ is independent of $(n_l^b)_{l\leq k}$. In particular we can choose $n_0^b = 0$ without changing the flow $n_k^a$. In this case we regard the observable at $b$ as being absent, so the concept of coalescence becomes vacuous. We use the convention that in this case $j_{ab}=\infty$.
If $n_0^b=0$ then no $b$-term or $ab$-term arise in the flow. Nevertheless, the estimates on $\mathbf{B}K^a$ and $K^a$ hold as before.

\item \textit{Extension to an infinite sequence:} We show that $(n_k^a)_{k \leq j_{ab} \wedge N}$ is independent of the size of the torus $\Lambda$. This allows us to extend the flow to an infinite sequence $n_k^{a,\Z^d}$ which can be written as
$$
n_k^{a,\Z^d} = n_0^a + \sum_{l=0}^{k-1} \left(\mathbf{B}_l K_l^{a,\Z^d} \right)^1.
$$

\item \textit{Convergence of the sequence:} A subtle argument shows that $n_k^{a,\Z^d} \rightarrow n_0^a$ and from this convergence we can deduce that
$$
\sum_{k=0}^{\infty} \left( \mathbf{B}_l K_l^{a,\Z^d} \right)^1 = 0,
\quad \mbox{and thus} \quad
\sum_{k=0}^{m-1} \left( \mathbf{B}_l K_l^{a,\Z^d} \right)^1 = \mathcal{O} \left(\eta^m\right).
$$
\item \textit{Back to finite volume:} If $\mathbf{B}_k K_k^{a,\Z^d} = \mathbf{B}_k K_k^{a,\Lambda}$ holds for any $k \leq j_{ab}-1$, then
$$
\sum_{k=0}^{j_{ab}-1} \left( \mathbf{B}_k K_k^{a,\Lambda} \right)^1 = \mathcal{O} \left(\eta^{j_{ab}}\right).
$$
\end{itemize}

~\\
The computation of the limit of $n_k^{a,\Z^d}$ can be motivated as follows.

From the result on the scaling limit in Theorem 2.1 in \cite{Hil19_1} we know that the Gaussian covariance~$\mathcal{C}^q$ arises without any correction term. We try to establish a connection to this result by smoothing the observable flow. Namely we will consider
$$
\int n_0^a (\nabla (\varphi + \xi), g_N) F_0^{\es}(\varphi + \xi) \mu_{\mathcal{C}^q}(\de\xi)
$$
for a suitable chosen $g_N$ (as in Theorem 2.1 in \cite{Hil19_1}). Here, we denote $F_0^{\es} = e^{H_0^{\es}}\circ K_0^{\es}$ the bulk flow.

On the one hand we can write this expression as
$$
\sum_x g_N(x) \int n_0^a \nabla (\varphi + \xi)(x) F_0^{\es}(\varphi + \xi) \mu_{\mathcal{C}^q}(\de\xi),
$$
which can be related to the observable flow if we show that the flow of coefficients $n_k^a$ is independent of the placing of the observable $a \in \Lambda$. Let us include the choice of a placing $a \in \Lambda$ in the notation as $Z_N(\varphi;a)$. Then
\begin{align*}
&\int n_0^a (\nabla (\varphi + \xi), g_N) F_0^{\es}(\varphi + \xi) \mu_{\mathcal{C}^q}(\de\xi)
\\
&\quad\quad\quad= \sum_x g_N(x)
\partial_s \big\vert_{s=0} \ln \int e^{s n_0^a \nabla(\varphi+\xi)(x)}F_0^{\es}(\varphi + \xi) \mu_{\mathcal{C}^q}(\de\xi)
\\
&\quad\quad\quad= \sum_x g_N(x)
 \partial_s \big\vert_{s=0} \ln Z_N(\varphi;x).
\end{align*}
On the other hand we can relate the original expression to the bulk flow and the scaling limit as follows:
\begin{align*}
&\int n_0^a (\nabla (\varphi + \xi), g_N) F_0^{\es}(\varphi + \xi) \mu_{\mathcal{C}^q}(\de\xi)
\\
& \quad\quad\quad=
n_0^a \partial_f \left[
\int e^{(\varphi + \xi,f)} F_0^{\es}(\varphi + \xi) \mu_{\mathcal{C}^q}(\de\xi)
\right]_{f=0}(\nabla^* g_N).
\end{align*}

\subsubsection{Proof of Proposition \ref{Prop:Cov_better}}

The procedure described above will be implemented here.

\paragraph{Single observable flow} 
Let $(H_k,K_k)_{k \leq N}$ be the flow from Section \ref{subsec:first_Cov_bound} with initial data $n_0^a = \delta_{m_a}$ and $n_0^b = \delta_{m_b}$ which satisfies $(H_k,K_k) \in \mathbb{D}_k (\rho_0,g_k,\Lambda)$. Remember from Lemma \ref{Lemma:FirstConclusions_Observables} that $n_+^a$ is independent of $n^b$, $K^b$ and $K^{ab}$. Thus we can consider the initial datum $n_0^b = 0$ without changing the $n_k^a$-flow. Moreover, no $b$- and $ab$-term will ever arise. We summarize the properties in the following lemma.

\begin{lemma}
\begin{enumerate}
	\item Let $n_0^{\alpha} \in \lbrace 0,\delta_m \rbrace$. For any $k \leq j_{ab} \wedge N$, the term $n_k^a$ is independent of $(n_l^b)_{l \leq k}$.
	\item If $n_0^{\alpha} = 0$ then $H_k^{\alpha} = 0 = K_k^{\alpha}$ for all $k \leq N$.
\end{enumerate}
\end{lemma}

\begin{proof}
The claims follows inductively from 4. in Lemma \ref{Lemma:FirstConclusions_Observables}.
\end{proof}

Since Propositions \ref{Prop:2nd_order_perturbation_effect_on_S}, \ref{Prop:Smoothness_Observables}, \ref{Prop:Contractivity_Observables} and \ref{Prop:Bounds_on_B_ObservableFlow} hold as before also in the case $n_0^b = 0$, the following proposition can be proven by induction in the same way as Proposition~\ref{Prop:Existence_ObservableFlow}.

\begin{prop}\label{Prop:Existence_SingleObservableFlow}
Let $n_0^a = \delta_{m_a}$ and $n_0^b=0$. By the same assumptions as in Proposition~\ref{Prop:Cov_first} the flow $(\zeta_k, H_k, K_k)_{k \leq N}$ exists with
\begin{align*}
&\zeta_k = \lambda_k^a = \sum_{l=0}^{k-1} \left( \mathbf{B}_l K_l^a \right)^1,
\\
&H_k(\varphi) = H_k^{\es}(\varphi) + s n_k^a \nabla \varphi(a) \1_a,
\quad \text{where} \quad
n_k^a = n_0^a + \sum_{l=0}^{k-1} \left( \mathbf{B}_l K_l^a \right)^1,
\end{align*}
and
$$
(H_k^{}, K_k^{}) \in \mathbb{D}_k(\rho_0,g_k,\Lambda).
$$
\end{prop}

\paragraph{Extension to an infinite sequence}
Now we extend $n_k^a$ to an infinite sequence. This is possible in view of the following independence property.

\begin{lemma}\label{lemma:indep_volume_observable}
Let us denote the dependence on the torus $\Lambda$ by writing $n_k^a = n_k^{a,\Lambda}$. Let $\Lambda'$ be a larger torus. Then
$$
n_k^{a,\Lambda} = n_k^{a,\Lambda'} \text{ for all } k \leq N(\Lambda).
$$
\end{lemma}

\begin{proof}
From the $N$-independence of the map $\mathbf{B}$ and the $(\Z^d)$-property for $K$ we can conclude that for $k < N$ and $B \in \mathcal{P}_{k+1}$
\begin{align*}
\mathbf{B} K_k^{a,\Lambda}(B)
= \sum_{b \in \mathcal{B}_k(B)} \Pi^a_k \mathcal{R}_{k+1} K_k^{a,\Lambda}(b)
= \sum_{b \in \mathcal{B}_k(B)} \Pi^a_k \mathcal{R}_{k+1} K_k^{a,\Lambda'}(b)
= \mathbf{B} K_k^{a,\Lambda'}(B)
\end{align*}
since for $b \in \mathcal{B}_k(B)$ and $k < N$ it holds that $\diam (b) \leq \frac{1}{2}\diam(\Lambda)$. For $k \leq N$ we thus get
$$
n_k^{a,\Lambda} = n_0^a + \sum_{l=0}^{k-1} \left( \mathbf{B}K_l^{a,\Lambda} \right)^1
= n_0^a + \sum_{l=0}^{k-1} \left( \mathbf{B}K_l^{a,\Lambda'} \right)^1
= n_k^{a,\Lambda'}.
$$
\end{proof}

For $k \in \N$ define
$$
n_k^{a,\Z^d} = n_k^{a,\Lambda},
\quad \Lambda \text{ large enough such that } k \leq N (\Lambda).
$$
The sequence is well-defined by Lemma \ref{lemma:indep_volume_observable}.
By definition, it holds that
\begin{align*}
n_k^{a,\Z^d} = n_0^a + \sum_{l=0}^{k-1} \left( \mathbf{B}K_l^{a,\Z^d} \right)^1.
\end{align*}

\paragraph{Convergence of the sequence} First of all we need another generalisation. Namely, we start with an arbitrary position $x \in \Lambda$ of the observable instead of a fixed $a$.

Let
\begin{align*}
H_0(\varphi;x) &= H_0^{\es}(\varphi) + s n_0 \nabla \varphi(x) \1_x,
\quad n_0 = \delta_m \text{ for } m\in {1, \ldots, d},
\quad
 K_0 & = e^{H_0} K_0^{\es}.
\end{align*}
\begin{lemma}
The sequence $(n_k^x)_k$ is independent of the choice of the position~$x$. More precisely, fix $x,a \in \Lambda$ and $n_0$and consider two flows with initial condition $H_0^x(\varphi;x) = n_0 \nabla\varphi(x)\1_x$ and $H_0^a(\varphi;a) = n_0 \nabla\varphi(a)\1_a$ and the corresponding $K_0$. Then $n_k^x = n_k^a$ for all $k \leq N$.
\end{lemma}

We can drop the superscript $x$ from the notation by this property.

\begin{proof}
Fix $x,a \in \Lambda$. We need the following "translation property" of $K^a$:
\begin{align}
\text{at any scale $k$, for any $X$ and $\varphi$, } K^a(\varphi,X) = K^x(\tau_{x-a} \varphi, \tau_{x-a}X). \label{eq:TranslationObservable}
\end{align}
We will prove \eqref{eq:TranslationObservable} subsequently. This property and translation invariance of the measure imply that the coefficients of $\mathbf{B}_k K_k^a$ equal the coefficients of $\mathbf{B}_k K_k^x$:
\begin{align*}
(\mathbf{B}K^a)^0 &= \int K^a(\varphi,B^a) \mu_+(\de\varphi)
= \int K^x(\tau_{x-a}\varphi, B^x) \mu_+(\de\varphi)
\\
&= \int K^x(\psi, B^x) \mu_+(\de\psi)
=(\mathbf{B}K^x)^0.
\end{align*}
and, since (by \eqref{eq:formula_linear_projection})
\begin{align*}
(\mathbf{B}K^a)^1_i = \langle \mathcal{R}_+ K^a, b_i^a \rangle_0
= D(\mathcal{R}_+ K^a)(0)(\varphi_i^a),
\quad \varphi_i^a(x) = x_i - a_i,
\end{align*}
we similarly get
\begin{align*}
(\mathbf{B}K^a)^1 
&= \int DK^a (\varphi,B^a) \varphi_i^a \mu_+(\de\varphi)
= \int D \left( K^x(\tau_{x-a} \varphi, \tau_{x-a}B^a) \right) \varphi_i^a \mu_+(\de\varphi)
\\
&= \int DK^x (\tau_{x-a}\varphi,B^x)(\tau_{x-a}\varphi_i^a) \mu_+(\de\varphi)
= \int DK^x(\psi, B^x)(\tau_{x-a}\varphi_i^a) \mu_+(\de\psi)
\\
&=\int DK^x(\psi,B^x)(\varphi_i^x) \mu_+(\de\varphi)
=(\mathbf{B}K^x)^1.
\end{align*}
By induction we verify that $n_k^a = n_k^x$ for any $k$.

~\\
It remains to prove \eqref{eq:TranslationObservable}. We again argue by induction. The induction hypothesis is
\begin{align}
\text{For all } l \leq k \text{ and } X \in \mathcal{P}_l, \quad K_l^a(\varphi,X) = K_l^x(\tau_{x-a}\varphi,\tau_{x-a}X).
\label{eq:InductionTranslationObservable}
\end{align}
The case $k=0$ is immediate:
\begin{align*}
K^x_0(\tau_{x-a}\varphi,\tau_{x-a}X) 
&= n_0 \nabla(\tau_{x-a}\varphi)(x) \1_x(\tau_{x-a}X) K_0^{\es}(\tau_{x-a}\varphi,\tau_{x-a}X)
\\
&=n_0 \nabla \varphi(a) \1_a(X) K_0^{\es}(\varphi,X)
=K_0^a(\varphi,X).
\end{align*}
For the induction step 
we have to show that for all $U \in \mathcal{P}_{k+1}$
$$
\left[ \mathbf{S}^{\ext}(H,K) \right]^a(\varphi,U) = \left[ \mathbf{S}^{\ext}(H,K) \right]^{x}(\tau_{x-a}\varphi,\tau_{x-a}U).
$$ 
From the definition of $\mathbf{S}^{\ext}$ it holds that
\begin{align*}
\left[\mathbf{S}^{\ext} (H,K)\right]^a(\varphi,U)
= (\mathbf{B}K^a)^0 \mathbf{S}^{\es}(H,K)(\varphi,U)
+ \left[\mathbf{S}(H,K)\right]^a(\varphi,U).
\end{align*}
We already showed that $(\mathbf{B}K^a)^0 = (\mathbf{B}K^x)^0$ and that the bulk part satisfies translation invariance, so the first term becomes
$$
(\mathbf{B}K^a)^0 \mathbf{S}^{\es}(H,K)(\varphi,U)
= (\mathbf{B}K^x)^0 \mathbf{S}^{\es}(H,K)(\tau_{x-a}\varphi,\tau_{x-a}U).
$$
For the second term, from the definition of $\mathbf{S}$, there is always one $a-$part falling on either $e^{\tilde{H}}(U\setminus X)$ or $e^{-\tilde{H}(X \setminus U)}$ or $(1-e^{\tilde{H}})$ or $(e^H-1)$ or $K$. The others form the bulk part. The bulk part always satisfies translation invariance, so we just have to check if the $a$-part translates correctly.

If the $a$-part falls on $K$, we use the induction hypothesis and translation invariance of the measure, and translate the sum over polymers $\sum_{X \in \mathcal{P}_k}\chi(X,U)$ into $\sum_{X \in \mathcal{P}_k} \chi(X,\tau_{x-a}U)$. The input field is then $\tau_{x-a}\varphi$.

If the $a$-part falls on $e^{\tilde{H}}$, we have
\begin{align*}
\tilde{H}^a(B_a)(\varphi)
= H^a(B_a,\varphi) + \mathbf{B}K^a(B_a,\varphi)
= H^{x}(B_x,\tau_{x-a} \varphi) + \mathbf{B}K^x(B_x,\tau_{x-a}\varphi).
\end{align*}
\end{proof}

Now we can prove the convergence result.

\begin{prop}\label{Prop:Tool_for_Cov_better}
Given the assumptions of Proposition \ref{Prop:Cov_first}, the sequence
$$
\left(n_k^{a,\Z^d}\right)_{k\in \N}
$$
converges to the limit $n_{\infty} = n_0^a$.
\end{prop}

\begin{proof}

Convergence of the sequence is clear since by Proposition \ref{Prop:Bounds_on_B_ObservableFlow} and Proposition~\ref{Prop:Existence_SingleObservableFlow} we can bound the sum uniformly in $N$:
\begin{align*}
\sum_{l=0}^{k-1} \left| \left( \mathbf{B} K_l^{a,\Z^d} \right)^1 \right|
\leq \sum_{l=0}^{k-1} \frac{A_{\mathcal{B}}}{2} h^{-1} \eta^k
< \infty.
\end{align*}
Let us denote the limit by $n_{\infty}$.

We show $n_{\infty} = n_0^a$ by a limiting procedure involving the result on the scaling limit in Theorem 2.1 in \cite{Hil19_1}.
Let
$$
Z_N(\varphi;x) = e^{\zeta_N} \left( e^{H_N(\varphi;x)} + K_N(\varphi) \right)
$$
be the generating partition function at scale $N$, with one observable at position $x$. Let $g_N(x) = L^{-N\frac{d}{2}} g\left( L^{-N}x \right)$ for $g \in C_c^{\infty}(\T^d)$ satisfying $\int g = 0$ as in the assumptions of Theorem 2.1 in \cite{Hil19_1} and $h_N = \mathcal{C}^q \nabla_j g_N$, $h = \mathcal{C}^q \partial_j g$ for a fixed direction $j \in \lbrace 1, \ldots, d \rbrace$. We will show that
\begin{equation}
D \left[
\int n_0^a \left( \nabla (\varphi + \xi),g_N \right) F_0^{\es}(\varphi + \xi) \mu_{\mathcal{C}^q}(\de\varphi)
\right]_{\varphi = 0}(h_N)
\,
 \xrightarrow{N \rightarrow \infty}
\,
n_{\infty} (\partial h,g)_{L^2(\T^d)}
\tag{A}
\label{eq:OneDirection}
\end{equation}
by the statements on the observable flow. Here, the left hand side denotes the directional derivative of the term in brackets of $\varphi$ in direction $h_N$.

 On the other hand,
by transforming the term into derivatives of the bulk partition function and using results there, we will show that

\begin{equation}
D \left[
\int n_0^a \left( \nabla (\varphi + \xi),g_N \right) F_0^{\es}(\varphi + \xi) \mu_{\mathcal{C}^q}(\de\varphi)
\right]_{\varphi = 0}(h_N)
\,
 \xrightarrow{N \rightarrow \infty}
\,
n_0^a (h, \partial^* g)_{L^2(\T^d)}.
\tag{B}
\label{eq:TheOtherDirection}
\end{equation}
By uniqueness of the limit we can conclude that $n_{\infty} = n^a_0$.

~\\
We start by proving \eqref{eq:OneDirection}. We can transform
\begin{align*}
&D \left[
\int n_0^a \left( \nabla (\varphi + \xi),g_N \right) F_0^{\es}(\varphi + \xi) \mu_{\mathcal{C}^q}(\de\varphi)
\right]_{\varphi = 0}(h_N)
\\ &\quad =
\sum_x g_N(x)
D \left[
\partial_s \big\vert_{s=0} \ln Z_N(\varphi;x) \right]_{\varphi=0} (h_N)
\\ & \quad =
n_N^a (\nabla h_N,g_N) \frac{e^{H_N^{\es}(0)}}{Z_N^{\es}(0)}
\\ & \quad\quad\quad\quad
+
\frac{1}{Z_N^{\es}(0)}
\sum_x g_N(x) DK_N^x(0)(h_N)
+
\frac{DZ_N^{\es}(0)(h_N)}{\left( Z_N^{\es}(0)\right)^2} \sum_x g_N(x) K_N^x(0).
\end{align*}

By Lemma \ref{lemma:smoothness_of_E} we can estimate
$$
\left|e^{H_N^{\es}(0)}-1\right|
\leq C
\vertiii{
e^{H_N^{\es}(0)}-1\
}_N
\leq C \Vert H_N^{\es} \Vert_{N,0},
$$
and, since $(H_N^{\es},K_N^{\es}) \in \mathbb{D}_k(\rho_0,g_k,\Lambda)$, we conclude that
$$
\left|e^{H_N^{\es}(0)}-1\right|, \left|Z_N^{\es}(0)-1\right| \rightarrow 0.
$$
Together with the convergence result of Proposition 4.7 in \cite{Hil16} we obtain
$$
n_N^a (\nabla h_N,g_N) \frac{e^{H_N^{\es}(0)}}{Z_N^{\es}(0)}
\quad
\rightarrow
\quad
n_{\infty}(\partial h, g)_{L^2(\T^d)}
\text{ as } N \rightarrow \infty.
$$
Furthermore,
$$
\left|
\sum_x g_N(x)
\right|
 = L^{Nd/2} L^{-Nd} 
\left|\sum_x g(L^{-N}x) \right|
 \leq C L^{Nd/2} \left|\int g(x) \de x\right|,
$$
but
$$
\vert D K_N^{a}(0)(h_N) \vert
\leq l_{\obs,N}^{-1} \Vert K_N \Vert_N^{(A),\ext} \vert \mathcal{C}^q \nabla^* g_N \vert_{N,\Lambda_N}
\leq C L^{-Nd/2} \eta^N 2^N
$$
for a constant independent of $N$,
such that
$$
\left|\sum_x g_N(x)	D K_N^x(0)(h_N)\right|
 \leq C (2 \eta)^N \rightarrow 0.
$$
We estimate $DZ_N^{\es}(0)(h_N)$ as in the proof of Theorem \ref{Thm:RepresentationPartitionFunction}. Namely,
\begin{align*}
\left| D Z_N^{\es}(0)(h_N) \right|
&= \left| D \left(Z_N^{\es}-1\right)(0)(h_N) \right|
\\ &
\leq 
\left| D \left(e^{H_N^{\es}}-1\right)(0)(h_N) \right|
+ \left| D K_N^{\es}(0)(h_N) \right|
\\
& \leq C
\left(
\vertiii{e^{H_N^{\es}} - 1}_{N} |h_N|_{N,\Lambda_N}
+ \left\Vert K_N^{\es} \right\Vert_N^{(A)} |h_N|_{N,\Lambda_N}
\right).
\end{align*}
By Lemma \ref{lemma:smoothness_of_E} it holds that
$$
\vertiii{e^{H_N^{\es}} - 1}_{N}
\leq 8 \left\Vert H_N^{\es} \right\Vert_{N,0}.
$$
Moreover, similar to Lemma 5.2 from \cite{Hil16} one can show that
$$
\left| h_N \right|_{N,\Lambda_N} =
\left| \mathcal{C}^q \nabla^* g_N \right|_{N,\Lambda_N} \leq C
$$
for a constant $C$ which is independent of $N$.
With $(H_N^{\es},K_N^{\es}) \in \mathbb{D}(\rho_0,g_k,\Lambda)$ it follows that
$$
\left| D Z_N^{\es}(0)(h_N) \right|
\leq C
\left(
\Vert H_N^{\es}\Vert_{N,0} 
+ \Vert K_N^{\es} \Vert_N^{(A)}
\right)
\leq C \eta^N.
$$

Thus
$$
\frac{\left|D Z_N^{\es}(0)(h_N)\right|}{\left|Z_N^{\es}(0)\right|^2}
\rightarrow 0,
$$
and
$$
\left|\sum_x g_N(x) K_N^x(0)\right|
\leq C
L^{Nd/2} L^{-Nd/2} (2 \eta)^N \rightarrow 0.
$$

Now we prove \eqref{eq:TheOtherDirection}. We start with the following transformations:
\begin{align*}
&
D \left[
\int n_0^a \left( \nabla (\varphi + \xi),g_N \right) F_0^{\es}(\varphi + \xi) \mu_{\mathcal{C}^q}(\de\varphi)
\right]_{\varphi = 0}(h_N)
\\ & 
\quad\quad\quad
 =
n_0^a \partial_f \left[
D \left[
\int e^{(\varphi + \xi,f)} F_0^{\es}(\varphi + \xi) \mu_{\mathcal{C}^q}(\de\xi)
\right]_{\varphi=0}(h_N)
\right]_{f=0}(\nabla^* g_N)
\\ & 
\quad\quad\quad
 =
n_0^a \partial_f \left[
e^{\frac{1}{2}(f,\mathcal{C}^q f)}
D \left[
e^{(\varphi,f)} Z_N^{\es}(\mathcal{C}^q f + \varphi)
\right]_{\varphi=0}(h_N)
\right]_{f=0} (\nabla^* g_N)
\\ & 
\quad\quad\quad
 =
n_0^a \partial_f \left[
e^{\frac{1}{2}(f,\mathcal{C}^q f)}
(h_N,f) Z_N^{\es}(\mathcal{C}^q f) + D Z_N^{\es}(\mathcal{C}^q f)(h_N)
\right]_{f=0} (\nabla^* g_N)
\\ &
\quad\quad\quad 
=
n_0^a
\left[
(h_N, \nabla^* g_N) Z_N^{\es}(0)
+ D^2 Z_N^{\es}(0)(h_N)(\mathcal{C}^q \nabla^* g_N)
\right].
\end{align*}
The first term converges
$$
(h_N,\nabla^* g_N) Z_N^{\es}(0)
 \rightarrow 
(h,\partial^* g)_{L^2(\T^d)}
$$ 
as $N \rightarrow \infty$, due to $\left|Z_N^{\es}(0)-1\right| \rightarrow 0$ and the convergence result of Proposition 4.7 from \cite{Hil16}. The second term tends to zero by the following considerations which resemble the arguments in the proof of Theorem \ref{Thm:RepresentationPartitionFunction} and \eqref{eq:OneDirection}.
It holds that
$$
D^2 Z_N^{\es}(0)(h_N,\mathcal{C}^q \nabla^* g_N)
= 
D^2 \left(Z_N^{\es}-1\right)(0)(h_N,\mathcal{C}^q \nabla^* g_N),
$$
and thus
\begin{align*}
&\left|
D^2 Z_N^{\es}(0)(h_N,\mathcal{C}^q \nabla^* g_N)
\right|
\\ & \quad
\leq
\left|
D^2 \left(e^{H_N^{\es}} - 1 \right)(0)(h_N,\mathcal{C}^q \nabla^* g_N)
\right|
+
\left|
D^2 K_N^{\es}(0)(h_N,\mathcal{C}^q \nabla^* g_N)
\right|
\\ & \quad
\leq C
\left(
\vertiii{e^{H_N^{\es}} - 1}_N  + \left\Vert K_N^{\es} \right\Vert_N^{(A)}
\right)
|h_N|_{N,\Lambda_N} |\mathcal{C}^q \nabla^*g_N|_{N,\Lambda_N}.
\end{align*}
As before it holds that
$$
\left| \mathcal{C}^q \nabla^* g_N \right|_{N,\Lambda_N},
\left| h_N \right|_{N,\Lambda_N}
 \leq C
$$
for a constant $C$ which is independent of $N$, and
$$
\vertiii{e^{H_N^{\es}} - 1}_N 
\leq C \left\Vert H_N^{\es} \right\Vert_{N,0}.
$$
Together with $(H_N^{\es},K_N^{\es}) \in \mathbb{D}_k(\rho_0,g_k,\Lambda)$ we conclude that
$$
\left|
D^2 Z_N^{\es}(0)(h_N,\mathcal{C}^q \nabla^* g_N)
\right|
\leq C \eta^k
\rightarrow 0.
$$
This proves the claim.

\end{proof}

\paragraph{Back to finite volume}

Now we can prove Proposition \ref{Prop:Cov_better}.

\begin{proof}[Proof of Proposition \ref{Prop:Cov_better}]
We conclude from Proposition \ref{Prop:Tool_for_Cov_better} and the construction of the flow that
\begin{align*}
n_{\infty} = n_0^a + \sum_{k=0}^{\infty} \left( \mathbf{B}_k K_k^{a,\Z^d} \right)^1 = n_0^a
\quad
\Rightarrow
\quad
\sum_{k=0}^{\infty} \left( \mathbf{B}_k K_k^{a,\Z^d} \right)^1 = 0.
\end{align*}
Using Proposition \ref{Prop:Existence_SingleObservableFlow} we can estimate
\begin{align*}
\left|
\sum_{k=0}^{m-1} \left( \mathbf{B}_k K_k^{a,\Z^d} \right)^1
\right|
=
\left|
\sum_{k=m}^{\infty} \left( \mathbf{B}_k K_k^{a,\Z^d} \right)^1
\right|
\leq
\sum_{k=m}^{\infty} \frac{A_{\mathcal{B}}}{2} h^{-1} \eta^k
= \frac{A_{\mathcal{B}}}{2} h^{-1} \frac{1}{1-\eta} \, \eta^{m}.
\end{align*}
By definition of the infinite sequence, the $(\Z^d)$-property and the local dependence of the relevant flow it holds that for all $k \leq j_{ab}-1$
$$
\mathbf{B}_k K_k^{a,\Z^d} = \mathbf{B}_k K_k^{a,\Lambda}.
$$
Thus
$$
S^a_{j_{ab}} = \sum_{k=0}^{j_{ab}-1} \left( \mathbf{B}_k K_k^a \right)^1
= \mathcal{O}\left(\eta^{j_{ab}}\right).
$$
\end{proof}

\begin{remark}\label{rem:higher_corr}
As noted in 4.\ in Remark \ref{rem:extension_main_thm}, we can use a similar method to show a fine estimate on higher correlations. We sketch the argument here.

Fix $n \in \N$ and let $m_i \in \lbrace 1, \ldots , d \rbrace$ and $a_i \in \Lambda_N$ for $i \in \lbrace 1, \ldots, n \rbrace$.

We consider
\begin{align*}
&\Cov_{\gamma_{N,\beta}^u}
\left( \nabla_{m_1 \varphi(a_1)}, \ldots, \nabla_{m_1 \varphi(a_1)} \right)
:=
\\ & \quad\quad\quad\quad
\partial_{s_1} \ldots \partial_{s_n} \big\vert_{s_1 = \ldots = s_n = 0} \ln
\int e^{\sum_{i=1}^n s_i \nabla_{m_i} \varphi(a_i)} F_0^{\es}(\varphi) \mu_{\mathcal{C}^q}(\de\varphi).
\end{align*}
For the sake of simplicity let us assume that the fixed sites $a_i$ are in different blocks until scale $j_*-1$, and then they all live in one single block at scale $j_*$, as indicated in the following picture for the case $L=4$.

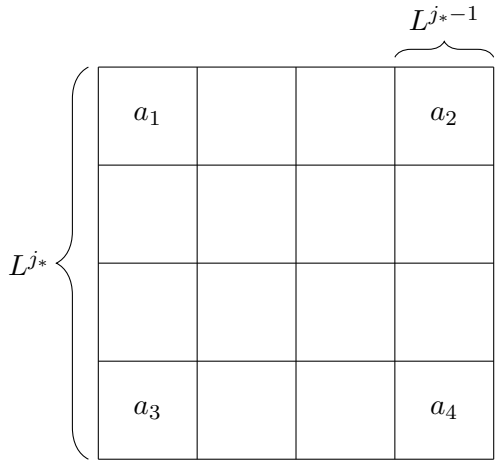
\vspace{0.3cm}
\begin{tikzpicture}[scale=1.3]
	\draw[decoration= {brace, amplitude = 12 pt}, decorate] (0.9, 0) -- (0.9,4);
	\node at (0.3,2) {$L^{j_{*}}$};
	\draw (1,0) -- (1,4);
		\node at (1.5,0.5) {$a_3$};
		\node at (4.5,0.5) {$a_4$};
	\draw (2,0) -- (2,4);
	\draw (3,0) -- (3,4);
	\draw (4,0) -- (4,4);
		\node at (1.5,3.5) {$a_1$};
		\node at (4.5,3.5) {$a_2$};
	\draw (5,0) -- (5,4);
	\draw (1,0) -- (5,0);
	\draw (1,1) -- (5,1);
	\draw (1,2) -- (5,2);
	\draw (1,3) -- (5,3);
	\draw (1,4) -- (5,4);
		\draw[decoration= {brace, amplitude = 6 pt}, decorate] (4,4.1) -- (5,4.1);
		\node at (4.5,4.5) {$L^{j_* - 1}$};
\end{tikzpicture}

~\\
We are only interested in the behaviour of functionals up to first order in $s_P$, $P~\subset~\lbrace 1,\ldots,n \rbrace$, where $s_P := \prod_{i\in P} s_i$.

Thus we consider functionals in a quotient algebra such that
$$
K = \sum_{P \subset \lbrace 1,\ldots,n \rbrace} s_P K^P
$$
for some $K^P \in \mathcal{N}^{\es}$.

We construct the flow exactly as before. Due to the special choice of sites $a_1, \ldots, a_n$, the observable part of the relevant Hamiltonian is of the following form:
$$
H^{\obs} = \sum_{i=1}^n s_i H^i + s_{\lbrace 1,\ldots,n\rbrace} H^{\lbrace 1,\ldots,n \rbrace}
$$
for a constant $H^{\lbrace 1,\ldots,n \rbrace}$ and
$$
H^i = \lambda^i + n^i \nabla\varphi(a_i).
$$
The map $\mathbf{A}$ is constructed as in the case of two observables. Note that
\begin{align*}
&
\mathcal{R}_+ H^{\obs} + \frac{1}{2} \mathcal{R}_+ \left( \left( H^{\obs} \right)^2 \right) - \frac{1}{2} \left( \mathcal{R}_+ H^{\obs} \right)^2
\\
&\quad\quad\quad =
H^{\obs} + \frac{1}{2} \sum_{i,j=1}^n s_i s_j n^in^j \nabla^* \nabla C_+(a_i,a_j)
\end{align*}
and so
$$
\mathbf{A} H = \mathbf{A}H^{\es} + H^{\obs} + \frac{1}{2} \sum_{i,j=1}^n s_i s_j n^in^j \nabla^* \nabla C_+(a_i,a_j).
$$
Then, similar to the case of two observables,
\begin{align*}
\lambda_N^i &= \sum_{k=0}^{N-1} \left( \mathcal{B}K_k^i \right)^0,
\\
q_N^{\lbrace 1,\ldots,n\rbrace} &= \frac{1}{2} \sum_{i,j=1}^n n^i_* n^j_* \nabla^* \nabla \mathcal{C}^q(a_i,a_j),
\\
n_*^i &= n_0^i + \sum_{k=0}^{j_*-1} \left( \mathbf{B}K_k^i \right)^1.
\end{align*}

The behaviour of $n_*^i$ as $j_* \rightarrow \infty$ can be analysed as before via the single observable flow. Thus
$$
n^i_* \rightarrow n_0^i
\quad
\text{as } j^* \rightarrow \infty.
$$

In summary we get
\begin{align*}
\lim_{N \rightarrow \infty}
\Cov_{\gamma_{N,\beta}^u} \left(
\nabla_{m_1} \varphi(a_1), \ldots, \nabla_{m_n} \varphi(a_n)
\right)
&=
\frac{1}{2} \sum_{i,j=1}^n \nabla_{m_j}^* \nabla_{m_i} \mathcal{C}^q (a_i,a_j) + R_{ab},
\\
 \left| R_{ab} \right| &\leq C \frac{1}{|a-b|^{d+\nu}},
\quad \nu > 0.
\end{align*}


\end{remark}

\subsection{Proof of Theorem \ref{Thm:RepresentationExtendedPartitionFunction}}\label{Subsec:Proof_Decay}

The proof of Theorem \ref{Thm:RepresentationExtendedPartitionFunction} consists of two steps. By direct observation of the flow we get the estimate for $q_N^{ab}$ in Proposition \ref{Prop:Cov_first}. In a second step Proposition \ref{Prop:Cov_better} is used to get a refined leading term.

\begin{proof}[Proof of Theorem \ref{Thm:RepresentationExtendedPartitionFunction}]
Let $L_1$ and $\epsilon_1$ be as in Proposition \ref{Prop:Existence_ObservableFlow} with $\eta < \frac{1}{4}$. Then, by Proposition \ref{Prop:Cov_first}, \eqref{eq:formula_extended_partition_function} holds with the estimates on $Z_N^{\ext}(\mathcal{K},0)$ and on $\lambda_N^{\alpha}$ and with
$$
q_N^{ab} = \Big( \delta_{m_a} + S^a_{j_{ab}} \Big)
\Big( \delta_{m_b} + S^b_{j_{ab}} \Big)
\nabla^* \nabla C^q(a,b) + \tilde{R}_{ab}.
$$
Proposition \ref{Prop:Cov_better} gives the improved estimate as can be found in the statement of Theorem~\ref{Thm:RepresentationExtendedPartitionFunction}.
\end{proof}

\section{Proofs of extensions and intermediate steps}\label{sec:Proofs_ExtendedFlow}

Note that in this section any dependencies on $q$ are omitted since $q \in B_{\kappa}(0)$ is fixed with $\kappa$ depending on $\zeta$, $\eta$ and $L$ in Proposition \ref{Prop:Existence_BulkFlow}. As an exception we note the dependency explicitely in Lemma \ref{lemma:Properties_of_weights} since this is the place where the parameter~$\kappa$ is determined in dependence on $L$.

In this whole section $R$ is a parameter which depends only on $d$.

\subsection{Properties of the norms}

In this subsection we follow closely the presentation in \cite{ABKM}. Arguments from \cite{ABKM} which can be applied without any change to the extended setting will be omitted in proofs.

\subsubsection{Properties of the weights} \label{subsubsec:Properties_of_weights}

For the sake of completeness we review Theorem 7.1 from \cite{ABKM}.
The last scale weights ($k=N$) differ from \cite{ABKM} due to the modified definition of the last scale covariance (see \eqref{eq:Last_Scale_Cov}).
However, this does not change the properties of the weights as stated in the following lemma.

\begin{lemma}
\label{lemma:Properties_of_weights}
Let $L \geq 2^{d+3} + 16 R$. The weight functions $w_k$, $w_{k:k+1}$ and $W_k$ are well-defined and satisfy the following properties:
\begin{enumerate}
\item
For any $Y \subset X \in \mathcal{P}_k$, $0 \leq k \leq N$, and $\varphi \in \mathcal{V}_N$
\begin{align*}
w_k^Y(\varphi) \leq w_k^X(\varphi)
\quad\mathrm{and}\quad
w_{k:k+1}^Y(\varphi) \leq w_{k:k+1}^X(\varphi).
\end{align*}
\item
For any strictly disjoint polymers $X,Y \in \mathcal{P}_k$, $0 \leq k \leq N$, and $\varphi \in \mathcal{V}_N$
\begin{align*}
w_{k}^{X \cup Y}(\varphi) = w_k^X(\varphi) w_k^Y(\varphi).
\end{align*}
\item
For any polymers $X,Y \in \mathcal{P}_k$ such that $\mathrm{dist}(X,Y) \geq \frac{3}{4} L^{k+1}$, $0 \leq k \leq N$, and $\varphi \in \mathcal{V}_N$
\begin{align*}
w_{k:k+1}^{X \cup Y}(\varphi) = w_{k:k+1}^X(\varphi)w_{k:k+1}^Y(\varphi).
\end{align*}
\item
For any disjoint polymers $X,Y \in \mathcal{P}_k$, $0 \leq k \leq N$, and $\varphi \in \mathcal{V}_N$
\begin{align*}
W_k^{X \cup Y}(\varphi) = W_k^X(\varphi) W_k^Y(\varphi).
\end{align*}
\end{enumerate}
Moreover, there is a constant $h_0 = h_0(L,\zeta)$ such that for all $h \geq h_0$ the weight functions satisfy the following properties:
\begin{enumerate}
\setcounter{enumi}{4}
\item
For any disjoint polymers $X,Y \in \mathcal{P}_k$ and $U = \pi(X) \in \mathcal{P}_{k+1}$, $0 \leq k \leq N-1$, and $\varphi \in \mathcal{V}_N$
\begin{align*}
w_{k+1}^U(\varphi) \geq w_{k:k+1}^X(\varphi) \left( W_k^{U^+}(\varphi)\right)^2.
\end{align*}
\item
For all $0 \leq k \leq N-1$, $X \in \mathcal{P}_{k+1}$ and $\varphi \in \mathcal{V}_N$,
\begin{align*}
e^{\frac{|\varphi|^2_{k+1,X}}{2}} w_{k:k+1}^X(\varphi)
\leq w_{k+1}^X(\varphi).
\end{align*}
\end{enumerate}
Lastly, there exists a constant $\kappa = \kappa(L,\zeta)$ with the following properties:
\begin{enumerate}
\setcounter{enumi}{6}
\item
There is a constant $A_{\mathcal{P}}$ such that for $q \in B_{\kappa}(0)$, $\rho = (1+ \frac{\zeta}{4})^{1/3}-1$, $Y \in \mathcal{P}_k$, $0 \leq k \leq N$, and $\varphi \in \mathcal{V}_N$
\begin{align*}
\left( \int_{\chi_N} \left( w_k^X(\varphi + \xi) \right)^{1+\rho} \mu_{k+1}(\de \xi) \right)^{\frac{1}{1+\rho}}
\leq \left( \frac{A_{\mathcal{P}}}{2} \right)^{|X|_k} w_{k:k+1}^X(\varphi).
\end{align*}
\item
There is a constant $A_{\mathcal{B}}$ independent of $L$ such that for $q \in B_{\kappa}$, $\rho = (1+\frac{\zeta}{4})^{1/3}-1$, $B \in \mathcal{B}_k$, $0 \leq k \leq N$, and $\varphi \in \mathcal{V}_N$
\begin{align*}
\left( \int_{\chi_N} \left( w_k^B(\varphi + \xi) \right)^{1+\rho} \mu_{k+1}(\de \xi) \right)^{\frac{1}{1+\rho}}
\leq \frac{A_{\mathcal{B}}}{2}  w_{k:k+1}^B(\varphi).
\end{align*}
\end{enumerate} 
\end{lemma}



\subsubsection{Pointwise properties of the norms}

The following lemma is an extension to observables of Lemma 8.1 from \cite{ABKM}.

\begin{lemma}
\label{lemma:submultiplicativity1}
Assume that $F,G \in \mathcal{N}$, $X \in \mathcal{P}_k$ and $F(\varphi)$ and $G(\varphi)$ depend only on $\varphi \vert_{X^*}$. Assume furthermore that $F(\varphi + \psi) = F(\varphi), G(\varphi + \psi)=G(\varphi)$ if $\psi\vert_{X^*}$ is constant. Then
\begin{align*}
|FG|_{k,X,T_{\varphi}}^{\ext} \leq |F|_{k,X,T_{\varphi}}^{\ext} |G|_{k,X,T_{\varphi}}^{\ext}
\end{align*}
and, for $X \in \mathcal{P}_k$ and $\alpha \in \lbrace \es,a,b,ab \rbrace$,
\begin{align*}
\vert F^{\alpha} \vert_{k+1,X,T_{\varphi}}
\leq \left( 1 + |\varphi|_{k+1,X} \right)^3 \left( |F^{\alpha}|_{k+1,X,T_0} + 16 L^{-\frac{3}{2}d} \sup_{0 \leq t \leq 1} |F^{\alpha}|_{k,X,T_{t \varphi}} \right).
\end{align*}
\end{lemma}

\begin{proof}
We write the extended norm as sum $|K|^{\ext}_{k,X,T_{\varphi}} = \sum_{\alpha} l^{|\alpha|}_{\obs,k} |K^{\alpha}|_{k,X,T_{\varphi}}$ and apply Lemma 8.1. from \cite{ABKM} on each (bulk) norm $\vert F^{\alpha} G^{\beta} \vert_{k,X,T_{\varphi}}$. This yields that
$$
l_{\obs,k}^{|\alpha|+|\beta|} \left|F^{\alpha} G^{\beta}\right|_{k,X,T_{\varphi}}
\leq \left( l_{\obs,k}^{|\alpha|} \vert F^{\alpha} \vert_{k,X,T_{\varphi}}\right)
\left( l_{\obs,k}^{|\beta|}  \vert G^{\beta} \vert_{k,X,T_{\varphi}} \right).
$$
Thus
\begin{align*}
\vert FG \vert_{k,X,T_{\varphi}}
&=
\sum_{\alpha} l^{\alpha}_{\obs} \vert (FG)^{\alpha} \vert_{k,X,T_{\varphi}}
\leq 
\left(\sum_{\alpha}l_{\obs,k}^{|\alpha|} |F^{\alpha}|_{T_{\varphi}}\right)\left(\sum_{\alpha}l_{\obs,k}^{|\alpha|} |G^{\alpha}|_{T_{\varphi}}\right)
\end{align*}
since
\begin{align*}
FG &= F^{\es} G^{\es}
+ s \left(F^{a} G^{\es} + F^{\es} G^{a}\right)
+ t \left(F^{b} G^{\es} + F^{\es} G^{b}\right)
\\ & \quad \quad\quad
+ st \left(F^{\es} G^{ab} + F^{ab} G^{\es} + F^{a} G^{b} + F^{b} G^{a}\right).
\end{align*}
This proves the first inequality. 
The second inequality is the same as in \cite{ABKM}.
\end{proof}

The following statement is an extension to observables of Lemma 8.2 from \cite{ABKM}. 
\begin{lemma}
\label{lemma:submultiplicativity2}
Let $\varphi \in \chi_N$. Then
\begin{enumerate}
\item
for any $F_1,F_2 \in M^{\ext}(\mathcal{P}_k)$ and any $X_1,X_2 \in \mathcal{P}_k$ we have
\begin{align*}
| F_1(X_1) F_2(X_2)|^{\ext}_{k,X_1 \cup X_2,T_{\varphi}}
\leq
| F_1(X_1)|^{\ext}_{k,X_1,T_{\varphi}} 
| F_2(X_2)|^{\ext}_{k,X_2,T_{\varphi}};
\end{align*}
\item
for any $F \in M^{\ext}(\mathcal{P}_k)$ and any polymer $X \in \mathcal{P}_k$ the bound
\begin{align*}
| F(X)|_{k+1,\pi(X),T_{\varphi}}
&\leq \max\left\lbrace 1, \frac{\eta^2}{4} L^d\right\rbrace |F(X)|_{k,X \cup \pi(X),T_{\varphi}}
\\
&\leq \max\left\lbrace 1, \frac{\eta^2}{4} L^d\right\rbrace |F(X)|_{k,X,T_{\varphi}}
\end{align*}
holds if $L \geq 2^d + R$.
\end{enumerate}
\end{lemma}

In 2., the factor $\frac{\eta^2}{4} L^d$ is new in comparison to \cite{ABKM}.

\begin{proof}
The first inequality follows from Lemma \ref{lemma:submultiplicativity1} and the estimate
$$
|F(X)|_{k,X\cup Y,T_{\varphi}} \leq |F(X)|_{k,X,T_{\varphi}}.
$$ as in \cite{ABKM}.

For the second inequality note that due to the change of scale we have an additional factor
$$
\frac{l_{\obs,k+1}^{|\alpha|}}{l_{\obs,k}^{|\alpha|}}
\leq \frac{\eta^2}{4} L^d
$$
for $|\alpha| = 1,2$, which appears in the stated inequality. The remaining steps are as in \cite{ABKM}.
\end{proof}

\subsubsection{Submultiplicativity of the norms}

The following claim is based on Lemma 8.3 in \cite{ABKM}, extended to observables. 

\begin{lemma}
\label{lemma:norm _estimate}
Let $L \geq 2^{d+3} + 16 R$ be an odd integer and $h \geq h_0(L)$, where $h_0$ is fixed in Lemma \ref{lemma:Properties_of_weights}. For $k \in \lbrace 0, \ldots, N-1 \rbrace$, let $K \in M^{\ext}(\mathcal{P}_k)$ factor at scale $k$ and let $F \in M(\mathcal{B}_k)$. Then the following bounds hold:
\begin{enumerate}
\item
$\Vert K(X) \Vert^{\ext}_{k,X} \leq \prod_{Y \in \mathcal{C}(X)} \Vert K(Y) \Vert_{k,Y}^{\ext}$ and

$\Vert K(X) \Vert^{\ext}_{k:k+1,X} \leq \prod_{Y \in \mathcal{C}(X)} \Vert K(Y) \Vert_{k:k+1,Y}^{\ext}$

and more generally the same bounds hold for any decomposition $X = \bigcup Y_i$ such that the $Y_i$ are strictly disjoint.
\item
$\Vert F^X K(Y) \Vert_{k,X \cup Y} \leq \Vert K(Y) \Vert_{k,Y} \vertiii{F}_k^{|X|_k}$ for $X,Y \in \mathcal{P}_k$ with $X$ and $Y$ disjoint.
\item For any polymers $X,Y,Z_1,Z_2 \in \mathcal{P}_k$ such that $X \cap Y = \es$, $Z_1 \cap Z_2 =  \es$, and $Z_1,Z_2 \subset \pi(X \cup Y) \cup X \cup Y$,
\begin{align*}
&\Vert F_1^{Z_1} F_2^{Z_2} F_3^{X} K(Y) \Vert_{k+1,\pi(X \cup Y)}
\\ & \qquad\qquad
\leq
\max\left\lbrace 1, \frac{\eta^2}{4} L^d\right\rbrace \Vert K(Y) \Vert_{k:k+1,Y} \vertiii{F_1}_k^{|Z_1|_k} \vertiii{F_2}_k^{|Z_2|_k} \vertiii{F_1}_k^{|X|_k}.
\end{align*}
\item
$\vertiii{\1(B)}_{k,B} = 1$ for $B \in \mathcal{B}_k$.
\end{enumerate}
\end{lemma}

In 3., the factor $\frac{\eta^2}{4} L^d$ is new in comparison to\cite{ABKM}.

\begin{proof}
Ingredients for the proof are the submultiplicativity of the $T_{\varphi}$-seminorm in Lemma \ref{lemma:submultiplicativity2} and properties of the weights. Since the submultiplicativity also holds for extended functionals the proof is exactly the same as in \cite{ABKM}. The new factor $\frac{\eta^2}{4} L^d$ appears in the transition from one scale to the next one using \eqref{scale_transformation}.
\end{proof}

\subsubsection{Regularity of the integration map}

We extend Lemma 8.4 from \cite{ABKM} to observables. 

\begin{lemma}
\label{lemma:Regularity_of_Integration_Map}
Let $L \geq 2^{d+3} + 16 R$ and let $A_{\mathcal{P}}$ be the constant from Lemma \ref{lemma:Properties_of_weights}. Then
\begin{align*}
\Vert \mathcal{R}_{k+1} K(X) \Vert_{k:k+1,X}^{\ext} \leq \left( \frac{A_{\mathcal{P}}}{2}\right)^{|X|_k} \Vert K(X) \Vert_{k,X}^{\ext}.
\end{align*}
If $X$ is a block the constant is $A_{\mathcal{B}}$ which is independent of $L$.
\end{lemma}

\begin{proof}
The proof in \cite{ABKM} does not use any special property of the $T_{\varphi}$-seminorm, so it works exactly as in \cite{ABKM}.
\end{proof}

For later reference we state the following inequality which appears in the proof of Lemma 8.4 from \cite{ABKM}.
\begin{lemma}
Assume that Lemma \ref{lemma:Regularity_of_Integration_Map} holds. Then
\begin{align}
\vert \mathcal{R}_{k+1} K(X) \vert_{k,X,T_{\varphi}}^{\ext}
\leq \Vert K(X) \Vert_{k,X}^{\ext} \left( \frac{A_{\mathcal{P}}}{2} \right)^{|X|_k} w_{k:k+1}^X(\varphi).
\label{integration_estimate}
\end{align}
If $X$ is a block the constant is $A_{\mathcal{B}}$ which is independent of $L$.
\end{lemma}

\subsubsection{The extended projection $\Pi_k$ to relevant Hamiltonians} \label{subsubsec:Projection}

We extend the space of relevant Hamiltonians to observables.

Let $\mathcal{U}= \lbrace e_1, \ldots, e_d \rbrace$. The monomials which appear in \cite{ABKM} are
\begin{align*}
M(\lbrace x \rbrace)_{\es}(\varphi) &= 1, \\
M(\lbrace x \rbrace)_{\beta}(\varphi) &= \nabla^{\beta} \varphi(x), \\
M(\lbrace x \rbrace)_{\beta,\gamma}(\varphi) &= \nabla^{\beta} \varphi(x) \nabla^{\gamma} \varphi(x).
\end{align*}

Then the corresponding index sets are
\begin{align*}
\mathfrak{v}_0 &= \lbrace\es\rbrace, \\
\mathfrak{v}_1 &= \lbrace \beta: \beta \in \N_0^{\mathcal{U}}, 1 \leq |\beta| \leq \lfloor d/2 \rfloor + 1 \rbrace, \\
\mathfrak{v}_2 &= \lbrace (\beta,\gamma): \beta,\gamma \in \N_0^{\mathcal{U}}, |\beta| = |\gamma| = 1, \beta < \gamma \rbrace.
\end{align*}
Here, $\beta < \gamma$ refers to any ordering of $\mathcal{U}$. We additionally define
\begin{align*}
\mathfrak{v}_0^{\alpha} &= \lbrace\es\rbrace, \quad\alpha \in \lbrace a,b,ab \rbrace, \\
\mathfrak{v}_1^{\alpha} &= \lbrace \beta \in \N_0^{\mathcal{U}}: |\beta| = 1 \rbrace, \quad\alpha \in\lbrace a,b \rbrace.
\end{align*}
We set
$$
\mathfrak{v}^{\ext} = \mathfrak{v}_0 \cup \mathfrak{v}_1 \cup \mathfrak{v}_2 \cup \mathfrak{v}_0^a \cup \mathfrak{v}_1^a \cup \mathfrak{v}_0^b \cup \mathfrak{v}_1^b \cup\mathfrak{v}_0^{ab}.
$$
The space of relevant Hamiltonians is given by
$$
\mathcal{V}^{\ext} = \mathcal{V}_0 \oplus \mathcal{V}_1 \oplus \mathcal{V}_2 \oplus \mathcal{V}_0^a \oplus \mathcal{V}_1^a \oplus \mathcal{V}_0^b \oplus \mathcal{V}_1^b \oplus\mathcal{V}_0^{ab}
$$
where
\begin{align*}
\mathcal{V}_0 &= \R , 
\\
\mathcal{V}_1 &= \text{span} \lbrace M_{\mathfrak{m}}(B): \mathfrak{m} \in \mathfrak{v}_1 \rbrace, 
\\
\mathcal{V}_2 &= \text{span} \lbrace M_{\mathfrak{m}}(B): \mathfrak{m} \in \mathfrak{v}_2 \rbrace ,
\\
\mathcal{V}_0^{\alpha} &= \R, \quad \alpha \in \lbrace a,b,ab \rbrace, 
\\
\mathcal{V}_1^{\alpha} &= \text{span} \lbrace M_{\mathfrak{m}}(\lbrace\alpha \rbrace): \mathfrak{m} \in \mathfrak{v}_1^{\alpha} \rbrace, \quad \alpha \in \lbrace a,b \rbrace.
\end{align*}
As in \cite{ABKM}, we set
\begin{align*}
b_{\beta}(z) = \binom{z_1}{\beta_1} \cdots \binom{z_d}{\beta_d},
\quad z \in \Z^d, \quad \beta \in \N_0^{\lbrace 1, \ldots, d \rbrace}.
\end{align*}
We extend the basis for polynomials on $\Z^d$ for $\alpha \in \lbrace a,b \rbrace$ by
$$
b_{\beta}^{\alpha}(z) = \binom{z_1 - \alpha_1}{\beta_1} \ldots \binom{z_d - \alpha_d}{\beta_d}.
$$
Using these functions we can extend the space $\mathcal{P}$ in \cite{ABKM} to observables by defining
\begin{align*}
\mathcal{P}_0^{\alpha} &= \R, \quad \alpha \in \lbrace a,b,ab \rbrace,
\\
\mathcal{P}_1^{\alpha} &= \text{span} \lbrace b_{\beta}^{\alpha}: \beta \in \mathfrak{v}_1^{\alpha} \rbrace,
\quad \alpha \in \lbrace a,b \rbrace,
\end{align*}
and setting
\begin{align*}
\mathcal{P}^{\ext} = \mathcal{P} \oplus \mathcal{P}_0^a \oplus \mathcal{P}_1^a \oplus \mathcal{P}_0^b \oplus \mathcal{P}_1^b \oplus \mathcal{P}_0^{ab}.
\end{align*}
Now we can formulate the extension of Lemma 8.5 from \cite{ABKM}. The notation $\langle F,g\rangle_{\varphi} = \langle \text{Tay}_{\varphi} F,g \rangle$ is used, as in \cite{ABKM}.

\begin{lemma}
\label{lemma:existence_of_projection}
Let $K \in M^{\ext}(\mathcal{P}_k^c,\chi_N)$ and let $B \in \mathcal{B}_k$. Then there exists one and only one $H \in \mathcal{V}^{\ext}$ such that
$$
\langle H,g \rangle_0 = \langle K(B),g \rangle_0 
\quad \text{for all } g \in \mathcal{P}^{\ext}.
$$
More precisely, for $\alpha \in \lbrace a,b \rbrace$,
$$
H^{\alpha}(\varphi) = K^{\alpha}(0) + n^{\alpha} \nabla\varphi(\alpha),
$$  where
\begin{align}
n_{\gamma}^{\alpha} = \langle K^{\alpha}(B),b_{\gamma}^{\alpha} \rangle_0
\quad \text{for all } \gamma \in \mathfrak{v}_1^{\alpha}
\label{eq:formula_linear_projection}
\end{align}
and 
$$
q^{ab} = K^{ab}(0).
$$
\end{lemma}

\begin{defn}
We define $\Pi K(B) = H$ where $H$ is given by Lemma \ref{lemma:existence_of_projection}.
\end{defn}

\begin{proof}[Proof of Lemma \ref{lemma:existence_of_projection}]
The bulk part of $K$ is handled in \cite{ABKM}. The constant observable part of $H \in \mathcal{V}^{\ext}$ is given by
\begin{align*}
\lambda^a = K^a(B,0),
\quad
\lambda^b = K^b(B,0),
\quad
q^{ab} = K^{ab}(B,0).
\end{align*}
We turn to the linear observable part of $H$. We claim that for $\alpha \in \lbrace a,b \rbrace$  there is a unique $H^{1,\alpha}~\in~\mathcal{V}_1^{\alpha}$ such that
\begin{align*}
\langle H^{1,\alpha},g \rangle_0
= \langle K^a(B),g \rangle_0
\quad\text{for all } g \in \mathcal{P}_1^{\alpha}.
\end{align*}
An element $H^{1,\alpha} \in \mathcal{V}_1^{\alpha}$ is of the form $ \sum_{\beta \in \mathfrak{v}_1^{\alpha}} n_{\beta}^{\alpha} M_{\beta}(\lbrace \alpha \rbrace)$ for some $n_{\beta}^{\alpha}$ yet to be determined.

Testing against the basis $\lbrace b_{\beta}^{\alpha}: \beta \in \mathfrak{v}_1^{\alpha}\rbrace$ of $\mathcal{P}_1^{\alpha}$ we have to show that there is a family $\left(n_{\beta}^{\alpha}\right)_{\beta \in \mathfrak{v}_1^{\alpha}}$ such that
\begin{align*}
\sum_{\beta \in \mathfrak{v}_1^{\alpha}} n_{\beta}^{\alpha}
\langle M_{\beta}(\lbrace\alpha \rbrace), b_{\gamma}^{\alpha} \rangle_0
= \langle K^{\alpha}(B), b_{\beta}^{\alpha} \rangle_0
\quad \text{for all } \gamma \in \mathfrak{v}_1^{\alpha}.
\end{align*}
The last equality is equivalent to
\begin{align*}
\sum_{\beta \in \mathfrak{v}_1^{\alpha}} n_{\beta}^{\alpha}
B_{\beta\gamma}
= \langle K^{\alpha}(B), b_{\beta}^{\alpha} \rangle_0
\quad \text{for all } \gamma \in \mathfrak{v}_1^{\alpha}
\end{align*}
with
$$
B_{\beta\gamma} = \langle \nabla^{\beta}\varphi(\alpha), b_{\gamma}^{\alpha} \rangle_0
= \langle \text{Tay}_0 \nabla^{\beta} \varphi(\alpha), b_{\gamma}^{\alpha} \rangle
= \nabla^{\beta} b_{\gamma}^{\alpha}(\alpha)
= b_{\gamma - \beta}^{\alpha}(\alpha).
$$
For $\beta, \gamma \in \mathfrak{v}_1^{\alpha}$ we get that $B_{\beta,\gamma} = \1_{\beta = \gamma}$ and thus
\begin{align*}
n_{\gamma}^{\alpha} = \langle K^{\alpha}(B),b_{\gamma}^{\alpha} \rangle_0
\quad \text{for all } \gamma \in \mathfrak{v}_1^{\alpha}.
\end{align*}
\end{proof}

The following statement is an extension to observables of Lemma 8.7 from \cite{ABKM}. 

\begin{lemma}
\label{Boundedness_of_Projection}
There exists a constant $C$ such that for $L\geq 2^d + R$ and $0 \leq k \leq N-1$
$$
\Vert \Pi_k K(B) \Vert_{k,0}^{\ext} \leq C |K(B)|_{k,B,T_0}^{\ext}.
$$
\end{lemma}

\begin{proof}
The bulk part of the estimate is done in \cite{ABKM}. What remains to prove is
$$
\Vert \Pi_k^{\alpha} K^{\alpha}(B) \Vert_{k,0}^{\alpha}
\leq C l^{|\alpha|}_{\obs,k} |K^{\alpha}(B)|_{k,B,T_0}.
$$
Since for the constant part of the projection we have $\lambda^{\alpha} = K^{\alpha}(B,0)$ for $\alpha \in \lbrace a,b \rbrace$ and $q^{ab} = K^{ab}(B,0)$ we just have to estimate the coefficients $n^{\alpha}$ of the linear part of the projection.

Since $n^{\alpha} = \langle K^{\alpha}(B), b^{\alpha} \rangle_0$ (see \eqref{eq:formula_linear_projection} in Lemma \ref{lemma:existence_of_projection}) we have to show that
$$
l_{\obs,k} l_k \vert \langle K^{\alpha}(B), b^{\alpha} \rangle_0 \vert
\leq C l_{\obs,k} |K^{\alpha}(B)|_{k,B,T_0}.
$$
However, this follows directly from the definition of the $T_{\varphi}$-seminorm and since $|b^{\alpha}|_{k,B}~=~l^{-1}_k$:
\begin{align*}
 \langle K^{\alpha}(B), b^{\alpha} \rangle_0
\leq |b^{\alpha}|_{k,B}  \sup_{|g|_{k,B} \leq 1} \langle K^{\alpha}(B),g\rangle_0
\leq l^{-1}_k \vert K^{\alpha}(B) \vert_{k,B,T_0}.
\end{align*}

\end{proof}

We extend Lemma 8.8 from \cite{ABKM} to observables.
\begin{lemma}
\label{lemma:relevant_variable}
For $H \in M_0^{\ext}$, $L \geq 3$, and $0 \leq k \leq N$ we have
$$
|H|_{T_{\varphi}}^{\ext} 
\leq (1 + |\varphi|_{k,B})^2 \Vert H \Vert_{k,0}^{\ext}
\leq 2 (1 + |\varphi|^2_{k,B}) \Vert H \Vert_{k,0}^{\ext}.
$$
\end{lemma}

\begin{proof}
The only difference to \cite{ABKM} is that additional terms in $|H|_{T_{\varphi}}^{\ext} $ and $\Vert H \Vert_{k,0}^{\ext}$ exist:
\begin{align*}
|H|_{T_{\varphi}}^{\ext}
&= |H^{\es}|_{T_{\varphi}}
+ l_{\obs,k} \left( |\lambda^a| + |n^a \nabla \varphi(a) \1_a|_{T_{\varphi}} \right)
\\
& \quad\quad \quad\quad +
 l_{\obs,k} \left( |\lambda^b| + |n^b \nabla \varphi(b) \1_b|_{T_{\varphi}} \right)
+ l_{\obs,k}^2 |q^{ab}|,
\\
\Vert H \Vert_{k,0}^{\ext}
&=
\Vert H^{\es} \Vert_{k,0}
+ l_{\obs,k} \left( |\lambda^a| + l_k |n^a| \right)
+ l_{\obs,k} \left( |\lambda^b| + l_k |n^b| \right)
+ l_{\obs,k}^2 |q^{ab}|.
\end{align*}
Thus the proof is finished if we show that, for $\alpha \in \lbrace a,b \rbrace$,
\begin{align*}
l_{\obs,k} |n^{\alpha} \nabla \varphi(\alpha) \1_{\alpha}(B)|_{T_{\varphi}}
\leq (1 + |\varphi|_{k,B})^2 l_{\obs,k} l_k |n^{\alpha}|.
\end{align*}
This follows straightforwardly since
\begin{align*}
 |\nabla \varphi(\alpha) \1_{\alpha}(B)|_{T_{\varphi}}
=  (|\nabla \varphi(a) | + l_k) \1_{\alpha}(B)
\leq l_k |\varphi |_{k,B} + l_k
\leq l_k \left(1 + | \varphi |_{k,B}^2\right).
\end{align*}
\end{proof}

The following lemma is an extension of Lemma 8.9 from \cite{ABKM}. 

\begin{lemma}
\label{lemma:Contraction_estimate_1}
Let $A(\alpha,k) = 0$ when $k \geq j_{ab}$, $\alpha\in \lbrace a,b,ab\rbrace$, and $A(\alpha,k) = 1$ when $k < j_{ab}$, $\alpha \in \lbrace a,b\rbrace$.
There exists a constant $C$ such that for $L \geq 2^d + R$, for $\alpha \in \lbrace a,b,ab \rbrace$, 
\begin{align*}
\vert (1-\Pi_k^{\alpha}) K^{\alpha}(B) \vert_{k+1,B,T_0}
\leq
C L^{-(d/2 + A(\alpha,k))} \vert K^{\alpha} \vert_{k,B,T_0}.
\end{align*}
\end{lemma}

\begin{proof}
We start with $\alpha \in \lbrace a,b,ab \rbrace$ and $k \geq j_{ab}$, i.e. $\Pi_k^{\alpha} = \Pi_0$. Note that
$$
\left\vert (1 - \Pi_0) K^{\alpha}(B) \right\vert_{k+1,B,T_0}
= \sup \left\lbrace \left\langle (1-\Pi_0) K^{\alpha},g \right\rangle_0: g \in \Phi, |g|_{k+1,B} \leq 1 \right\rbrace.
$$
For $g \in \chi^{\otimes r}$, $r \geq 1$, it holds that
$$
\left\langle \left(1-\Pi_0\right) K^{\alpha},g\right\rangle_0 
= \left\langle K^{\alpha},g \right\rangle_0
$$
since $\Pi_0 K^{\alpha}$ depends only on the first order Taylor polynomial. For $g \in \chi^{\otimes r}$, $r \geq 1$, we can use the estimate
$$
|g|_{k,B} \leq 8 L^{-\frac{1}{2} d}|g|_{k+1,B}
$$
 as in \cite{ABKM}.
Thus
$$
\left\vert \left\langle (1-\Pi_0) K^{\alpha},g\right\rangle_0 \right\vert
\leq \vert K^{\alpha} \vert_{k,B,T_0} |g|_{k,B}
\leq 8 L^{-\frac{1}{2} d}|g|_{k+1,B} \vert K^{\alpha} \vert_{k,B,T_0}.
$$
For $g \in \chi^{\otimes 0} = \R = \mathcal{P}_0^{\alpha}$ it holds that
$$
\left\langle \Pi_0 K^{\alpha} , g \right\rangle_0
= \left\langle K^{\alpha},g \right\rangle_0
$$
and thus
$$
\left\langle (1-\Pi_0) K^{\alpha},g \right\rangle_0
= 0 \quad \text{for all } g \in \R.
$$
This argument finishes the case $k \geq j_{ab}$.

Now let $\alpha \in \lbrace a,b \rbrace$ and $k < j_{ab}$, i.e., $\Pi_k^{\alpha} = \Pi_1$. As above we can use for all $g \in \chi^{\otimes r}$ and $r \geq 2$
$$
\left\vert \left\langle (1-\Pi_1) K^{\alpha},g \right\rangle_0 \right\vert 
= \left\vert \left\langle K^{\alpha},g \right\rangle_0 \right\vert
\leq 8 L^{-\frac{1}{2} d} \vert K^{\alpha} \vert_{k,B,T_0} |g|_{k+1,B}.
$$
Again,
$$
\left\langle (1-\Pi_0) K^{\alpha},g\right\rangle_0 = 0
\quad
\text{for all }
g \in \R = \mathcal{P}_0^{\alpha}.
$$
Let $\varphi \in \chi$. For all $P \in \mathcal{P}_1^{\alpha}$ we have $\langle \Pi_1 K^{\alpha},P \rangle_0 = \langle K^{\alpha},P \rangle_0$. Using additionally boundedness of $\Pi$, we can estimate
\begin{align*}
\left\vert
\left\langle (1-\Pi_1) K^{\alpha},\varphi \right\rangle_0
\right\vert
&= \min_{P \in \mathcal{P}_1^{\alpha}}
\left\vert
\left\langle (1-\Pi_1) K^{\alpha},\varphi - P \right\rangle_0
\right\vert
\\ &
\leq \vert (1-\Pi_1) K^{\alpha} \vert_{k,B,T_0} \min_{P \in \mathcal{P}_1^{\alpha}} |\varphi - P|_{k,B}
\\ &
\leq C \vert K^{\alpha} \vert_{k,B,T_0} \min_{P \in \mathcal{P}_1^{\alpha}} |\varphi - P|_{k,B}.
\end{align*}
With Lemma \ref{lemma:...} below the proof is finished.
\end{proof}

\begin{lemma}
\label{lemma:...}
There exists a constant $C$ such that for $L \geq 2^d + R$ and for all $\varphi \in \chi$
\begin{align*}
\min_{P \in \mathcal{P}_1^{\alpha}} |\varphi - P|_{k,B}
\leq C L^{-(\frac{d}{2} + 1)} \vert \varphi \vert_{k+1,B}.
\end{align*}
\end{lemma}

\begin{proof} The statement is an extension of Lemma 8.10 from \cite{ABKM}.
The proof is as in \cite{ABKM} with the only difference being the choice of parameter $s = 1$, which originally was $s = \left\lfloor \frac{d}{2} \right\rfloor + 1$. The reason for this change is that $\mathcal{P}_1^{\alpha} = \text{span} \left\lbrace b_{\beta}^{\alpha} : |\beta| = 1 \right\rbrace$, whereas in the bulk flow higher derivatives are also allowed. Then $P = \text{Tay}_a^s \varphi$ provides the minimizer. 
\end{proof}

\subsection{Smoothness of the extended renormalisation map} \label{subsec:Proof_Smoothness}

In this section we prove Proposition \ref{Prop:Smoothness_Observables} which claims that there is $L_0$ and corresponding $A_0$ and $h_0$ and a parameter $\rho^*(A)$ such that $\mathbf{S}_k^{\ext} \in U_{\rho^*(A)}$ with bounds on derivatives which are uniformly in $N$. 

~\\
Remember that
$$
\mathbf{S}^{\ext}(H,K)
=e^{
-s\left( \mathbf{B} K^a \right)^0
- t \left( \mathbf{B} K^b \right)^0
- st \left( \int H^a H^b \de\mu_+ + \mathbf{B} K^{ab} \right)
}
\mathbf{S}(H,K)
$$
where we drop the subscript $k$ and $k+1$ in the notation. To nevertheless note the change of scale, we abbreviate $k+1$ by $+$.

Let us denote
$$
F = sF^a + t F^b + st F^{ab} :=  -s(\mathbf{B} K^a)^0 - t (\mathbf{B} K^b)^0 - st \left( \int H^a H^b \de \mu_+ +\mathbf{B} K^{ab}\right).
$$
We divide the proof of Proposition \ref{Prop:Smoothness_Observables} into two steps.
The first step is the analysis of $\mathbf{S}$.

\begin{lemma} \label{lemma:estimate_on_S}
There is $L_0$ such that for all odd integers $L \geq L_0$ there is $A_0,h_0$ with the following property. For all $A \geq A_0$, $h \geq h_0$ there is $\rho^* = \rho^*(A)$ such that
$$
\mathbf{S} \in C^{\infty} \left( U_{\rho^*}, M^{\ext}(\mathcal{P}_{k+1}^c) \right)
$$
and for any $p,q \in \N$ there is a constant $C_{p,q} = C_{p,q}(L,h,A)$ such that for any $(H,K) \in U_{\rho^*}$
 \begin{align*}
 \left\Vert D_H^p D_K^q \mathbf{S}(H,K)(\dot{H}^p, \dot{K}^q) \right\Vert_{k+1}^{(A),\ext}
 \leq
 C_{p,q} \left( \Vert \dot{H} \Vert_{k,0}^{\ext} \right)^p \left( \Vert \dot{K} \Vert_k^{(A),\ext} \right)^q.
 \end{align*}

 \end{lemma}

The second step includes the analysis of the prefactor $e^F$.

\begin{lemma} \label{lemma:estimate_on_prefactor_F}
Assume that Lemma \ref{lemma:estimate_on_S} holds. Then 
$$
\mathbf{S}^{ext} \in C^{\infty}\left( U_{\rho^*},M^{\ext}(\mathcal{P}_{k+1}^c)\right)
$$
and for each $p,q \in \N$ there is a constant $C^*_{p,q}$ such that for any $(H,K) \in U_{\rho^*}$,
$$
\left\Vert D_H^p D_K^q \mathbf{S}^{\ext}(H,K)(\dot{H}^p,\dot{K}^q) \right\Vert_{k+1}^{(A),\ext}
\leq
C^*_{p,q} \left( \Vert \dot{H} \Vert_{k,0}^{\ext} \right)^p \left( \Vert \dot{K} \Vert_k^{(A),\ext} \right)^q.
$$
\end{lemma}

Proposition \ref{Prop:Smoothness_Observables} follows from Lemma~\ref{lemma:estimate_on_prefactor_F} with the assumptions of Lemma \ref{lemma:estimate_on_S}.

~\\
We first prove Lemma \ref{lemma:estimate_on_prefactor_F}.

\begin{proof}[Proof of Lemma \ref{lemma:estimate_on_prefactor_F}]
We show smoothness via bounds on the derivatives.

Since $F$ is a constant in $\varphi$, we can estimate
\begin{align*}
&\Big\Vert D_H^p D_K^q \mathbf{S}^{\ext}(H,K)(\dot{H}^p,\dot{K}^q) \Big\Vert_{k+1}^{(A),\ext}
=
\Big\Vert D_H^p D_K^q \left[ e^{F} \mathbf{S}(H,K) \right](\dot{H}^p,\dot{K}^q) \Big\Vert_{k+1}^{(A),\ext}
\\
& \quad\quad
 \leq C_{p,q}
\sum_{\substack{p_1 + p_2=p\\ q_1 + q_2=q}}
\Big\Vert D_H^{p_1} D_K^{q_1} \left[ e^{F} \right] (\dot{H}^{p_1},\dot{K}^{q_1})
D_H^{p_2} D_K^{q_2} \mathbf{S}(H,K) (\dot{H}^{p_2},\dot{K}^{q_2}) \Big\Vert_{k+1}^{(A),\ext}
\\
& \quad\quad
\leq C_{p,q}
\sum_{\substack{p_1 + p_2=p\\ q_1 + q_2=q}} \sup_U \bigg\lbrace
 A^{|U|_{k+1}}
 \Big\vert 
 D_H^{p_1} D_K^{q_1} \left[ e^{F (U)} \right] (\dot{H}^{p_1},\dot{K}^{q_1})  \Big\vert^{\ext}_{k+1,U,T_{0}} 
 \\ &   \qquad \qquad \qquad \qquad \qquad \qquad
 \Big\Vert
 D_H^{p_2} D_K^{q_2} \mathbf{S}(H,K)(U) (\dot{H}^{p_2},\dot{K}^{q_2}) \Big\Vert_{k+1,U}^{\ext} \bigg\rbrace.
\end{align*}
By assumption $\mathbf{S}$ is smooth with the desired bounds, so it is enough to show that
$$
\Big\vert 
 D_H^{p_1} D_K^{q_1} \left( e^{F(U)} \right) (\dot{H}^{p_1},\dot{K}^{q_1})  \Big\vert^{\ext}_{k+1,U,T_{0}}
 \leq C \left( \Vert \dot{H}\Vert_{k,0}^{\ext} \right)^{p_1}
 \left( \Vert \dot{K}\Vert_{k}^{(A),\ext} \right)^{q_1}.
$$

Note that if $a,b \notin U$ then $e^{F(U)}=1$ such that any derivative $D_H^{p_1}$ or $ D_K^{q_1}$ gives just zero which is not optimal for the supremum. Thus either $a,b \notin U$ and $p_1=q_1=0$ or $\alpha \in U$ for $\alpha \in \lbrace a,b,ab \rbrace$.
In the first case we are done -- the constant we get is $1$. In the second case we go through all possible cases. Let $(H,K) \in U_{\rho^*}$.
\begin{itemize}
\item $p_1=0$, $q_1=0$:

We use Lemma \ref{lemma:Estimate_for_B}, Lemma \ref{lemma:Estimate_for_A} and estimate \eqref{scale_transformation} to get
\begin{align*}
\big\vert e^{F(U)} \big\vert^{\ext}_{k+1,U,T_0}
= 
1 &+ \vert F^a(U) \vert + \vert F^b(U) \vert + \vert F^{ab}(U) \vert + \vert F^a(U) F^b(U) \vert
\\
= 1 &+ l_{\obs,k+1} \left( \big\vert(\mathbf{B}K^a)^0\big\vert + \big\vert(\mathbf{B}K^b)^0\big\vert \right) 
\\
& 
+ l_{\obs,k+1}^2 \left( \left\vert\mathbf{B}K^{ab}\right\vert+\left\vert\int H^a H^b \de \mu_+\right\vert + \left\vert(\mathbf{B}K^a)^0 (\mathbf{B}K^b)^0\right\vert \right)
\\
\leq
1 &+ \frac{A_{\mathcal{B}}}{2} L^{d/2} \frac{\eta}{2} \rho^*
+ L^d \frac{\eta^2}{4} \rho^* \left( \frac{A_{\mathcal{B}}}{2} + \frac{A_{\mathcal{B}}^2}{4} \rho^* + C_{FRD} h^{-2} \rho^* \right)
\end{align*}
which is bounded by a constant.

\item $p_1=0$, $q_1=1$:
By  Lemma \ref{lemma:Estimate_for_B} and estimate \eqref{scale_transformation} we get
\begin{align*}
&\left\vert D_K e^{F(U)} \dot{K} \right\vert^{\ext}_{k+1,U,T_0}
\\
& \quad
=
l_{\obs,k+1} \left( \big\vert (\mathbf{B}\dot{K}^a)^0 \big\vert + \big\vert (\mathbf{B}\dot{K}^b)^0 \big\vert \right)
\\ & \quad \qquad 
+ l_{\obs,k+1}^2 \left( \big\vert (\mathbf{B}\dot{K}^{ab})^0  \big\vert + |(\mathbf{B} K^a)^0 | \big\vert (\mathbf{B}\dot{K}^b)^0 \big\vert + |(\mathbf{B}K^b)^0 | \big\vert (\mathbf{B}\dot{K}^a)^0 \big\vert \right)
\\ & \quad \leq 
l_{\obs,k+1} l_{\obs,k}^{-1} A_{\mathcal{B}} \Vert \dot{K} \Vert_k^{(A),\ext}
+ l_{\obs,k+1}^2 l_{\obs,k}^{-2} \left( \frac{A_{\mathcal{B}}}{2} + 2  \left(\frac{A_{\mathcal{B}}}{2}\right)^2 \rho^*   \right) \Vert \dot{K} \Vert_k^{(A),\ext}
\\ &\qquad \leq
C \Vert \dot{K} \Vert_k^{(A),\ext}.
\end{align*}

\item $p_1=0$, $q_1=2$:
By  Lemma \ref{lemma:Estimate_for_B} and estimate \eqref{scale_transformation} we get
\begin{align*}
\left\vert D_K^2 \left(e^{F(U)}\right) (\dot{K},\dot{K}) \right\vert^{\ext}_{k+1,U,T_0}
&=l_{\obs,k+1}^2 2 \big\vert (\mathbf{B}\dot{K}^a)^0 \big\vert \big\vert (\mathbf{B}\dot{K}^b)^0 \big\vert
\\&
 \leq
 2 l_{\obs,k+1}^2 l_{\obs,k+1}^{-2} \left( \frac{A_{\mathcal{B}}}{2} \right)^2   \left( \Vert \dot{K} \Vert_k^{(A),\ext} \right)^2
 \\ &
 \leq C \left(\Vert \dot{K} \Vert_k^{(A),\ext}\right)^2.
\end{align*}

\item $p_1=0$, $q_1>2$:
The derivative is zero.

\item $p_1=1$, $q_1=0$:
By Lemma \ref{lemma:Estimate_for_A} we get
\begin{align*}
\left\vert D_H e^{F(U)} \dot{H} \right\vert^{\ext}_{k+1,U,T_0}
&= l^2_{\obs,k+1} \left\vert \int \dot{H}^a H^b \de \mu_{+} + \int H^a \dot{H}^b \de \mu_{+} \right\vert
\\ &
\leq 2 C_{FRD} l_{\obs,k+1}^2 l_{\obs,k}^{-2} h_k^{-2} \rho^* \Vert \dot{H} \Vert_{k,0}^{\ext}
\\
&\leq
C \Vert \dot{H} \Vert_{k,0}^{\ext}.
\end{align*}

\item $p_1=1$, $q_1>0$:
The derivative is zero.

\item $p_1=2$, $q_1=0$:
By Lemma \ref{lemma:Estimate_for_A} we get
\begin{align*}
\left\vert D_H^2 e^F \dot{H}^2 \right\vert^{\ext}_{k+1,U,T_0}
&= l_{\obs,k+1}^2 2 \Big\vert \int \dot{H}^a \dot{H}^b \de \mu_+ \Big\vert
\\ &
\leq 
2 C_{FRD} l_{\obs,k+1}^2 l_{\obs,k}^{-2} h_k^{-2} \left( \Vert \dot{H} \Vert_{k,0}^{\ext} \right)^2
\leq
 C \left( \Vert \dot{H} \Vert_{k,0}^{\ext} \right)^2.
\end{align*}

\item $p_1=2$, $q_1>0$:
The derivative is zero.
\end{itemize}
In summary we get
 \begin{align*}
 &\left\Vert
D_H^p D_K^q \mathbf{S}^{\ext} (H,K)(\dot{H}^p,\dot{K}^q)
\right\Vert_{k+1}^{(A),\ext}
 \\ & \quad
  \leq
 C_{p,q} \sum_{\substack{p_1 + p_2 = p\\q_1+q_2=q}}
\left(\Vert \dot{H} \Vert_{k,0}^{\ext}\right)^{p_1}
 \left(\Vert \dot{K} \Vert_{k}^{(A),\ext}\right)^{q_1}
 \left\Vert D_H^{p_2} D_K^{q_2} \mathbf{S}(H,K)(\dot{H}^{p_2},\dot{K}^{q_2}) \right\Vert_{k+1}^{(A),\ext}.
 \end{align*}



\end{proof}

Now we turn to the analysis of $\mathbf{S}$ and the proof of Lemma \ref{lemma:estimate_on_S}.

~\\
 As in \cite{ABKM}, the strategy is to write the map $\mathbf{S}^{\ext}$ as a composition of simpler maps and show smoothness for those maps. We follow closely the presentation in \cite{ABKM} and do not repeat arguments in proofs which can be applied without change to the extended setting here.

We consider the following spaces:
\begin{align*}
\mathbf{M}^{(A)} &= \left( M^{\ext}(\mathcal{P}_k^c), \Vert \cdot \Vert_k^{(A),\ext} \right),
\\
\mathbf{M}'^{(A)} &= \left( M^{\ext}(\mathcal{P}_{k+1}^c), \Vert \cdot \Vert_{k+1}^{(A),\ext} \right),
\\
\mathbf{M}_0 &= \left( M^{\ext}(\mathcal{B}_k), \Vert \cdot \Vert_{k,0}^{\ext} \right),
\\
\mathbf{M}_{|||} &= \left( M^{\ext}(\mathcal{B}_k), \vertiii{ \cdot }_k^{\ext} \right).
\end{align*}

We need a slight modification of $\mathbf{M}^{(A)}$. Define $\mathcal{P}_k^{c'} \subset \mathcal{P}_k$ as
$$
\mathcal{P}_k^{c'} =\lbrace X \in \mathcal{P}_k: \pi(X) \in \mathcal{P}_{k+1}^c \rbrace.
$$
 The space $M^{\ext}(\mathcal{P}_k^{c'})$ of functionals is defined similarly to $M^{\ext}(\mathcal{P}_k^c)$ except that $\mathcal{P}_k^c$ is replaced by $\mathcal{P}_k^{c'}$ in the definition.

A norm on $M^{\ext}(\mathcal{P}_k^{c'})$ with parameters $A,B > 1$ is given by
\begin{align*}
\Vert K \Vert_k^{(A,B),\ext} = \sup_{X \in \mathcal{P}_k^{c'}} A^{|X|_k} B^{|\mathcal{C}(X)|} \Vert K(X) \Vert_{k,X}^{\ext}.
\end{align*}
We also use the norm $\Vert \cdot \Vert^{(A,B),\ext}_{k:k+1}$ where we replace the $\Vert \cdot \Vert_{k,X}^{\ext}$ norm by the norm $\Vert \cdot \Vert_{k:k+1,X}^{\ext}$ on the right hand side.

As in \cite{ABKM}, we introduce short hand notations for the corresponding normed spaces
\begin{align*}
\widehat{\mathbf{M}}^{A,B} = \left\lbrace M(\mathcal{P}_k^{c'}), \Vert \cdot \Vert_k^{(A,B),\ext} \right\rbrace,
\quad
\widehat{\mathbf{M}}^{A,B}_{:} = \left\lbrace M(\mathcal{P}_k^{c'}), \Vert \cdot \Vert_{k:k+1}^{(A,B),\ext} \right\rbrace.
\end{align*}

 The map $\mathbf{S}$ is, as in \cite{ABKM}, rewritten in terms of the following maps. Observe the use of the subspace $\mathcal{V}_k^{(0)}$ of $\mathbf{M}_0$ here in the definition of $R_2$ in comparison to \cite{ABKM}. However, on the bulk flow part, this subspace coincides with the whole space. Another difference to \cite{ABKM} is the definition of the map $R_2$, since the second order perturbation in the observable part appears.
 \begin{align*}
 &E: \mathbf{M}_0 \rightarrow \mathbf{M}_{|||}, \quad E(H) = e^H,
 \\
 &P_1: \mathbf{M}_{|||} \times \mathbf{M}_{|||} \times \mathbf{M}_{|||} \times \widehat{\mathbf{M}}_{:}^{(A/(2 A_{\mathcal{P}}),B)} \rightarrow \mathbf{M}'^{(A)},
 \\
 &P_1(I_1,I_2,J,K)(U) = \sum_{\substack{X_1,X_2 \in \mathcal{P}_k \\ X_1 \cap X_2 = \es}} \chi(X_1 \cup X_2,U)I_1^{U \setminus (X_1 \cup X_2)}I_2^{(X_1 \cup X_2)\setminus U} J^{X_1}K(X_2)
 \\
 &P_2: \mathbf{M}_{|||} \times \mathbf{M}^{(A)} \rightarrow \mathbf{M}^{(A/2)},
 \quad P_2(I,K) = (I-1)\circ K,
 \\
 &P_3: \mathbf{M}^{(A/2)} \rightarrow \widehat{\mathbf{M}}^{(A/2,B)}, 
 \quad P_3 K (X,\varphi) = \prod_{y \in \mathcal{C}(X)} K(Y,\varphi),
 \\
& R_1: \widehat{\mathbf{M}}^{(A/2,B)} \rightarrow \widehat{\mathbf{M}}_{:}^{(A/(2A_{\mathcal{P}}),B)},
 \quad
  R_1(P) = \mathcal{R}_+ P,
 \\
 & R_2: \mathcal{V}^{(0)}_k \times \mathbf{M}^{(A)} \rightarrow \mathbf{M}_0,
 \quad
 R_2(H,K) = \mathcal{R}_+ H + st \int H^a H^b \de \mu_+  + \Pi \mathcal{R}_+ K.
 \end{align*}
Then
\begin{align*}
&\mathbf{S}(H,K) = 
\\
&P_1\left(
E(R_2(H,K)),
E(-R_2(H,K)),
1-E(R_2(H,K)),
R_1(P_3(P_2(E(H),K)))
\right).
\end{align*}

In the following we extend estimates on these maps to observables.

\subsubsection{The immersion $E$}

The following statement is an extension of Lemma 9.3 from \cite{ABKM} to observables.

\begin{lemma}
\label{lemma:smoothness_of_E}
Let $L \geq 3$.
The map
$$
E: B_{\frac{1}{8}} (0) \subset \mathbf{M}_0 \rightarrow \mathbf{M}_{|||},
\quad
E(H) = e^H,
$$
is smooth and for any $r \in \N$ there is a constant $C_r$ (which is independent or $A$) such that for all $H \in B_{\frac{1}{8}}(0)$
$$
\vertiii{ D^r E(H) (\dot{H}_1, \ldots, \dot{H}_r) }_k^{\ext}
= \vertiii{ e^H \dot{H}_1 \ldots \dot{H}_r }_k^{\ext}
 \leq C_r \Vert \dot{H}_1 \Vert_{k,0}^{\ext} \cdots \Vert \dot{H}_r \Vert_{k,0}^{\ext}.
$$
Moreover, for all $H \in B_{\frac{1}{8}}(0)$,
$$
\vertiii{ e^H - 1 }_{k}^{\ext} \leq 8 \Vert H \Vert_{k,0}^{\ext}.
$$
\end{lemma}

\begin{proof}

The difference to \cite{ABKM} is that $H \in \mathbf{M}_0$ is of the following form:
$$
H = H^{\es}
+ s \left( \lambda^a + \sum_i n_i^a \nabla_i \varphi(a) \right) \1_a
+ t \left( \lambda^b + \sum_i n_i^b \nabla_i \varphi(b) \right) \1_b
+ st q^{ab}.
$$
In Lemma \ref{lemma:relevant_variable} it is shown that for the extended relevant variable $H \in \mathbf{M}_0$
\begin{align*}
|H|_{k,B,T_{\varphi}}^{\ext} \leq 2(1 + |\varphi |^2_{k,B}) \Vert H \Vert_{k,0}^{\ext}.
\end{align*}
This is the only ingredient for the proof where the observables play a role; for $\Vert H \Vert_{k,0}^{\ext} \leq \frac{1}{8}$ the remaining proof follows as in \cite{ABKM}.

\end{proof}

\subsubsection{The map $P_2$}

We extend Lemma 9.4 from \cite{ABKM} to the setting with observables. Here, $h_0(L)$ is fixed in Lemma \ref{lemma:Properties_of_weights}.

\begin{lemma}
\label{lemma:P2}
Let $L \geq 2^{d+3} + 16 R$ and $h \geq h_0(L)$.
Consider the map
$$
P_2: \mathbf{M}_{|||} \times \mathbf{M}^{(A)}
\rightarrow \mathbf{M}^{(A/2)},
\quad
P_2(I,K) = (I-1) \circ K.
$$
Restricted to $B_{\rho_1}(1) \times B_{\rho_2} (0)$ with $\rho_1 < (2A)^{-1}$ and $\rho_2 < \frac{1}{2}$, the map $P_2$ is smooth for any $A \geq 2$ and satisfies
\begin{align*}
\frac{1}{j_1!j_2!} 
&
\Vert (D_I^{j_1} D_K^{j_2} P_2)(I,K)(\dot{I}, \ldots\dot{I},\dot{K}, \ldots,\dot{K}) \Vert_k^{(A/2),\ext}
\\ &
\qquad \qquad 
\leq
\left( 2A \vertiii{ \dot{I} }_k^{\ext}\right)^{j_1} \left( 2 \Vert \dot{K} \Vert_k^{(A),\ext} \right)^{j_2}.
\end{align*}
This implies in particular for $I \in B_{\rho_1}(1)$ and $K \in B_{\rho_2}(0)$ that
\begin{align*}
\Vert P_2(I,K) \Vert_k^{(A/2),\ext} \leq 2A \vertiii{ I-1 }_k^{\ext} + 2 \Vert K \Vert_k^{(A),\ext}.
\end{align*}
\end{lemma}

\begin{proof}
Ingredients here are the norm estimates in Lemma \ref{lemma:norm _estimate} which also hold for the extended norms. Thus the claim follows as in \cite{ABKM}.
\end{proof}

\subsubsection{The map $P_3$}

The following lemma is based on Lemma 9.5 in \cite{ABKM} and extended to observables. Here, $h_0(L)$ is fixed in Lemma \ref{lemma:Properties_of_weights}.

\begin{lemma}
\label{lemma:P3}
Assume $L \geq 2^{d+3} + 16R$ and $h \geq h_0(L)$.
Let $A \geq 2$ and $B \geq 1$. Consider the map 
$$
P_3: \mathbf{M}^{(A/2)} \rightarrow \widehat{\mathbf{M}}^{(A/2,B)},
\quad
P_3K(X) = \prod_{Y \in \mathcal{C}(X)} K(Y).
$$ 
Its restriction to $B_{\rho} (0)$ is smooth for any $\rho$ such that $\rho \leq (2B)^{-1}$ and it satisfies the following bound for $j \geq 0$,
$$
\frac{1}{j!}
\left\Vert (D^j P_3 K)(\dot{K}, \ldots, \dot{K})
\right\Vert_k^{(A/2,B),\ext} \leq\left( 2B \Vert \dot{K} \Vert_{k,r}^{(A/2),\ext} \right)^j.
$$
\end{lemma}

\begin{proof}
The proof follows as in \cite{ABKM} by using 1. from Lemma \ref{lemma:norm _estimate}.
\end{proof}

\subsubsection{The map $R_2$}

The following statement is an extension of Lemma 9.8 in \cite{ABKM}. The estimates look different from those in \cite{ABKM} due to the second order perturbation in the observable flow. 

\begin{lemma}
\label{lemma:R2}
Assume $L \geq 2^{d+3} + 16 R$.
Consider
$$
R_2: \mathcal{V}_k^{(0)} \times \mathbf{M}^{(A)} \rightarrow \mathbf{M}_0,
\quad
R_2(H,K)
= \mathcal{R}_+ H + st \int H^a H^b \de\mu_+ + \Pi \mathcal{R}_+ K.
$$

For any $h \geq 1$ and $A \geq 1$ the map $R_2$ is smooth and there is a constant $C$ which is independent of $A$ such that
\begin{align*}
&\Vert D^{j_1}_H D^{j_2}_K R_2(H,K)(\dot{H}, \ldots, \dot{H}, \dot{K}, \ldots, \dot{K}) \Vert_{k,0}^{\ext}
\\& \qquad
\leq C 
\begin{cases}
\Vert H \Vert_{k,0}^{\ext} + \Vert H^a \Vert^a_{k,0} \Vert H^b \Vert^b_{k,0} + \Vert K \Vert_k^{(A),\ext}
& \text{ if } j_1 = j_2 = 0 \\
\left(\Vert \dot{H} \Vert_{k,0}^{\ext} + \Vert \dot{H}^a \Vert^a_{k,0} \Vert H^b \Vert^b_{k,0} + \Vert H^a \Vert^a_{k,0} \Vert \dot{H}^b \Vert^b_{k,0} \right)
& \text{ if } j_1 = 1, j_2 = 0 \\
\Vert \dot{K} \Vert_k^{(A),\ext}
& \text{ if } j_1 = 0, j_2 = 1 \\
\Vert \dot{H}^a \Vert^a_{k,0} \Vert \dot{H}^b \Vert^b_{k,0}
& \text{ if } j_1 = 2, j_2 = 0
\end{cases}
\end{align*}
and $D^{j_1}_H D^{j_2}_k R_2(H,K)(\dot{H}, \ldots, \dot{H}, \dot{K}, \ldots, \dot{K}) = 0$ else.
\end{lemma}

\begin{proof}
The extended norm consists of the following terms:
\begin{align*}
&\Vert R_2(H,K) \Vert_{k,0}^{\ext}
= \sum_{\alpha \in \lbrace \es, a, b, ab \rbrace} \left\Vert \left(R_2(H,K) \right)^{\alpha} \right\Vert_{k,0}^{\alpha}
\\ & 
= \Vert \mathcal{R}_+ H^{\es} \Vert_{k,0} + \Vert H^a \Vert^a_{k,0} + \Vert H^b \Vert^b_{k,0} + \left\Vert \int H^a H^b \de\mu_+ \right\Vert_{k,0}^{ab}
+ \sum_{\alpha \in \lbrace \es,a,b,ab \rbrace} \left\Vert \Pi^{\alpha} \mathcal{R}_+ K^{\alpha} \right\Vert_{k,0}^{\alpha}.
\end{align*}
The first four terms can be estimated, using Lemma \ref{lemma:Estimate_for_A}, as follows:
\begin{align*}
&\left\Vert \mathcal{R}_+ H^{\es} \right\Vert_{k,0} + \Vert H^a \Vert^a_{k,0} + \Vert H^b \Vert^b_{k,0} + \left\Vert \int H^a H^b \de\mu_+ \right\Vert_{k,0}^{ab}
\\ & \qquad \qquad \qquad \qquad
\leq C \Vert H \Vert_{k,0}^{\ext} + C_{FRD} h^{-1} \Vert H^a \Vert^a_{k,0} \Vert H^b \Vert^b_{k,0}.
\end{align*}
Derivatives with respect to $H$ are bounded similarly since
$$
\left[D_H R_2(H,K) \dot{H}\right]^{\obs} = s \dot{H}^a + t \dot{H}^b + st \left( \int \dot{H}^a H^b \de\mu_+ + \int H^a \dot{H}^b \de \mu_+ \right)
$$
and
$$
\left[D_H^2 R_2(H,K) (\dot{H})^2\right]^{\obs}
=  2 st \int \dot{H}^a \dot{H}^b \de \mu_+. 
$$

It remains to show that, for $\alpha \in \lbrace a,b,ab \rbrace$,
$$
\Vert \Pi^{\alpha} \mathcal{R}_+ K^{\alpha} \Vert_{k,0}^{\alpha}
\leq C \Vert K \Vert_k^{(A)}.
$$
To show this inequality, we use Lemma \ref{Boundedness_of_Projection} to obtain
\begin{align*}
\Vert \Pi^{\alpha} \mathcal{R}_+ K^{\alpha} \Vert_{k,0}^{\alpha} \leq C |\mathcal{R}_+ K|_{k,B,T_0}^{\ext}.
\end{align*}
For the extended seminorm it holds as in \cite{ABKM} that
$$
\Vert F(B) \Vert_{k:k+1,B}^{\ext} = \sup_{\varphi} w_{k:k+1}^{-B}(\varphi)|F(B)|_{k,B,T_{\varphi}}^{\ext}
\geq |F(B)|_{k,B,T_0}^{\ext}.
$$
Thus
$$
\Vert \Pi^{\alpha} \mathcal{R}_+ K^{\alpha} \Vert_{k,0}^{\alpha}
\leq C \Vert \mathcal{R}_+ K(B) \Vert_{k:k+1,B}^{\ext}.
$$
Now we can proceed as in \cite{ABKM}, using Lemma \ref{lemma:Regularity_of_Integration_Map}.

Due to the linearity with respect to $K$ the bounds for the derivatives with respect to $K$ follow from the case without derivatives.

\end{proof}

\subsubsection{The map $R_1$}

We extend Lemma 9.7 from \cite{ABKM} to our setting. 

\begin{lemma}
\label{lemma:R1}
Assume $L \geq 2^{d+3} + 16 R$.
Consider the map
$$
R_1: \widehat{\mathbf{M}}^{(A/2,B)} \rightarrow \widehat{\mathbf{M}}_{:}^{(A/(2A_{\mathcal{P}}),B)},
 \quad
  R_1(P) = \mathcal{R}_+ P.
$$
For $B \geq 1$ and any $A \geq 4 A_{\mathcal{P}}$ the map $R_1$ is smooth and satisfies
\begin{align*}
\Vert D_P^j R_1(P)(\dot{P},\ldots,\dot{P}) \Vert_{k:k+1}^{(A/(2A_{\mathcal{P}}),B),\ext}
\leq 
 \left( \Vert \dot{P} \Vert_k^{(A/2),\ext}\right)^j \left( \Vert P \Vert_k^{(A/2),\ext} \right)^{1-j}
\end{align*}
for $j \in \lbrace 0,1 \rbrace$. The derivatives vanish for $j > 1$.
\end{lemma}

\begin{proof}
The statement for $j=0$ follows directly from Lemma \ref{lemma:Regularity_of_Integration_Map}. Note that the map $R_1$ is linear in $P$ so that the statement for $j>0$ is trivial.
\end{proof}

\subsubsection{The map $P_1$}

In the following we extend Lemma 9.6 from \cite{ABKM} to observables. Here, $h_0(L)$ is fixed in Lemma \ref{lemma:Properties_of_weights}.

\begin{lemma}
\label{lemma:P1}
Assume $L \geq \max \left\lbrace 2^{d+3} + 16R,4d(2^d + R) \right\rbrace$, and $h \geq h_0(L)$.
Consider the map
\begin{align*}
& P_1:
\mathbf{M}_{|||} \times \mathbf{M}_{|||} \times \mathbf{M}_{|||} \times \widehat{\mathbf{M}}_{:}^{(A/(2 A_{\mathcal{P}}),B)} \rightarrow \mathbf{M}'^{(A)},
 \\
 &P_1(I_1,I_2,J,K)(U) = \sum_{\substack{X_1,X_2 \in \mathcal{P}_k \\ X_1 \cap X_2 = \es}} \chi(X_1 \cup X_2,U)I_1^{U \setminus (X_1 \cup X_2)}I_2^{(X_1 \cup X_2)\setminus U} J^{X_1}K(X_2).
\end{align*}
Let $A_0(L,d) = (48 A_{\mathcal{P}})^{\frac{L^d}{\alpha}}$ with $\alpha =(1 + 2^d)^{-1}(1 + 6^d)^{-1}$.
If $A \geq A_0, B = A$ and if $\rho_1$, $\rho_2$, $\rho_3$ satisfy
\begin{align*}
\rho_1 \leq \frac{1}{2},
\quad
\rho_2 \leq A^{-2},
\quad
\rho_3 \leq 1,
\end{align*}
then the map $P_1$ restricted to $U= B_{\rho_1}(1) \times B_{\rho_1}(1) \times B_{\rho_2}(0) \times B_{\rho_3}(0)$ is smooth and satisfies
\begin{align*}
&\frac{1}{i_1!i_2!j_1!j_2!}
\\
&\left\Vert
D_{I_1}^{i_1} D_{I_2}^{i_2} D_{J}^{j_1} D_{K}^{j_2}
P_1(I_1,I_2,J,K)
(\dot{I_1}, \ldots, \dot{I_1}, \dot{I_2},\ldots,\dot{I_2}, \dot{J},\ldots,\dot{J}, \dot{K}, \ldots, \dot{K})
\right\Vert_{k+1,r}^{(A),\ext}
\\
&\leq \frac{\eta^2}{4} L^d
\left( \vertiii{ \dot{I}_1 }^{\ext} \right)^{i_1}
\left( \vertiii{ \dot{I}_2 }^{\ext} \right)^{i_2}
\left( A^2 \vertiii{ \dot{J} }^{\ext} \right)^{j_1}
\left( \left\Vert \dot{K} \right\Vert_{k:k+1}^{(A/(2A_{\mathcal{P}}),B),\ext} \right)^{j_2}.
\end{align*}
\end{lemma}

\begin{proof}
The difference to \cite{ABKM} is the additional factor $\frac{\eta^2}{4} L^d$ here which appears in Lemma \ref{lemma:norm _estimate}. Apart from that the proof is the same as in \cite{ABKM}.
\end{proof}

\begin{remark}\label{P_1_improved}
Consider the case of the bulk flow, i.e., set $s=t=0$. When inspecting the proof of Lemma 9.6 in \cite{ABKM}, we get
\begin{align*}
A^{|U|_{k+1}}
&
\left\Vert
D_{I_1} D_{I_2} D_{J} D_{K}
P_1 (I_1,I_2,J,K)(U)(\dot{I_1},\dot{I_2},\dot{J}, \dot{K})
\right\Vert
\\ &
\leq
A^{-x|U|_{k+1}}A^2 
\vertiii{ \dot{I}_1 }
\vertiii{ \dot{I}_2 }
\vertiii{ \dot{J} }
\left\Vert \dot{K} \right\Vert_{k:k+1}^{(A/(2A_{\mathcal{P}}),B) }
\end{align*}
for $x \in (0,2\alpha)$. Namely, we have that
\begin{align*}
A^{|U|_{k+1}}
&
\left\Vert
D_{I_1} D_{I_2} D_{J} D_{K}
P_1 (I_1,I_2,J,K)(U)(\dot{I_1},D_{I_2},D_{J}, D_{K})
\right\Vert
\\ &
\leq
\left(
\frac{(48 A_{\mathcal{P}})^{2L^d}}{A^{2 \alpha}} 
\right)^{|U|_{k+1}}A^2 
\vertiii{ \dot{I}_1 }
\vertiii{ \dot{I}_2 }
\vertiii{ \dot{J} }
\left\Vert \dot{K} \right\Vert_{k:k+1}^{(A/(2A_{\mathcal{P}}),B) }
\\ &
\leq A^{-x|U|_{k+1}}A^2 
\vertiii{ \dot{I}_1 }
\vertiii{ \dot{I}_2 }
\vertiii{ \dot{J} }
\left\Vert \dot{K} \right\Vert_{k:k+1}^{(A/(2A_{\mathcal{P}}),B) }
\end{align*}
if we choose
$$
A \geq \left(48 A_{\mathcal{P}} \right)^{\frac{2 L^d}{2 \alpha - x}}.
$$
\end{remark}

\subsubsection{Proof of Lemma \ref{lemma:estimate_on_S}}

For the sake of completeness we review the proof as it is done in \cite{ABKM}.

\begin{proof}[Proof of Lemma \ref{lemma:estimate_on_S}]
The assertion follows from the smoothness of the individual maps $E,P_1,P_2,P_3,R_1$ and $R_2$ and the chain rule.

Let $A_0$ be as in Lemma \ref{lemma:P1} and set $B=A$.
By Lemma \ref{lemma:P1} there exists a neighbourhood
$$
O_1 = B_{\rho_1}(1) \times B_{\rho_1}(1) \times B_{\rho_2}(0) \times B_{\rho_3}(0)
$$
such that $P_1$ is smooth in $O_1$.
By Lemma \ref{lemma:smoothness_of_E} there is a neighbourhood
$$
O_2 = B_{\rho_4}(0) \subset B_{\frac{1}{8}}(0)
$$
such that $E$ is smooth in $O_2$ and $E(O_2) \subset B_{\rho_1}(1)$ and $1-E(O_2) \subset B_{\rho_2}(0)$.
By Lemma \ref{lemma:R2} there is a neighbourhood
$$
O_3 = B_{\rho_5}(0) \times B_{\rho_6}(0)
$$
such that $R_2$ is smooth in $O_3$ and $R_2(O_3) \subset O_2$.
This defines the first restriction on $U_{\rho^*}$, namely
$$
U_{\rho^*} \subset B_{\rho_5}(0) \times B_{\rho_6}(0)
$$
The second restriction comes from the condition
$$
R_1 \left( P_3 \left( P_2 \left( E(H),K \right) \right) \right) \in B_{\rho_3}(0).
$$
By Lemma \ref{lemma:R1} there is a neighbourhood
$$
O_4 = B_{\rho_7}(0)
$$
such that $R_1$ is smooth in $O_4$ and $R_1(O_4) \subset B_{\rho_3}(0)$.
By Lemma \ref{lemma:P3} there is a neighbourhood
$$
O_5 \subset B_{\rho}(0)
$$
such that $P_3$ is smooth in $O_5$ and $P_3(O_5) \subset O_4$.
By Lemma \ref{lemma:P2} there is a neighbourhood
$$
O_6 = B_{\rho_8}(1) \times B_{\rho_9}(0)
$$
such that $P_2$ is smooth in $O_6$ and $P_2(O_6) \subset O_5$.
Finally, by Lemma \ref{lemma:smoothness_of_E} there is a neighbourhood
$$
O_7 = B_{\rho_{10}}(0) \subset B_{\rho_4}(0)
$$
such that $E(O_7) \subset B_{\rho_8}(1)$.
We obtain the second restriction:
$$
U_{\rho_*} \subset B_{\rho_{10}}(0) \times B_{\rho_9}(0).
$$
The combination of both constraints yields that $\mathbf{S}$ is $C^{\infty}$ in the set
$$
U_{\rho_*} \subset B_{\rho_{10} \wedge \rho_5}(0) \times B_{\rho_9 \wedge \rho_6}(0).
$$

The chain rule implies the bounds on the derivatives.
\end{proof}

\begin{remark} \label{S_improved}
Remark \ref{P_1_improved} and chain rule implies that in the case of the bulk flow there is a constant $C_1$ such that for any $x \in (0,2 \alpha)$ and $(H,K) \in U_{\rho}$
\begin{align*}
&
A^{|U|_{k+1}}
\left\Vert
D_H D_K D_q \mathbf{S}_k (H,K,q) (\dot{H},\dot{K},\dot{q})(U)
\right\Vert_{k+1,U}
\\ &
\leq
 C_1 A^{-x |U|_{k+1}} A^4
\Vert \dot{H} \Vert_{k,0}
\Vert \dot{K} \Vert_k^{(A)}
\Vert \dot{q} \Vert,
\end{align*}
where the factors $A$ come from the estimates on $D_J P_1$, $DP_3$, and $D_I P_2$.
\end{remark}


\subsection{Derivatives of the extended renormalisation map at $(0,0)$}

In this section we prove the bounds on $\mathbf{C}$ stated in Proposition \ref{Prop:Contractivity_Observables}, the bounds on $\mathbf{B}$ stated in Proposition \ref{Prop:Bounds_on_B_ObservableFlow}, a bound on the second order part in $\mathbf{A}$ as used in the proof of Lemma \ref{lemma:estimate_on_prefactor_F}, and we compute the $ab$-part of the second derivative of $\mathbf{S}^{\ext}$ at $(0,0)$ as stated in Proposition \ref{Prop:2nd_order_perturbation_effect_on_S}.

\subsubsection{Bound on the extended operator $\mathbf{C}$} \label{subsec:Bound_C}

Let $K \in M^{\ext}(\mathcal{P}_k^c)$, $U \in \mathcal{P}_k^c$, and $\varphi \in \chi_N$. Then $\mathbf{C}K$ can be decomposed into two parts,
\begin{align}
\mathbf{C}K(U,\varphi) = F(U,\varphi) + G(U,\varphi).
\end{align}
The large-polymer part $F \in M^{\ext}(\mathcal{P}^c_{k+1})$ is defined by
\begin{align*}
F(U,\varphi) = \sum_{\substack{X \in \mathcal{P}_k^c \setminus \mathcal{B}_k \\ \pi(X) = U}} \mathcal{R}_+ K(X,\varphi),
\end{align*}
and $G$ satisfies $G(U,\varphi) = 0$ for all $U \in \mathcal{P}^c_{k+1} \setminus \mathcal{B}_{k+1}$, otherwise, for $U = B_+ \in\mathcal{B}_{k+1}$,
\begin{align*}
G(B_+,\varphi) = \sum_{B \in \mathcal{B}_k(B_+)} G(B,\varphi)
\quad
\text{with}
\quad
G(B,\varphi) = (1-\Pi) \mathcal{R}_+ K(B,\varphi).
\end{align*}
We restate the key bound from Proposition \ref{Prop:Contractivity_Observables} as Lemma \ref{lemma:contractivity_of_C} below. 

\begin{lemma}\label{lemma:contractivity_of_C}
For any $\theta \in (0,1)$ there exists an $L_0$ such that for all odd integers $L \geq L_0$ there is $A_0$ and $h_0$ with the following property. For all $A \geq A_0$ and for all $h \geq h_0$,
\begin{align*}
\Vert \mathbf{C} \Vert^{(A),\ext}_+ \leq \theta
\end{align*}
independently of $k$ and $N$.
\end{lemma}

The proof is very similar to the proof in \cite{ABKM}. For the argument of the large-polymer part $F$ we have to deal with the additional factor $\frac{\eta^2}{4} L^d$ arising in the transformation of scales from the factor $\frac{l_{\obs,k+1}^{|\alpha|}}{l_{\obs,k}^{|\alpha|}}$, see  2. in Lemma \ref{lemma:submultiplicativity2}.

The following lemma extends Lemma 10.2. from \cite{ABKM} to observables.

\begin{lemma}
Let $L \geq 2^{d+3} + 16 R$. There is $A_0$ such that for all $A \geq A_0$
\begin{align*}
\Vert F \Vert_{k+1}^{(A),\ext} \leq \frac{\theta}{2} \Vert K \Vert_k^{(A),\ext}.
\end{align*}
\end{lemma}

\begin{proof}
Lemma \ref{lemma:submultiplicativity2} states that for $U = \pi(X)$
\begin{align*}
\Big\vert \mathcal{R}_+ K(X,\varphi) \Big\vert_{k+1,U,T_{\varphi}}^{\ext}
\leq \frac{\eta^2}{4} L^d \Big\vert \mathcal{R}_+ K(X,\varphi) \Big\vert_{k,X,T_{\varphi}}^{\ext}.
\end{align*}
By Lemma \ref{lemma:Properties_of_weights} it follows that
\begin{align*}
w_{k:k+1}^X(\varphi) \leq w_{k+1}^U(\varphi).
\end{align*}
We conclude that
\begin{align*}
\Vert \mathcal{R}_+ K(X,\varphi) \Vert_{k+1,U}^{\ext}
\leq \frac{\eta^2}{4} L^d \Vert \mathcal{R}_+ K(X,\varphi) \Vert_{k:k+1,X}^{\ext}.
\end{align*}
By this inequality we can estimate
\begin{align}
& A^{|U|_{k+1}} \Vert F(U) \Vert_{k+1,U}^{\ext}
\nonumber
\\ & 
\leq A^{|U|_{k+1}} \frac{\eta^2}{4} L^d
\left(
\sum_{\substack{X \in \mathcal{P}_k^c \setminus \mathcal{S}_k \\ \pi(X) = U}}
\Vert \mathcal{R}_+ K(X) \Vert_{k:k+1,X}^{\ext}
+
\sum_{\substack{X \in \mathcal{P}_k^c \setminus \mathcal{S}_k \\ \pi(X) = U}}
\Vert \mathcal{R}_+ K(X) \Vert_{k:k+1,X}^{\ext}
\right).
\label{eq:bla}
\end{align}
We bound the two summands in \eqref{eq:bla} seperately. The first term can be estimated similar to \cite{ABKM}, with a change in the choice of $A$:
\begin{align*}
&
A^{|U|_{k+1}} \frac{\eta^2}{4} L^d \sum_{\substack{X \in \mathcal{P}_k^c \setminus \mathcal{S}_k \\ \pi(X) = U}}
\Vert \mathcal{R}_+ K(X) \Vert_{k:k+1,X}^{\ext}
\\
& \quad\quad\quad\quad\quad\quad\quad
\leq
\Vert K \Vert_k^{(A),\ext} \frac{\eta^2}{4} L^d \sum_{\substack{X \in \mathcal{P}_k^c \setminus \mathcal{S}_k \\ \bar{X} = U}} \left( A_{\mathcal{P}} A^{-\frac{2\alpha}{1+2\alpha}} \right)^{|X|_k},
\end{align*}
where $\alpha = \left[ (1 + 2^d)(1 + 6^d) \right]^{-1}$.
Let
\begin{align*}
A \geq \left( \frac{A_{\mathcal{P}}}{\bar{\delta}} \frac{4}{\theta} \frac{\eta^2}{4} L^d \right)^{\frac{1+2\alpha}{2\alpha}}
\end{align*}
where $\bar{\delta}$ is the constant from Lemma C.2 in \cite{ABKM}.
Then
\begin{align*}
\sum_{\substack{X \in \mathcal{P}_k^c \setminus \mathcal{S}_k \\ \pi(X) = U}}
\Vert \mathcal{R}_+ K(X) \Vert_{k:k+1,X}^{\ext}
\leq
\frac{\theta}{4} \Vert K \Vert_k^{(A),\ext}.
\end{align*}
For a bound on the second contribution in \eqref{eq:bla} we again follow closely the proof from \cite{ABKM}, with a change in the choice of $A$. For $U \in \mathcal{B}_{k+1}$ we have
\begin{align*}
A^{|U|_{k+1}} \frac{\eta^2}{4} L^d \sum_{\substack{X \in \mathcal{P}_k^c \setminus \mathcal{S}_k \\ \pi(X) = U}}
\Vert \mathcal{R}_+ K(X) \Vert_{k:k+1,X}^{\ext}
\leq
A \Vert K \Vert_k^{(A),\ext} L^d (2^{d+1} + 1)^{d 2^d} \frac{A_{\mathcal{P}}^2}{A^2} \frac{\eta^2}{4} L^d.
\end{align*}
If
\begin{align*}
A \geq \frac{4}{\theta} A_{\mathcal{P}}^2 L^d (2^{d+1} + 1)^{d 2^d} \frac{\eta^2}{4} L^d,
\end{align*}
then
\begin{align*}
A^{|U|_{k+1}} \frac{\eta^2}{4} L^d \sum_{\substack{X \in \mathcal{P}_k^c \setminus \mathcal{S}_k \\ \pi(X) = U}}
\Vert \mathcal{R}_+ K(X) \Vert_{k:k+1,X}^{\ext}
\leq
\frac{\theta}{4} \Vert K \Vert_k^{(A),\ext}.
\end{align*}
For $A$ large enough this finishes the claim.
\end{proof}

Next we consider the contribution from single blocks. We extend Lemma 10.4 from \cite{ABKM} to observables.

\begin{lemma}
There is $L_0$ such that for all $L \geq L_0$, $h \geq h_0(L)$
and for all $A \geq 1$
\begin{align*}
\Vert G \Vert_{k+1}^{(A),\ext} \leq \frac{\theta}{2} \Vert K \Vert_k^{(A),\ext}.
\end{align*}
\end{lemma}

\begin{proof}
Remember that $G(U) = 0$ for $U \notin \mathcal{B}_{k+1}$ and
$$
G(B_+) = \sum_{B \in \mathcal{B}_k(B_+)} G(B) = \sum_{B \in \mathcal{B}_k(B_+)} (1-\Pi) \mathcal{R}_+ K(B)
$$
for $B_+ \in \mathcal{B}_{k+1}$. Thus
\begin{align*}
\Vert G \Vert_{k+1}^{(A),\ext}
&\leq
A \sup_{\varphi} w_{k+1}^{-B'}(\varphi) \sum_{B \in \mathcal{B}_k(B_+)} \vert G(B)\vert_{k+1,B,T_{\varphi}}^{\ext}
\\&
\leq
A \sup_{\varphi} w_{k+1}^{-B'}(\varphi) \sum_{B \in \mathcal{B}_k(B_+)} \sum_{\alpha \in \lbrace \es,a,b,ab \rbrace} \1_{\alpha \in B} l_{\obs,k+1}^{|\alpha|} \vert G^{\alpha}(B)\vert_{k+1,B,T_{\varphi}}.
\end{align*}
Fix $\alpha \in \lbrace a,b,ab\rbrace$. We use the second inequality in Lemma \ref{lemma:submultiplicativity1} to get
\begin{align*}
\vert G^{\alpha}(B) \vert_{k+1,B,T_{\varphi}}
& \leq
\left( 1 + |\varphi|_{k+1,B} \right)^3
\Big(
\vert (1-\Pi_k^{\alpha}) \mathcal{R}_+ K^{\alpha}(B) \vert_{k+1,B,T_0}
\\ & \qquad \qquad
+ 16 L^{-\frac{3}{2}d} \sup_{0 \leq t \leq 1} \vert (1-\Pi_k^{\alpha}) \mathcal{R}_+ K^{\alpha}(B) \vert_{k,B,T_{t\varphi}}
\Big).
\end{align*}
By Lemma \ref{lemma:Contraction_estimate_1} we proceed the estimate as follows
\begin{align*}
\vert G^{\alpha}(B) \vert_{k+1,B,T_{\varphi}}
& \leq 
\left( 1 + |\varphi|_{k+1,B} \right)^3
\Big(
C L^{-(d/2 + A(\alpha,k))} \vert \mathcal{R}_+ K^{\alpha} \vert_{k,B,T_0}
\\
& \qquad \qquad
+ 16 L^{-\frac{3}{2}d} \sup_{0 \leq t \leq 1} \vert (1-\Pi_k^{\alpha}) \mathcal{R}_+ K^{\alpha}(B) \vert_{k,B,T_{t\varphi}}
\Big).
\end{align*}
We continue as in \cite{ABKM} with the estimates
\begin{align*}
\vert \mathcal{R}_+ K^{\alpha}(B) \vert_{k,B,T_0}
& \leq  l_{\obs,k}^{-|\alpha|}  A_{\mathcal{B}} \Vert K \Vert_{k,B},
\\
\vert \Pi_k^{\alpha} \mathcal{R}_+ K^{\alpha}(B) \vert_{k,B,T_{t\varphi}}
& \leq C (1 + |\varphi|_{k,B})^2 A_{\mathcal{B}} l_{\obs,k}^{-|\alpha|} \Vert K \Vert_{k,B},
\quad \text{and}
\\
\vert \mathcal{R}_+ K^{\alpha}(B) \vert_{k,B,T_{t\varphi}}
& \leq A_{\mathcal{B}} w_{k:k+1}^B(\varphi) l_{\obs,k}^{-|\alpha|} \Vert K(B) \Vert_{k,B},
\end{align*}
where we have the additional factor $l_{\obs,k}^{-|\alpha|}$ on the right hand sides in contrast to \cite{ABKM}. We obtain
\begin{align*}
&\vert G^{\alpha}(B) \vert_{k+1,B,T_{\varphi}}
\\
& \qquad\qquad
 \leq l_{\obs,k}^{-|\alpha|}
\left( 1 + |\varphi|_{k+1,B} \right)^3
\left(
C L^{-(d/2 + A(\alpha,k))} A_{\mathcal{B}} \Vert K \Vert_{k,B}
\right.
\\
& \qquad\qquad \quad \quad \left.
+ 16 L^{- \frac{3}{2} d} A_{\mathcal{B}} w_{k:k+1}^B(\varphi) \Vert K \Vert_{k,B}
+ 16 L^{- \frac{3}{2} d} C (1 + |\varphi|_{k,B})^2 A_{\mathcal{B}} \Vert K \Vert_{k,B}
\right)
\\
& \qquad\qquad
 \leq A_{\mathcal{B}} C l_{\obs,k}^{-|\alpha|}
\left( 1 + |\varphi|_{k+1,B} \right)^5
\Vert K \Vert_{k,B}
\left(
L^{-(d/2 + A(\alpha,k))} + L^{-\frac{3}{2}d} w_{k:k+1}^B(\varphi)
\right)
\\
& \qquad\qquad \leq
C' l_{\obs,k}^{-|\alpha|} w_{k+1}^{B'}(\varphi) \Vert K \Vert_{k,B}
\left(
L^{-(d/2 + A(\alpha,k))} + L^{-\frac{3}{2}d}
\right).
\end{align*}
For $\alpha = \es$ we use the result from \cite{ABKM}, namely that
$$
\vert G^{\es}(B) \vert_{k+1,B,T_{\varphi}}
 \leq
C' w_{k+1}^{B'}(\varphi) \Vert K \Vert_{k,B}
\left(
L^{-d'} + L^{-\frac{3}{2}d}
\right)
$$
with $d' = \frac{d}{2} + \lfloor d/2 \rfloor + 1 > d$.

Let $d'(\alpha,k) = d'$ for $\alpha = \es$ and $d'(\alpha,k) = d/2 + A(\alpha,k)$ else. We combine the estimates obtained so far and obtain
\begin{align*}
\Vert G \Vert_{k+1}^{(A),\ext}
& \leq
C'  \sum_{\alpha \in \lbrace \es,a,b,ab \rbrace} \sum_{B \in \mathcal{B}_k(B_+)} \1_{\alpha \in B} l_{\obs,k+1}^{|\alpha|} l_{\obs,k}^{-|\alpha|} A^{|B|_k} \Vert K \Vert_{k,B}
\left(
L^{-d'(\alpha,k)} + L^{-\frac{3}{2}d}
\right).
\end{align*}
In the case $\alpha = \es$, the sum over all $B \in \mathcal{B}_k(B_+)$ gives an additional factor $L^d$. In contrast, for $\alpha \in \lbrace a,b,ab \rbrace$, the sum reduces to one term so this factor does not arise. However, we have
$$
\left( \frac{l_{\obs,k+1}}{l_{\obs,k}} \right)^{|\alpha|}
=
\begin{cases}
(2 \eta)^{|\alpha|}			&\text{if } \alpha \in \lbrace a,b,ab \rbrace, \, k \geq j_{ab},
\\
\left(\frac{\eta}{2} L^{d/2} \right)^{|\alpha|}			&\text{if } \alpha \in \lbrace a,b \rbrace, \, k < j_{ab}
\end{cases}
$$
which is canceled by $L^{-d'(\alpha,k)}$.
In summary we thus get
\begin{align*}
\Vert G \Vert_{k+1}^{(A),\ext}
& \leq
C \Vert K \Vert_k^{(A),\ext}
\left(
L^{d-d'} + L^{-\frac{1}{2}d} + L^{-1} + L^{-d} + L^{-\frac{d}{2}} + L^{-\frac{3}{2}d}
\right).
\end{align*}
Now choose $L$ large enough such that
$$
\Vert G \Vert_{k+1}^{(A),\ext}
\leq \frac{\theta}{2} \Vert K \Vert_k^{(A),\ext}.
$$

\end{proof}

\subsubsection{Bounds on the extended operator $\mathbf{B}$}

Here we prove Proposition \ref{Prop:Bounds_on_B_ObservableFlow}. We restate the result in the following lemma.

\begin{lemma}\label{lemma:Estimate_for_B}
For $\alpha \in \lbrace a,b \rbrace$, with the constant $A_{\mathcal{B}}$ from Lemma \ref{lemma:Properties_of_weights} which is independent of $L$, the following estimates hold:
\begin{align*}
\big\vert(\mathbf{B} K_k^{\alpha})^1 \big\vert
& \leq l_k^{-1} l_{\obs,k}^{-1} \frac{A_{\mathcal{B}}}{2} \Vert K_k \Vert_k^{(A),\ext},
\\
\big\vert(\mathbf{B} K_k^{\alpha})^0 \big\vert
& \leq l_{\obs,k}^{-1} \frac{A_{\mathcal{B}}}{2} \Vert K_k \Vert_k^{(A),\ext},
\\
\big\vert \mathbf{B} K_k^{ab} \big\vert
& \leq l_{\obs,k}^{-2} \frac{A_{\mathcal{B}}}{2} \Vert K_k \Vert_k^{(A),\ext}.
\end{align*}
\end{lemma}

\begin{proof}
The proof is similar to the one of Lemma \ref{Boundedness_of_Projection}. First, by Lemma \ref{lemma:existence_of_projection},
\begin{align*}
\big\vert(\mathbf{B} K_k^{\alpha})^1 \big\vert
= \big\vert \langle \mathcal{R}_+ K_l^{\alpha},b^{\alpha} \rangle_0 \big\vert
\leq |b^{\alpha}|_{k,B} \big\vert \mathcal{R}_+ K_k^{\alpha}(B) \big\vert_{k,B,T_0}
 \leq l_k^{-1} l_{\obs,k}^{-1} \frac{A_{\mathcal{B}}}{2} \Vert K_k \Vert_k^{(A),\ext}.
\end{align*}
Furthermore,
\begin{align*}
\big\vert(\mathbf{B} K_k^{\alpha})^0 \big\vert
 \leq \int \big\vert K_k^{\alpha}(B,\xi) \big\vert \mu_{k+1}(\de\xi)
\leq l_{\obs,k}^{-1} \frac{A_{\mathcal{B}}}{2} \Vert K_k \Vert_k^{(A),\ext}
\end{align*}
and similarly,
\begin{align*}
\big\vert \mathbf{B} K_k^{ab} \big\vert
 \leq \int \big\vert K_k^{ab}(B,\xi) \big\vert \mu_{k+1}(\de\xi)
\leq l_{\obs,k}^{-2} \frac{A_{\mathcal{B}}}{2} \Vert K_k \Vert_k^{(A),\ext}.
\end{align*}
\end{proof}

\subsubsection{Bound on the extended operator $\mathbf{A}$}

\begin{lemma}\label{lemma:Estimate_for_A}
Let $H^a = n^a \nabla \varphi(b)$, $H^b = n^b \nabla \varphi(b)$, $k \geq j_{ab}$.
Then
$$
\Big\vert \int H^a H^b \de \mu_{k+1} \Big\vert
\leq C_{FRD} l_{\obs,k}^{-2} h_k^{-2} \Vert H^a \Vert^a_{k,0} \Vert H^b \Vert_{k,0}^b.
$$
\end{lemma}

\begin{proof}
Note that
 $$
\int \nabla \varphi(a) \nabla \varphi(b) \mu_{k+1}(\de\varphi) = \nabla^* \nabla C_{k+1} (a,b)
$$
and 
$$
|n^a| \leq l_{\obs,k}^{-1}l_k^{-1} \Vert H^a \Vert^a_{k,0}.
$$ 
By the properties of the finite-range decomposition the proof follows straightforwardly.
\end{proof}

\subsubsection{Second derivative of $\mathbf{S}^{\ext}$ at $(0,0)$}

Here we prove Proposition \ref{Prop:2nd_order_perturbation_effect_on_S}. We restate the result in the following lemma.

\begin{lemma}\label{Lemma:2nd_order_perturbation_exact}
The $st$-part of the second derivative in direction $H$ of $\mathbf{S}^{\ext}$ is zero:
$$
\left[ D^2_H \mathbf{S}^{\ext}(0,0) (\dot{H},\dot{H}) \right]^{ab} = 0.
$$
\end{lemma}

\begin{proof}
Note that
\begin{align*}
D_H^2 \mathbf{S}^{\ext} (0,0) (\dot{H},\dot{H})
= D_H^2 \mathbf{S}(0,0)(\dot{H},\dot{H})
\end{align*}
since $\mathbf{S}(0,0)=0$ and
\begin{align*}
D_H \left( e^{-s\left(\mathbf{B}K^a\right)^0 - t\left(\mathbf{B}K^b\right)^0 - st\left(\int H^a H^b \de \mu_+ +\mathbf{B}K^{ab}\right)} \right) \Big\vert_{H=K=0} \dot{H}
=0.
\end{align*}
By the product rule we get a sum of the following three terms:
\begin{align*}
 D_H^2 \mathbf{S}^{\ext} & (0,0) (\dot{H},\dot{H})
\\
& =
2 \sum_{X \in \mathcal{P}_k} \chi(X,U)
D_H \left( \left( e^{\tilde{H}}\right)^{U \setminus X} \right) \dot{H} \Big\vert_{H=K=0} \times
\\
& \qquad\qquad\qquad\qquad\qquad\qquad
\int D_H \left( \left( e^H - e^{\tilde{H}} \right)^X \right) \dot{H} \Big\vert_{H=K=0} \de\mu_+
\\
& \qquad +
2 \sum_{X \in \mathcal{P}_k} \chi(X,U)
D_H \left( \left( e^{\tilde{H}}\right)^{-X \setminus U} \right) \dot{H} \Big\vert_{H=K=0} \times
\\
& \qquad\qquad\qquad\qquad\qquad\qquad
\int D_H \left( \left( e^H - e^{\tilde{H}} \right)^X \right) \dot{H} \Big\vert_{H=K=0} \de\mu_+
\\
& \qquad +
\sum_{X \in \mathcal{P}_k} \chi(X,U)
\int D_H^2 \left( \left( e^H - e^{\tilde{H}} \right)^X \right) (\dot{H},\dot{H}) \de\mu_+.
\end{align*}
Let us consider the second term in the right hand side above.
We compute
\begin{align*}
D_H \left( \left( e^H - e^{\tilde{H}} \right)^X \right) \dot{H} \big\vert_{H=K=0}
= \1_{X=B} \left( \dot{H}(B) - D_H \tilde{H}(B) \dot{H} \big\vert_{H=K=0} \right).
\end{align*}
The constraint $X=B$ for any $B\in \mathcal{B}_k$ implies that $X \setminus U = \es$ for any $U$ satisfying $\chi(X,U) \neq 0$.
Thus the second term is zero.

~\\
The $ab$-part of the first term is zero as well. We compute
\begin{align*}
D_H \left( \left(e^{\tilde{H}} \right)^{U \setminus X} \right) \dot{H} \Big\vert_{H=K=0}
= \sum_{B \in \mathcal{B}_k( U \setminus X)} \left( \tilde{A} \dot{H}^{\es} + s \dot{H}^a + t \dot{H}^b \right)(B)
\end{align*}
and
\begin{align*}
\int D_H & \left( \left( e^H - e^{\tilde{H}} \right)^X \right) \dot{H} \Big\vert_{H=K=0} \de\mu_+
\\
&\qquad \qquad=
\1_{X=B} \int
\dot{H}^{\es}(B,\varphi + \xi) + s \dot{H}^a (B,\varphi + \xi) + t \dot{H}^b (B,\varphi + \xi) 
\\
&\qquad \qquad \qquad\qquad
- \tilde{A} \dot{H}^{\es}(B,\varphi) - s \dot{H}^a (B,\varphi ) - t \dot{H}^b (B,\varphi)
\de \mu_+
\\
&\qquad \qquad = 
\1_{X=B} \int
\dot{H}^{\es}(B,\varphi + \xi)- \tilde{A} \dot{H}^{\es}(B,\varphi) \de\mu_+.
\end{align*}
The last equality holds since 
$$\dot{H}^a(B,\varphi + \xi) = \dot{H}^a(B,\varphi) + \dot{H}^a(B,\xi)
$$ and 
$$\int \dot{H}^a(B,\xi) \de\mu_+ =~0$$ due to linearity.
Thus the first term has bulk parts and $a$- and $b$-parts, but the projection to the $ab$-part is zero.

~\\
For the third term we distinguish the case that $X=B$ for $B \in \mathcal{B}_k$ and $X=B \cup B'$ for $B,B' \in \mathcal{B}_k$, $B \neq B'$.
In the case $X=B$ we compute
\begin{align*}
&\int D_H^2 \left( \left(e^H - e^{\tilde{H}} \right)^B \right) (\dot{H},\dot{H}) \de\mu_+
\\ & \qquad \qquad\qquad
= \int \left( \dot{H}(B,\varphi + \xi) \right)^2 - 2 st \int \dot{H}^a(B) \dot{H}^b(B)\de\mu_+
\\ & \qquad \qquad\qquad\qquad
- \left( \tilde{A} \dot{H}^{\es}(B,\varphi) + s \dot{H}^a(B,\varphi) + t \dot{H}^b(B,\varphi) \right)^2 \de\mu_+
\\& \qquad \qquad\qquad
= 2 \int \dot{H}^a(B,\varphi + \xi) \dot{H}^b(B,\varphi + \xi)\de\mu_+
- 2 \int \dot{H}^a(B,\xi) \dot{H}^b(B,\xi)\de\mu_+
\\ &\qquad \qquad\qquad\qquad
-2 \int \dot{H}^a(B,\varphi) \dot{H}^b(B,\varphi)\de\mu_+
=0.
\end{align*}
In the other case we compute
\begin{align*}
&\int
D_H \left( \left(e^H - e^{\tilde{H}} \right)^B \right) \dot{H}
 D_H \left( \left(e^H - e^{\tilde{H}} \right)^{B'} \right) \dot{H}
 \de\mu_+
 \\ & \qquad\qquad
 = \int
 \left( \dot{H}(B,\varphi + \xi) - \tilde{A}\dot{H}^{\es}(B,\varphi) -s \dot{H}^a(B,\varphi) - t  \dot{H}^b(B,\varphi) \right)
 \\
 & \qquad\qquad\qquad
  \left( \dot{H}(B',\varphi + \xi) - \tilde{A}\dot{H}^{\es}(B',\varphi) -s \dot{H}^a(B',\varphi) - t  \dot{H}^b(B',\varphi) \right)
  \de\mu_+.
\end{align*}
We project this term to the $ab$-part and obtain:
\begin{align*}
\int & \left( \dot{H}^a(B,\varphi + \xi) - \dot{H}^a(B,\varphi) \right)
\left(\dot{H}^b(B',\varphi + \xi) - \dot{H}^b(B',\varphi) \right) \de\mu_+
\\ & \quad \qquad\qquad
+ \int \left( \dot{H}^b(B,\varphi + \xi) - \dot{H}^b(B,\varphi) \right)
\left(\dot{H}^a(B',\varphi + \xi) - \dot{H}^a(B',\varphi) \right)\de\mu_+
\\
&= \quad
 \int \dot{H}^a(B,\xi) \dot{H}^b(B',\xi) \de\mu_+
+ \int \dot{H}^b(B,\xi) \dot{H}^a(B',\xi) \de\mu_+.
\end{align*}

Now we distinguish the scales $k \geq j_{ab}$ and the scales $k<j_{ab}$. If $k \geq j_{ab}$, then $a,b \in B_{ab} \in \mathcal{B}_k$, and either $B=B_{ab}$ and the $B'$-term is zero, or vice versa.
If $k<j_{ab}$ only the choices $B \cup B' =B_a \cup B_b$ and $B \cup B' =B_b\cup B_a$ are relevant. Then we get
\begin{align*}
\int \dot{H}^a(B,\xi) & \dot{H}^b(B',\xi) \de\mu_+
+ \int \dot{H}^b(B,\xi) \dot{H}^a(B',\xi) \de\mu_+
\\ & \quad
=
2 n_an_b \int \nabla \varphi(a) \nabla \varphi(b) \de\mu_{k+1}
= 2 n_a n_b \nabla^* \nabla C_{k+1}(ab).
\end{align*}
Due to the definition of the scale $j_{ab}$ and the finite-range property of the covariances we have
$$
\nabla^* \nabla C_{k+1}(a,b) = 0 \text{ for all } k<j_{ab}.
$$
This finishes the claim.
\end{proof}

\bibliography{meinbib}
\bibliographystyle{alpha}

\end{document}